\documentclass[twoside]{article}

\usepackage{color}
\usepackage{amssymb}
\usepackage{dsfont}
\usepackage{amssymb,amsmath,ulem,cancel}

\numberwithin{equation}{section}

\newtheorem{theorem}{Theorem}[section]
\newtheorem{lem}{Lemma}[section]
\newtheorem{pro}{Proposition}[section]
\newtheorem{cor}{Corollary}[section]
\newtheorem{rem}{Remark}[section]
\newtheorem{rems}{Remarks}[section]
\newtheorem{ex}{Example}[section]
\newtheorem{defi}{Definition}[section]
\newtheorem{hyp}{Assumption}[section]
\newtheorem{con}{Conjecture}[section]

\newcommand{\bt}{\begin{theorem}}
\newcommand{\et}{\end{theorem}}
\newcommand{\bl}{\begin{lem}}
\newcommand{\el}{\end{lem}}
\newcommand{\bp}{\begin{pro}}
\newcommand{\ep}{\end{pro}}
\newcommand{\bcor}{\begin{cor}}
\newcommand{\ecor}{\end{cor}}
\newcommand{\bcon }{\begin{con} \rm }
\newcommand{\econ }{\end{con}}
\newcommand{\lab }{\label }
\newcommand{\bd}{\begin{defi} \rm }
\newcommand{\ed}{\end{defi}}
\newcommand{\brem }{\begin{rem} \rm }
\newcommand{\erem }{\end{rem}}
\newcommand{\brems }{\begin{rems} \rm }
\newcommand{\erems }{\end{rems}}
\newcommand{\bhyp }{\begin{hyp} \rm }
\newcommand{\ehyp }{\end{hyp}}
\newcommand{\bex}{\begin{ex} \rm }
\newcommand{\eex}{\end{ex}}
\newcommand{\be}{\begin{equation}}
\newcommand{\ee}{\end{equation}}
\newcommand{\bde}{\begin{displaymath}}
\newcommand{\ede}{\end{displaymath}}
\newcommand{\beq}{\begin{eqnarray*}}
\newcommand{\eeq}{\end{eqnarray*}}
\newcommand{\beqa}{\begin{eqnarray}}
\newcommand{\eeqa}{\end{eqnarray}}
\newcommand{\bea}{\begin{align*}}
\newcommand{\eea}{\end{align*}}

\def\proof{\noindent {\it Proof. $\, $}}
\def\endproof{\hfill $\Box$ \vskip 5 pt}

\def\I{\mathbf{1}}
\def\wh{\widehat}
\def\wt{\widetilde}
\def\tilde{\widetilde}
\def\phi{\varphi }

\newcommand{\VCc}{V^C}

\newcommand{\norm}{|\!\!|\!|\!\!|}
\newcommand{\Sst}{S^{s,t}}
\newcommand{\aass}{\mbox{\rm a.s.}}
\newcommand{\aaee}{\mbox{\rm a.e.}}

\newcommand{\zzb}{\bar z^i_t}
\newcommand{\ssx}{s}
\newcommand{\ssy}{s}

\newcommand{\Leb }{\ell }
\newcommand{\cac}{\alpha }

\newcommand{\zzhi}{{\widehat z}^i_t}

\newcommand{\pA}{A}
\newcommand{\pC}{C}

\newcommand{\pCc}{c}

\newcommand{\etab}{\eta^b}
\newcommand{\etal}{\eta^l}

\newcommand{\Blr}{B^{l}}
\newcommand{\Bbr}{B^{b}}

\newcommand{\Bilr}{B^{i,l}}
\newcommand{\Bibr}{B^{i,b}}

\newcommand{\rll}{r^{l}}
\newcommand{\rbb}{r^{b}}

\newcommand{\ribb}{r^{i,b}}

\newcommand{\Vnet}{{V}^{\textrm{net}}}
\newcommand{\whVnet}{\widehat{V}^{\textrm{net}}}

\newcommand{\Vnettl}{{\widetilde{V}}^{l,\textrm{net}}}
\newcommand{\Vnettb}{{\widetilde{V}}^{b,\textrm{net}}}

\def\Xhat{\widehat{x}}
\def\Xtil{\widetilde{x}}
\def\Fhat{\widehat{\phi}}
\def\Ftil{\widetilde{\phi}}

\def\t1{\tau_{(1)}}

\def\rr{\mathbb R}
\def\ff{{\mathbb F}}

\def\gg{{\mathbb G}}

\def\G{{\cal G}}

\def\VLL{V^0}
\def\P{\mathbb P}

\def\PT{\wt {\mathbb P}}
\def\PTb{\wt {\mathbb P}^{\beta}}

\def\EP{{\mathbb E}_{\mathbb P}}

\newcommand{\wHlamo}{\wh{\mathcal{H}}_{\lambda}^{2}}
\newcommand{\wHlamd}{\wh{\mathcal{H}}_{\lambda}^{2,d}}
\newcommand{\wHzero}{\wh{\mathcal{H}}_{0}^{2}}
\newcommand{\wHzerd}{\wh{\mathcal{H}}_{0}^{2,d}}

\newcommand{\sumik}{\textstyle{\sum}}

\newcommand{\Keywords}[1]{\par\noindent{\small{\bf Keywords\/}: #1}}
\newcommand{\Class}[1]{\par\noindent{\small{\bf Mathematics Subjects Classification (2010)\/}: #1}}

\title{{\Large \bf FAIR AND PROFITABLE BILATERAL PRICES UNDER \\ \vskip 1 pt FUNDING COSTS AND COLLATERALIZATION} \vskip 85 pt }

\author{Tianyang Nie and Marek Rutkowski\footnote{The research of Tianyang Nie and Marek Rutkowski
was supported under Australian Research Council's Discovery Projects funding scheme (DP120100895).}
\\ School of Mathematics and Statistics \\ University of Sydney
\\ Sydney, NSW 2006, Australia}

\date{\vskip 60 pt 1 December 2014 \vskip 60 pt}

\begin{document}

\maketitle

\begin{abstract}
Bielecki and Rutkowski \cite{BR-2014} introduced and studied a generic nonlinear market model, which includes several risky assets, multiple funding accounts and margin accounts. In this paper, we examine the pricing and hedging of contract both from the perspective of the hedger and the counterparty with arbitrary initial endowments. We derive inequalities for unilateral prices and we give the range for either fair bilateral prices or bilaterally profitable prices. We also study the monotonicity of a unilateral price with respect to the initial endowment. Our study hinges on results for BSDE driven by continuous martingales obtained
in \cite{NR3}, but we also derive the pricing PDEs for path-independent contingent claims of European style in a Markovian framework.
\vskip 20 pt
\Keywords{hedging, fair prices, funding costs, margin agreement, BSDE, PDE}
\vskip 20 pt
\Class{91G40,$\,$60J28}
\end{abstract}

%%%%%%%%%%%%%%%%%%%%%%%%%%%%%%%%%%%%%%%%%%%%%%%%%%%%%%%%%%%%%%%%%%%%%%%%%%%

%\newpage
%\tableofcontents
\newpage

%%%%%%%%%%%%%%%%%%%%%%%%%%%%%%%%%%%%%%%%%%%%%%%%%%%%%%%%%%%%%%%%%%%%%%%%%%%
%%%%%%%%%%%%%%%%%%%%%%%%%%%%%%%%%%%%%%%%%%%%%%%%%%%%%%%%%%%%%%%%%%%%%%%%%%%
\section{Introduction} \label{sect1}
%%%%%%%%%%%%%%%%%%%%%%%%%%%%%%%%%%%%%%%%%%%%%%%%%%%%%%%%%%%%%%%%%%%%%%%%%%%
%%%%%%%%%%%%%%%%%%%%%%%%%%%%%%%%%%%%%%%%%%%%%%%%%%%%%%%%%%%%%%%%%%%%%%%%%%%

In Bielecki and Rutkowski \cite{BR-2014}, the authors introduced a generic nonlinear trading model for bilateral collateralized contracts, which includes several risky assets, multiple funding accounts, as well as the margin account. For related recent studies by other authors, see also \cite{BCPP11,BK09,BK11,SC12a,SC12b,PPB12,P10}. Using a suitable version of the no-arbitrage argument, they first discussed the hedger's fair price for a contract in the market model without collateralization (see Section 3.2 in \cite{BR-2014}). Subsequently, for a collateralized contract that can be replicated, they defined the hedger's ex-dividend price (see Section 5 in \cite{BR-2014}). It was also shown in \cite{BR-2014} that the theory of backward stochastic differential equations (BSDEs) is an important tool to compute the ex-dividend price (see, e.g., Propositions 5.2 and 5.4 in \cite{BR-2014}). It is worth mentioning that all the pricing and hedging arguments in \cite{BR-2014} are given from the viewpoint of the hedger and no attempt was made there to derive no-arbitrage bounds for unilateral prices and to examine the existence of fair bilateral prices.
In the present work, we consistently examine the issue of pricing and hedging of an OTC derivative contract from the perspective of the hedger and his counterparty. Since we work within a nonlinear trading set-up, where the nonlinearity may stem from the different cash interest rates, funding costs for risky assets and collateralization, the hedger's and counterparty's price do not necessarily coincide. Therefore, our goal is to compare the hedger's and counterparty's prices and to derive the range for no-arbitrage prices. In the case of different lending and borrowing rates, which is a relatively simple instance of a nonlinear market model, the no-arbitrage price of any contingent claims must belong to an arbitrage band with the lower (resp., upper) bound given by the counterparty's (resp., the hedger's) price of the contract (see Bergman \cite{B-1995}). In a recent paper by Mercurio \cite{M-2013}, the author extended the results from \cite{B-1995} by examining the pricing of European options in a model with different lending and borrowing interest rates and under collateralization.

As emphasized in \cite{BR-2014}, in the nonlinear setup (for instance, in a market model with different borrowing and lending interest rates), the initial endowments of the hedger and the counterparty play an important role in pricing considerations.
Unlike in the classic options pricing model, which enjoys the linearity of the no-arbitrage pricing rule, it is no longer true
that it suffices to consider the case where the initial endowments are null. This is due to the fact that, for instance, the hedger's ex-dividend price may depend on his initial endowment, in general  (see Proposition 5.2 in \cite{BR-2014}).
Note in this regard that the results established in \cite{B-1995} and \cite{M-2013} only cover the case when the initial endowments of the hedger and the counterparty are null. In this paper, one of our main goals is to examine how the initial endowment of each party affects his unilateral price. For the sake of concreteness, we consider the model with partial netting and collateralization which was introduced in \cite{BR-2014}. A similar analysis was also done for the model previously studied by Bergman \cite{B-1995}, but with non-zero initial endowments of counterparties (see Nie and Rutkowski \cite{NR4}). It is clear that the method developed in these papers can be applied to other set-ups.

This work is organized as follows.  In Section \ref{sect2}, we give a brief overview the set-up studied in
\cite{BR-2014} and we describe the main model considered in the foregoing sections, dubbed the market model with partial netting.
In Section \ref{sect3}, we present definitions of no-arbitrage and fair prices, as introduced in \cite{BR-2014}. Some preliminary results from  \cite{BR-2014} are extended to the case of a collateralized contract with an exogenous margin account and we introduce and discuss the concepts of {\it fair bilateral prices} and {\it bilaterally profitable prices}.
In Section \ref{sect4}, we show that the pricing and hedging problems for both parties can be represented
by solutions of some BSDEs and we establish the existence and uniqueness results for these BSDEs. Although
the BSDEs  are well known to be a convenient tool to deal with prices and hedging strategies (see, e.g.,
\cite{SC12a,SC12b,EPQ-1997}), we stress that the BSDEs studied in this work are formally derived using no-arbitrage
arguments under a judiciously chosen martingale measure, whereas in some other papers on funding costs the existence of
a `risk-neutral probability' is postulated a priori, rather than formally justified.
Section \ref{sect5}, which is the main part of this work, deals with the properties of unilateral prices.
Under alternative assumptions on initial endowments of both parties, we establish several inequalities for unilateral prices,
which in turn allow us to obtain the ranges for fair bilateral prices or bilaterally profitable prices. We also examine the monotonicity of prices with respect to initial endowment and we present the PDE approach within a Markovian framework.
Lengthy proofs of some results are gathered in the appendix.

%In Section \ref{sect6}, we revisit the model with different lending and borrowing rates previously studied by Bergman %\cite{B-1995}. We complement his study by applying the BSDEs approach to unilateral prices for an arbitrary initial endowment %(recall that it was implicitly assumed in \cite{B-1995} that the initial endowment
%of the hedger is null).

%In Section \ref{sect7}, we collect auxiliary results on solutions to BSDEs driven be either
%a one-dimensional or a multi-dimensional continuous martingale. The development in this section hinges on related results
%from papers by  El Karoui and Huang \cite{ELH-1997} and Carbone et al. \cite{CFS-2008}.

\newpage

%%%%%%%%%%%%%%%%%%%%%%%%%%%%%%%%%%%%%%%%%%%%%%%%%%%%%%%%%%%%%%%%%%%%%%%%%%%%%%%%%%%%%%%%%
%%%%%%%%%%%%%%%%%%%%%%%%%%%%%%%%%%%%%%%%%%%%%%%%%%%%%%%%%%%%%%%%%%%%%%%%%%%%%%%%%%%%%%%%%
\section{Trading under Funding Costs and Collateralization}   \label{sect2}
%%%%%%%%%%%%%%%%%%%%%%%%%%%%%%%%%%%%%%%%%%%%%%%%%%%%%%%%%%%%%%%%%%%%%%%%%%%%%%%%%%%%%%%%%
%%%%%%%%%%%%%%%%%%%%%%%%%%%%%%%%%%%%%%%%%%%%%%%%%%%%%%%%%%%%%%%%%%%%%%%%%%%%%%%%%%%%%%%%%

Let us first recall the following setting of \cite{BR-2014} for the market models.
Throughout the paper, we fix a finite trading horizon date $T>0$ for our model of the financial market.
Let $(\Omega, \G, \gg , \P)$ be a filtered probability space satisfying the usual conditions of right-continuity and completeness, where the filtration $\gg = (\G_t)_{t \in [0,T]}$ models the flow of information available to all traders. For convenience, we assume that the initial $\sigma$-field ${\cal G}_0$ is trivial. Moreover, all processes introduced in what follows are implicitly assumed to be $\gg$-adapted  and any semimartingale is assumed to be c\`adl\`ag.

\noindent {\bf Risky assets.} For $i=1,2, \dots, d$, we denote by $S^i$ the {\it ex-dividend price} of the $i$th risky asset with the {\it cumulative dividend stream} $\pA^i$. $S^i$ is aimed to represent the price of any traded security, such as, stock, stock option, interest rate swap, currency option, cross-currency swap, CDS, CDO, etc.

\noindent {\bf Cash accounts.} The riskless  {\it lending} (resp., {\it borrowing}) {\it cash account} $B^l$ (resp., $B^b$) is used for unsecured lending (resp., borrowing) of cash. When the borrowing and lending cash rates are equal, we denote the {\it cash account} simply by $B^{0}$.

\noindent {\bf Funding accounts.}  We denote by $\Bilr$ (resp., $\Bibr$) the {\it lending} (resp., {\it borrowing}) {\it funding account} associated with the $i$th risky asset. In case when borrowing and lending rates are equal, we simply denote by $B^i$ the {\it funding account} for the $i$th risky asset. Unless explicitly stated otherwise, we work under the assumption:
long and short funding rates for each risky asset $S^i$ are identical, that is, $\Bilr=\Bibr=B^i$ for
$i=1, 2, \dots, d$.

\bhyp \lab{assumption for primary assets}
The price processes of {\it primary assets} are assumed to satisfy: \hfill \break
(i) For each $i=1,2,\dots , d$, the price $S^i$ is semimartingale and the cumulative dividend stream $\pA^i$ is finite variation process with $\pA^i_{0}=0$.\hfill \break
(ii) The riskless account $B^{l}$, $B^{b}$ and $B^{i}$ are strictly positive and continuous finite variation processes with $B^l_0=B^b_0=B^{i}_0=1$, for $i=1,2\dots , d$.
\ehyp

For a {\it bilateral financial contract}, or simply a {\it contract}, we mean an arbitrary c\`adl\`ag process $\pA$ of finite variation. The process $A$ is aimed to represent the {\it cumulative cash flows} of a given contract from time 0 till its maturity date $T$. By convention, we set $\pA_{0-}=0$.

The process $\pA$ is assumed to model all cash flows of a given contract, which are either paid out from the wealth or added to the wealth, as seen from the perspective of the {\it hedger} (recall that the other party is referred to as the {\it counterparty}).
Note that the process $A$ includes the initial cash flow $A_0$ of a contract at its inception date $t_0=0$.
For instance, if a contract has the initial {\it price} $p$ and stipulates that the hedger will receive cash flows  $\bar{\pA}_1,
\bar{\pA}_2, \dots , \bar{\pA}_k$ at times $t_1, t_2, \dots , t_k \in (0,T]$, then we set $A_0=p$ so that
\bde
\pA_t = p + \sum_{l=1}^k \I_{[t_l,T]}(t) \bar{\pA}_l .
\ede
The symbol $p$ is frequently used to emphasize that all future cash flows $\bar{\pA}_l$ for $l=1,2, \dots, k$ are explicitly specified by the contract's covenants, but the initial cash flow $A_0$ is yet to be formally defined and evaluated.
%Obviously, if the assumption that  $\sigma$-field ${\cal G}_0$ is trivial is relaxed, then $A_0$ becomes a ${\cal G}_0$-measurable %random variable, rather than a real number, as will be the case hereafter.
Valuation of a contract $A$ means, in particular, searching for the range of {\it fair values} $p$ at time $0$ from the viewpoint of either the hedger or the counterparty. Although the valuation paradigm will be the same for the two parties, due either to the asymmetry in their trading costs and opportunities, or the non-linearity of the wealth dynamics, they will typically obtain different sets of fair prices for~$A$. This is the main objective of our current work.

%%%%%%%%%%%%%%%%%%%%%%%%%%%%%%%%%%%%%%%%%%%%%%%%%%%%%%%%%%%%%%
\subsection{Collateralization} \label{sect2.1}
%%%%%%%%%%%%%%%%%%%%%%%%%%%%%%%%%%%%%%%%%%%%%%%%%%%%%%%%%%%%%%

In this paper, we examine the situation when the hedger and the counterparty enter a contract and either receive or post collateral with the value formally represented by a stochastic process~$\pC $, which is assumed to be a semimartingale (or, at least, a c\`adl\`ag process). The process $C$ is called the {\it margin account} or the {\it collateral amount}. Let
\be \lab{collss}
\pC_t = \pC_t \I_{\{ \pC_t \geq 0\}} +  \pC_t \I_{\{ \pC_t < 0\}} = \pC^+_t - \pC^-_t.
\ee
By convention, $\pC^+_t$ is the cash value of collateral received at time $t$ by the hedger, whereas $\pC^-_t$ represents the cash value of collateral posted by him. For simplicity of presentation, it is postulated throughout that only cash collateral may be delivered or received (for other conventions, see \cite{BR-2014}). We also make the following natural assumption regarding the
state of the margin account at the contract's maturity date.

\bhyp \lab{assumption for margin account on maturity date}
The $\gg$-adapted collateral amount process $C$ satisfies $C_T=0$.
\ehyp

The equality $C_T=0$ is a convenient way of ensuring that any collateral amount posted is returned in full to its owner when a contract matures, provided that the default event does not occur at $T$. Of course, if the default event is also modeled,
then one needs to specify the closeout payoff.

Let us first make some comments from the perspective of the hedger regarding the crucial features of the margin accounts.
The current financial practice typically requires the collateral amounts to be held in {\it segregated} margin accounts,
so that the hedger, when he is a collateral taker, cannot make use of the collateral amount for trading.
Another collateral convention encountered  in practice is {\it rehypothecation}, which refers to the situation where a bank is allowed to reuse the collateral pledged by its counterparties as collateral for its own borrowing.
Note that if the hedger is a collateral giver, then a particular convention regarding segregation or rehypothecation is immaterial for the wealth dynamics of his portfolio.

We are in a position to introduce trading strategies based on a finite family of primary assets.

\brem \label{how to discuss counterparty}
For simplicity, we discuss from the point view of hedger, unless explicitly stated. The similar discussions hold for the counterparty by changing $(A,C)$ to $(-A,-C)$.
\erem

\bd \lab{tsx2}
A {\it collateralized hedger's trading strategy} is a quadruplet $(x,\phi , \pA , \pC )$ where a portfolio $\phi $, given by
\be \lab{vty}
\phi = \big( \xi^1,\dots , \xi^{d},\psi^{l},\psi^{b},\psi^1,\dots ,\psi^{d+1}, \etab, \etal,\eta^{d+2} \big)
\ee
is composed of the {\it risky assets} $S^i,\, i=1,2,\ldots,d$, the {\it unsecured lending cash account} $B^l$ the {\it unsecured borrowing cash account} $B^b$, the {\it funding accounts} $B^i,\, i=1,2,\ldots,d$, the {\it borrowing account} $B^{d+1}$  for the posted cash collateral, the {\it collateral accounts} $B^{\pCc,b}$ and $B^{\pCc,l}$, and the {\it lending account} $B^{d+2}$ associated with the received cash collateral.
\ed

If $B^{\pCc,b}\neq B^{\pCc,l}$,  for example if the hedger post the collateral, he will receives interest from the counterparty determined by $B^{\pCc,l}$, i.e., the counterparty pays the hedger the interest determined by $B^{\pCc,l}$ not $B^{\pCc,b}$. This creates the non-identical financial environment between the hedger and counterparty.
We make the following standing assumption.

\bhyp \lab{assumption for collateral account}
The collateral accounts $B^{c,l}$, $B^{c,b}$, $B^{d+1}$, $B^{d+2}$  are strictly positive, continuous processes
of finite variation with $B^{c,l}_0=B^{c,b}_0=B^{d+1}_0=B^{d+2}_0=1$.
\ehyp

\brem \label{remark for cash collateral}
The {\it cash collateral} is described by the following postulates: \hfill \break
(i) If the hedger receives at time $t$ the amount $\pC^+_t$ as cash collateral, then he pays
to the counterparty interest determined by the amount $\pC^+_t$ and the account $B^{\pCc,b}$.
Under segregation, he receives interest determined by the amount $\pC^+_t$ and the account $B^{d+2}$ and thus $\eta^{d+2}_t B^{d+2}_t = C^+_t$.  When rehypothecation is considered,  the hedger may temporarily (that is, before the contract's maturity date or the default time, whichever comes first) utilize the cash amount $\pC^+_t$ for his trading purposes, then $\eta^{d+2}= 0$.
\hfill \break (ii) If the hedger posts a cash collateral at time $t$, then the collateral amount is borrowed from
 the dedicated {\it collateral borrowing account} $B^{d+1}$. He receives interest determined by the amount $\pC^-_t$ and the collateral account $B^{\pCc,l}$. We postulate that
\be \lab{565656}
\psi^{d+1}_t B^{d+1}_t = - C^-_t .
\ee
\erem

%%%%%%%%%%%%%%%%%%%%%%%%%%%%%%%%%%%%%%%%%%%%%%%%%%%%%%%%%%%%%%
\subsection{Self-Financing Trading Strategies} \label{sect2.2}
%%%%%%%%%%%%%%%%%%%%%%%%%%%%%%%%%%%%%%%%%%%%%%%%%%%%%%%%%%%%%%

In the context of a collateralized contract, we find it convenient to introduce the following three processes:
\hfill \break (i) the process $V_t(x,\phi , \pA ,\pC)$ representing the hedger's wealth at time $t$,
\hfill \break (ii) the process $V_t^p(x,\phi , \pA ,\pC)$ representing the value of hedger's portfolio at time $t$,
\hfill \break (iii) the {\it adjustment process} $\VCc_t(x,\phi , \pA ,\pC) := V_t(x,\phi , \pA ,\pC) - V^p_t(x,\phi , \pA ,\pC)$, which is aimed to quantify the impact of the margin account on trading strategy.

\bd \lab{ts2x}
The hedger's {\it portfolio's value} $V^p(x,\phi , \pA ,\pC)$ is given by
\be \lab{poutf1}
V^p_t (x,\phi , \pA ,\pC) =  \sum_{i=1}^{d} \xi^i_t S^i_t + \sum_{j=1}^{d+1} \psi^j_t B^j_t+\psi^l_t B^l_t+\psi^b_t B^b_t.
\ee
The hedger's {\it wealth} $V(x,\phi , \pA ,\pC)$ equals
\be \lab{portf1}
V_t (x,\phi , \pA ,\pC) =  \sum_{i=1}^{d} \xi^i_t S^i_t + \sum_{j=1}^{d+1} \psi^j_t B^j_t+\psi^l_t B^l_t+\psi^b_t B^b_t
+  \etab_t B^{\pCc,b}_t+ \etal_t B^{\pCc,l}_t +  \eta^{d+2}_t B^{d+2}_t  .
\ee
\ed

It is clear that the adjustment process $\VCc(x,\phi , \pA ,\pC)$ equals
\be \lab{portf1b}
 \VCc_t (x,\phi , \pA ,\pC) = \etab_t B^{\pCc,b}_t+ \etal_t B^{\pCc,l}_t +  \eta^{d+2}_t B^{d+2}_t
 = - C_t +  \eta^{d+2}_t B^{d+2}_t
\ee
where $\etab_t =-  (B^{\pCc,b}_t)^{-1}\pC_t^+$ and $\etal_t =(B^{\pCc,l}_t)^{-1} \pC_t^-$.

The self-financing property of the hedger's strategy is defined in terms of the dynamics of the value process.
Note that we use here the process $V^p(x,\phi , \pA ,\pC)$, and not $V(x,\phi , \pA ,\pC)$,
to emphasize the role of $V^p(x,\phi , \pA ,\pC)$ as the value of the hedger's portfolio of traded assets.
Observe also that the equality  $V^p(x,\phi , \pA ,\pC) = V(x,\phi , \pA ,\pC)$ holds when the process $\pC$ vanishes, that is, $C=0$, since then $\eta^{d+2}=0$ as well. Let $x$ stand for the {\it initial endowment} of the hedger.

\bd \label{definition for self financing}
A collateralized hedger's trading strategy $(x,\phi , \pA , \pC )$ with $\phi $ given by \eqref{vty} is {\it self-financing} whenever
the {\it portfolio's value} $V^p(x,\phi , \pA ,\pC)$, which is given by \eqref{poutf1}, satisfies, for every $t \in [0,T]$,
\begin{align} \lab{porfx}
V^p_t (x,\phi , \pA ,\pC)  = \, \, &x + \sum_{i=1}^{d} \int_{(0,t]} \xi^i_u \, d(S^i_u + \pA^i_u )
+  \sum_{j=1}^{d+1} \int_0^t \psi^j_u \, dB^j_u+\int_0^t \psi^l_u \, dB^l_u+\int_0^t \psi^b_u \, dB^b_u + \pA_t \\
&+ \int_0^t \etab_u\, dB^{\pCc,b}_u + \int_0^t \etal_u \, dB^{\pCc,l}_u + \int_0^t \eta^{d+2}_u\, dB^{d+2}_u
  - \VCc_t (x,\phi , \pA ,\pC). \nonumber
\end{align}
\ed

The terms $\int_0^t \etab_u\, dB^{\pCc,b}_u$, $ \int_0^t \etal_u \, dB^{\pCc,l}_u$ and $\int_0^t \eta^{d+2}_u\, dB^{d+2}_u$ represent the accrued interest generated by the margin account. The first two processes are given uniquely in terms
of $C$ since  $\etab_t =-  (B^{\pCc,b}_t)^{-1}\pC_t^+$ and $\etal_t =(B^{\pCc,l}_t)^{-1} \pC_t^-$, whereas the last one depends on the collateral convention.

%%%%%%%%%%%%%%%%%%%%%%%%%%%%%%%%%%%%%%%%%%%%%%%%%%%%%%%%%%%%%%%%%%%%%%%%%%%%%%%%%%%%
\subsection{Market Model with Partial Netting} \label{sect2.3}
%%%%%%%%%%%%%%%%%%%%%%%%%%%%%%%%%%%%%%%%%%%%%%%%%%%%%%%%%%%%%%%%%%%%%%%%%%%%%%%%%%%%

In this section, we consider a specific model with partial netting and collateralization with rehypothecation, which was previously studied in \cite{BR-2014}. Besides postulating that the accounts $\Blr$ and $\Bbr$ may differ, we also allow for the inequality $B^{i,l}\ne B^{i,b},\, i=1,2 \dots , d$ to hold, in general. We also make the following simplifying assumption.

\bhyp \lab{assumption for collateral borrowing account}
The collateral borrowing account $B^{d+1}$ coincides with $B^{b}$.
\ehyp

We follow here the offsetting/netting terminology adopted in \cite{BR-2014}.
Hence by {\it offsetting} we mean the compensation of long and short positions either for a given risky asset or for the non-risky asset. This compensation is not relevant, unless the borrowing and lending rates are different for at least one risky asset or for the cash account. By {\it netting}, we mean the aggregation of long or short cash positions across various risky assets, which share some funding accounts. Obviously, several alternative models with netting can be studied, for more details, see \cite{BR-2014}.

In this paper, we focus on the case of partial netting positions across risky assets, which means that the offsetting of long/short positions for every risky asset combined with some form of netting of long/short cash positions
for all risky assets that are funded from common funding accounts. More precisely, we postulate
that all short cash positions in risky assets $S^1, S^2, \dots , S^d$ are aggregated and invested in the
common lending account $B^{l}$, which means that we assume that $B^{i,l}=B^{l}$ for $i=1,\ldots, d$. This means
that all positive cash flows, inclusive of proceeds from short-selling of risky assets,
are transferred to the cash account $\Blr$. By contrast, long cash positions in risky assets $S^i$ are assumed to be funded
from their respective funding accounts $\Bibr$. For brevity, the trading framework described in this subsection
will be henceforth referred to as the {\it market model with partial netting.}

Accordingly, we consider a trading portfolio (note that $\eta^{d+2}=0$ in case of rehypothecation,
as was explained in Remark \ref{remark for cash collateral})
\[
\phi = \big(\xi^1,\dots ,\xi^d, \psi^{l}, \psi^{b},\psi^{1,b},\dots ,\psi^{d,b}, \etab, \etal\big)
\]
and the corresponding wealth process for the hedger
\bde  % \label{wealth process 1}
V_t(x, \phi , \pA, C ) = \psi^{l}_t \Blr_t+\psi^{b}_t\Bbr_t+  \sum_{i=1}^d ( \xi^i_tS^i_t+\psi^{i,b}_t B^{i,b}_t)+ \etab_t B^{\pCc,b}_t+ \etal_t B^{\pCc,l}_t.
\ede
It follows that,  for all $t\in[0,T]$,
\begin{equation}\label{portfolio choose}
\begin{array}
[c]{l}
\etab_t =-  (B^{\pCc,b}_t)^{-1}\pC_t^+, \quad \etal_t =(B^{\pCc,l}_t)^{-1} \pC_t^- , \quad \psi^{i,b}_t =  -(\Bibr_t)^{-1} (\xi^i_t S^i_t)^+.
\end{array}
\end{equation}
In the present set-up, the hedger's trading strategy $(x, \phi , \pA, C )$ is self-financing  whenever
the process $V^p(x,\phi , \pA ,\pC)$, which is given by
\be  \label{value of strategy 1}
V^p_t (x,\phi , \pA ,\pC) = \psi^l_t B^l_t+\psi^b_t B^b_t+\sum_{i=1}^{d} \big(\xi^i_t S^i_t + \psi^{i,b}_t B^{i,b}_t \big),
\ee
satisfies
\begin{align} \label{value of strategy 2}
V^p_t (x,\phi , \pA ,\pC)  = \, \, &x + \sum_{i=1}^{d} \int_{(0,t]} \xi^i_u \, d(S^i_u + \pA^i_u )
+  \sum_{j=1}^{d} \int_0^t \psi^{i,b}_u \, dB^{i,b}_u+\int_0^t \psi^l_u \, dB^l_u+\int_0^t \psi^b_u \, dB^b_u  \nonumber\\
&+ \int_0^t \etab_u\, dB^{\pCc,b}_u + \int_0^t \etal_u \, dB^{\pCc,l}_u
  - \VCc_t (x,\phi , \pA ,\pC) + \pA_t
\end{align}
where in turn
\bde
\VCc_t (x,\phi , \pA ,\pC) = \etab_t B^{\pCc,b}_t+ \etal_t B^{\pCc,l}_t=-C_{t}.
\ede
From equations (\ref{portfolio choose}) and (\ref{value of strategy 1}), we get
\bde  % \label{wealth process 2}
V^p_t (x,\phi , \pA ,\pC) = \psi^{l}_t \Blr_t + \psi^{b}_t \Bbr_t - \sum_{i=1}^d ( \xi^i_t S^i_t )^-.
\ede
Since we postulate that $ \psi^{l}_t\ge 0$, $\psi^{b}_t\leq0$ and $\psi^{l}_t\psi^{b}_t=0$ for all $t\in[0,T]$, we also obtain
\bde % \label{strategy for lending and borrowing}
\psi^{l}_t = (\Blr_t)^{-1} \Big( V^p_t (x,\phi , \pA ,\pC) + \sum_{i=1}^d ( \xi^i_t S^i_t )^- \Big)^+
\ede
and
\bde
\psi^{b}_t = - (\Bbr_t)^{-1} \Big( V^p_t (x,\phi , \pA ,\pC) + \sum_{i=1}^d ( \xi^i_t S^i_t )^- \Big)^-.
\ede
Finally, the self-financing condition for the trading strategy $(x,\phi , \pA ,\pC)$ can be represented as follows
\begin{align} \label{value of strategy 3}
dV^p_t (x,\phi , \pA ,\pC)  = \, \, & \sum_{i=1}^d \xi^i_t \, (dS^i_t + d\pA^i_t)
- \sum_{i=1}^d ( \xi^i_t S^i_t )^+  (\Bibr_t)^{-1} \, d\Bibr_t + d\pA_t^{C} \\
&+  \Big(V^p_t (x,\phi , \pA ,\pC) + \sum_{i=1}^d ( \xi^i_t S^i_t )^- \Big)^+ (\Blr_t)^{-1} \, d\Blr_t\nonumber\\
&-  \Big( V^p_t (x,\phi , \pA ,\pC)+ \sum_{i=1}^d ( \xi^i_t S^i_t )^- \Big)^- (\Bbr_t)^{-1}\, d\Bbr_t \nonumber
\end{align}
where $A^{C}:=A+C+F^{C}$ and, in view of Assumption \ref{additional assumption for collateral account},
\begin{align} \label{definition for FC}
F^{C}_t&:=\int_0^t \etab_u\, dB^{\pCc,b}_u + \int_0^t \etal_u \, dB^{\pCc,l}_u \nonumber \\
&=-\int_0^t \pC_u^+ (B^{\pCc,b}_u)^{-1}\, dB^{\pCc,b}_u + \int_0^t \pC_u^- (B^{\pCc,l}_u)^{-1}\, dB^{\pCc,l}_u \\
&=-\int_0^t  \pC_u (B^{\pCc}_u)^{-1} \, dB^{\pCc}_u. \nonumber
\end{align}

In general, one may consider the situation where the hedger and the counterparty are exposed to a different financial environment, which means that their respective hedging strategies for the same contract are based on different risky assets, cash accounts, funding accounts and collateral accounts. To make the analysis less cumbersome, we henceforth assume that the hedger and counterparty face exactly the same market conditions, but they may have different initial endowments. In particular, we make the following
standing assumption.

\bhyp \label{additional assumption for collateral account}
The collateral accounts $B^{c,l}$ and $B^{c,b}$ satisfy $B^{\pCc,l}=B^{\pCc,b}=B^{c}$.
\ehyp

\brem
Suppose that Assumption \ref{additional assumption for collateral account} is not postulated, so that the accounts
 $B^{c,b}$ and $B^{c,l}$ may be different and thus the hedger and the counterparty
may be subject to different financial conditions with respect to the margin account.
Then we define the process $\Theta $ by setting
\bde
\Theta_{t} := (-A)_t^{-C}+A_t^{C}=-\int_0^t |\pC_u| (B^{\pCc,b}_u)^{-1}\, dB^{\pCc,b}_u
+ \int_0^t |\pC_u| (B^{\pCc,l}_u)^{-1}\, dB^{\pCc,l}_u.
\ede
Let us postulate, in addition, that the processes $B^{\pCc,b}$ and $B^{\pCc,l}$ are absolutely
continuous, so that
\begin{align*}
dB^{\pCc,b}_t = r^{\pCc,b}_t B^{\pCc,b}_t \, dt, \quad dB^{\pCc,l}_t = r^{\pCc,l}_t B^{\pCc,l}_t \, dt,
\end{align*}
for some non-negative processes $r^{\pCc,b}$ and $r^{\pCc,l}$ satisfying $r^{\pCc,l}\leq r^{\pCc,b}$. The additional assumption that $r^{\pCc,l}\leq r^{\pCc,b}$ means that the counterparty has the advantage over the hedger in regard to the margin account.
Indeed, when posting (resp., receiving) the collateral, the counterparty obtains a higher (resp., lower) interest than the hedger.

Under the assumption that $r^{\pCc,l}\leq r^{\pCc,b}$, the process $\Theta$ is decreasing and thus $\Theta_{t}\leq0$ for all $t\in[0,T]$. Then, in all foregoing considerations in the paper, the process $A^{C}$ should be replaced by $A^{C}-\Theta$. For example, in Lemma \ref{nettled wealth formula} or in Section \ref{sect4} when we consider the counterparty's BSDE of the contract $(A,C)$, we should replace $A^{C}$ by $A^{C}-\Theta$ or, equivalently, replace $F^{C}$ by $F^{C}-\Theta$. Since $\Theta $ is a decreasing process, we claim that all the results will still hold, except for Theorem \ref{stability property of price}.

Let us finally mention that if  $r^{\pCc,l}\ge r^{\pCc,b}$, which means that the hedger has the advantage over
the counterparty in regard to the margin account, then the process $\Theta$ is increasing and thus most results
established in what follows will no longer be valid.
\erem

The following commonly standard assumption will allow us to derive more explicit formulae for the wealth dynamics
and thus also to compute the so-called {\it generator} (or {\it driver}) for the associated BSDEs.

\bhyp \label{assumption for absolutely continuous}
The riskless accounts are absolutely continuous, so that they can be represented as follows:
\be \label{absol}
dB^{l}_{t}=r^{l}_{t}B^{l}_{t}\, dt, \quad dB^{b}_{t}=r^{b}_{t}B^{b}_{t}\, dt, \quad dB^{i,b}_{t}=r^{i,b}_{t}B^{i,b}_{t}\, dt,
\ee
for some $\mathbb{G}$-adapted processes $r^{l}$, $r^{b}$ and $r^{i,b}$ for $i=1,2,\ldots,d$. Moreover, we assume $0 \leq \rll \leq \rbb$ and $\rll \leq \ribb$ for $i=1,2, \dots , d$.
\ehyp

Let the processes $\wt S^{i,l,\textrm{cld}}$ and $\wt S^{i,b,\textrm{cld}}$ for $i=1,2, \dots , d$ be given by the following expressions
\bde % \label{processes}
\wt S^{i,l,{\textrm{cld}}}_t :=  (\Blr_t)^{-1}S^i_t + \int_{(0,t]} (\Blr_u)^{-1} \, d\pA^i_u
\ede
and
\bde
\wt
S^{i,b,{\textrm{cld}}}_t :=  (\Bbr_t)^{-1}S^i_t + \int_{(0,t]} (\Bbr_u)^{-1} \, d\pA^i_u
\ede
so that
\be \label{cumulative dividend risk asset price1}
d\wt S^{i,l,{\textrm{cld}}}_t=(\Blr_t)^{-1}\left(dS^i_t - \rll_t S^i_t \, dt + d\pA^i_t\right)
\ee
and
\be  \label{cumulative dividend risk asset price2}
d\wt S^{i,b,{\textrm{cld}}}_t=(\Bbr_t)^{-1}\left(dS^i_t - \rbb_t S^i_t \, dt + d\pA^i_t\right).
\ee
We also denote
\bde
A^{C,l}_t := \int_{(0,t]}(\Blr_{u})^{-1}\, dA^C_{u}, \quad A^{C,b}_t := \int_{(0,t]}(\Bbr_{u})^{-1}\, dA^C_{u}.
\ede
In view (\ref{value of strategy 3}), the following lemmas are straightforward (see also Lemma 5.1 and Remark 5.3 in \cite{BR-2014}).

\bl \label{discounted wealth of strategy lending}
The discounted wealth  $Y^{l} := \wt V^{p,l}(x,\phi , \pA ,\pC)= (\Blr)^{-1} V^p(x,\phi , \pA ,\pC)$ satisfies
\bde % \label{BSDE for discounted lending wealth of strategy}
dY^{l}_t  =\sum_{i=1}^{d} \xi^{i}_t \, d \wt S^{i,l,{\textrm{cld}}}_t
+ \wt{f}_l(t, Y ^{l}_t, \xi_t )\, dt + dA^{C,l}_t
\ede
where the mapping $\wt{f}_l : \Omega \times [0,T] \times \mathbb{R}\times\mathbb{R}^{d} \to \mathbb{R}$ is given by
\be \label{drift function lending}
\wt{f}_l( t, y ,z ): = (\Blr_t)^{-1} f_l(t,\Blr_t y ,z ) -  \rll_t y
\ee
and $f_l : \Omega \times [0,T] \times \mathbb{R}\times\mathbb{R}^{d} \to \mathbb{R}$ equals
\bde
f_l(t, y,z)  :=   \sum_{i=1}^d \rll_t z^i S^i_t
- \sum_{i=1}^d \ribb_t( z^i S^i_t )^+
+   \rll_t \Big( y  + \sum_{i=1}^d ( z^i S^i_t )^- \Big)^+
 - \rbb_t \Big( y + \sum_{i=1}^d ( z^i S^i_t )^- \Big)^-  .
\ede
\el

%\bl \label{discounted wealth of strategy lending}
%The discounted wealth  $Y^{l} := \wt V^{p,l}(x,\phi , \pA ,\pC)= (\Blr)^{-1} V^p(x,\phi , \pA ,\pC)$ satisfies
%\bde % \label{BSDE for discounted lending wealth of strategy}
%dY^{l}_t  =\sum_{i=1}^{d} Z^{l,i}_t \, d \wt S^{i,l,{\textrm{cld}}}_t
%+ \wt{f}_l(t, Y ^{l}_t, Z^{l}_t )\, dt +(\Blr_t)^{-1} \, dA^C_t
%\ede
%where the mapping $\wt{f}_l : \Omega \times [0,T] \times \mathbb{R}\times\mathbb{R}^{d} \to \mathbb{R}$ is given by
%\be \label{drift function lending}
%\wt{f}_l( t, y ,z ): = (\Blr_t)^{-1} f_l(t,\Blr_t y ,z ) -  \rll_t y
%\ee
%and $f_l : \Omega \times [0,T] \times \mathbb{R}\times\mathbb{R}^{d} \to \mathbb{R}$ equals
%\bde
%f_l(t, y,z)  :=   \sum_{i=1}^d \rll_t z^i S^i_t
%- \sum_{i=1}^d \ribb_t( z^i S^i_t )^+
%+   \rll_t \Big( y  + \sum_{i=1}^d ( z^i S^i_t )^- \Big)^+
% - \rbb_t \Big( y + \sum_{i=1}^d ( z^i S^i_t )^- \Big)^-  .
%\ede
%\el

\bl \label{discounted wealth of strategy borrowing}
The discounted wealth  $Y^{b} := \wt V^{p,b}(x,\phi , \pA ,\pC)= (\Bbr)^{-1} V^p (x,\phi , \pA ,\pC)$ satisfies
\bde % \label{BSDE for discounted borrowing wealth of strategy}
dY^{b}_t  =\sum_{i=1}^{d} \xi^{i}_t \, d \wt S^{i,b,{\textrm{cld}}}_t
+ \wt{f}_b(t, Y^{b}_t, \xi_t )\, dt + dA^{C,b}_t
\ede
where the mapping $\wt{f}_b : \Omega \times [0,T] \times \mathbb{R}\times\mathbb{R}^{d} \to \mathbb{R}$ is given by
\be \label{drift function borrowing}
\wt{f}_b( t, y ,z ): = (\Bbr_t)^{-1} f_b(t,\Bbr_t y ,z ) -  \rbb_t y
\ee
and $f_b : \Omega \times [0,T] \times \mathbb{R}\times\mathbb{R}^{d} \to \mathbb{R}$ equals
\bde
f_b(t, y,z)  :=   \sum_{i=1}^d \rbb_t z^i S^i_t
- \sum_{i=1}^d \ribb_t( z^i S^i_t )^+
+   \rll_t \Big( y  + \sum_{i=1}^d ( z^i S^i_t )^- \Big)^+
 - \rbb_t \Big( y + \sum_{i=1}^d ( z^i S^i_t )^- \Big)^-  .
\ede
\el

%\bl \label{discounted wealth of strategy borrowing}
%The discounted wealth  $Y^{b} := \wt V^{p,b}(x,\phi , \pA ,\pC)= (\Bbr)^{-1} V^p (x,\phi , \pA ,\pC)$ satisfies
%\bde % \label{BSDE for discounted borrowing wealth of strategy}
%dY^{b}_t  =\sum_{i=1}^{d}Z^{b,i}_t \, d \wt S^{i,b,{\textrm{cld}}}_t
%+ \wt{f}_b(t, Y^{b}_t, Z^{b}_t )\, dt +(\Bbr_t)^{-1} \, dA^C_t
%\ede
%where the mapping $\wt{f}_b : \Omega \times [0,T] \times \mathbb{R}\times\mathbb{R}^{d} \to \mathbb{R}$ is given by
%\be \label{drift function borrowing}
%\wt{f}_b( t, y ,z ): = (\Bbr_t)^{-1} f_b(t,\Bbr_t y ,z ) -  \rbb_t y
%\ee
%and $f_b : \Omega \times [0,T] \times \mathbb{R}\times\mathbb{R}^{d} \to \mathbb{R}$ equals
%\bde
%f_b(t, y,z)  :=   \sum_{i=1}^d \rbb_t z^i S^i_t
%- \sum_{i=1}^d \ribb_t( z^i S^i_t )^+
%+   \rll_t \Big( y  + \sum_{i=1}^d ( z^i S^i_t )^- \Big)^+
% - \rbb_t \Big( y + \sum_{i=1}^d ( z^i S^i_t )^- \Big)^-  .
%\ede
%\el

\newpage

%%%%%%%%%%%%%%%%%%%%%%%%%%%%%%%%%%%%%%%%%%%%%%%%%%%%%%%%%%%%%%%%%%%
%%%%%%%%%%%%%%%%%%%%%%%%%%%%%%%%%%%%%%%%%%%%%%%%%%%%%%%%%%%%%%%%%%%
\section{Arbitrage Opportunities and Ex-Dividend Prices} \label{sect3}
%%%%%%%%%%%%%%%%%%%%%%%%%%%%%%%%%%%%%%%%%%%%%%%%%%%%%%%%%%%%%%%%%%%
%%%%%%%%%%%%%%%%%%%%%%%%%%%%%%%%%%%%%%%%%%%%%%%%%%%%%%%%%%%%%%%%%%%

We consider throughout the hedger's self-financing trading strategies $(x, \phi , \pA , \pC )$, as specified by Definition \ref{definition for self financing}, where $x$ is the hedger's initial endowment. We set $V_{T}^{0}(x):=x\Blr_T\I_{\{x\ge0\}}+x\Bbr_T\I_{\{x<0\}}$ and we
define the {\it discounted wealth process} $\wh V(x, \phi,  A, C )$ by the following expression, for all $t \in [0,T]$,
\[
\wh V_t(x, \phi, A, C ):=(\Blr_t)^{-1} V_t (x,\phi , \pA ,\pC)\I_{\{x\ge0\}}+(\Bbr_t)^{-1} V_t (x,\phi , \pA ,\pC)\I_{\{x<0\}}.
\]

%%%%%%%%%%%%%%%%%%%%%%%%%%%%%%%%%%%%%%%%%%%%%%%%%%%%%
\subsection{Netted Wealth and Arbitrage Opportunities}   \label{sect3.1}
%%%%%%%%%%%%%%%%%%%%%%%%%%%%%%%%%%%%%%%%%%%%%%%%%%%%%

We first extend the results obtained in Section 3 of \cite{BR-2014} to the case of a collateralized contract.
For the financial interpretation of the {\it netted wealth}, the reader is referred to \cite{BR-2014}.
Let  us only mention that $A_0= p^{A,C} \in \mathbb{R}$ stands here for a generic price of a contract at time 0,
as seen from the perspective of the hedger.

\bd \label{nettled wealth}
The {\it netted wealth} $\Vnet(x, \phi , \pA, C)$ of a trading strategy $(x, \phi, \pA, C)$ is given by
$\Vnet(x, \phi , \pA, C):= V(x, \phi , \pA, C) + V(0, \wt \phi  , -\pA, -C)$ where $(0,  \wt \phi ,-A, -C)$
is the unique self-financing strategy satisfying the following conditions: \hfill \break
(i) $V_0(0, \wt \phi  , -\pA ) = - A_0 $, \hfill \break
(ii) $\widetilde{\xi}^i_t=0$ (hence $\widetilde{\psi}^{i,b}_t = 0$ in view of (\ref{portfolio choose}))
for all $i=1,2,\dots ,d$ and $t \in [0,T]$, \hfill \break
(iii) $\wt{\psi}^l_t \geq 0 ,\, \wt{\psi}^b_t \leq 0$ and  $\wt{\psi}^l_t \wt{\psi}^b_t =0$ for all $t \in [0,T]$.
\ed

We note that
\bde
\Vnet_0 (x, \phi , \pA, C) = V_0(x, \phi , \pA, C) + V_0(0, \wt \phi  , -\pA, -C)= x+ A_0+C_{0} - A_0-C_{0}=x ,
\ede
so that the initial netted wealth $\Vnet_0 (x, \phi , \pA, C)$ is independent of $(A_{0},C_{0})$ and it simply equals
the hedger's initial endowment.

\bd \label{admissible strategy}
A self-financing trading strategy $(x,\phi ,A,C)$ is {\it admissible for the hedger} whenever the discounted netted wealth process $\whVnet (x, \phi ,A,C)$ is bounded from below by a constant.
\ed

\bd \label{arbitrage using nettled wealth}
An admissible trading strategy $(x, \phi ,A,C)$ is an {\it arbitrage opportunity for the hedger} with respect to $(A,C)$ whenever
\be
\P ( \Vnet_T(x, \phi , A, C) \geq \VLL_T (x))=1\quad \text{ and }\quad \P ( \Vnet_T (x, \phi , A, C) > \VLL_T (x) ) > 0. \nonumber
\ee
A market model is {\it  arbitrage-free} for the hedger if no arbitrage opportunities for the hedger exist in regard to
any contract $(A,C)$.
\ed

The condition that the discounted netted wealth process $\whVnet (x, \phi , A, C)$ is bounded
from below by a constant is a commonly used criterion of {\it admissibility}, which ensures that, if
the process $\whVnet (x, \phi , A, C)$ a local martingale under some equivalent probability measure, then it is also
a supermartingale. It is well known that some technical assumption of this nature cannot be avoided even in the classic case of
the Black and Scholes model.

\bl \label{nettled wealth formula}
We have $\Vnet (x, \phi , \pA, C) =  V(x, \phi , \pA, C) + U(A,C)$,
where the $\gg$-adapted process of finite variation $U(A,C)=U$ is the unique solution to the following equation
\be \lab{pyy2}
U_t = \int_0^t (\Blr_u)^{-1} ( U_u-C_{u})^+\, d\Blr_u - \int_0^t (\Bbr_u)^{-1} ( U_u-C_{u})^-\, d\Bbr_u - \pA_t-F^{C}_{t}
\ee
where $F^{C}$ is defined by (\ref{definition for FC}).
\el

\proof
We set $\widetilde{\xi}^i=\widetilde{\psi}^{i,b}=0$ in (\ref{value of strategy 1}) and (\ref{value of strategy 2}). Then the process $V^{p} := V^{p}(0, \wt{\psi}^l, \wt{\psi}^b ,\etab, \etal, -\pA,-C)$ satisfies $V_t^{p} = \wt{\psi}^{l}_t \Blr_t + \wt{\psi}^{b}_t \Bbr_t$ for every $t \in [0,T]$. Noting that $V^{c}:=V^{c}(0, \wt{\psi}^l, \wt{\psi}^b ,\etab, \etal, -\pA,-C)=C$ and recalling the definition of $F^{C}$ and $A^{C}$, we obtain
\begin{align*}
V_t^{p}&=\int_0^t (\Blr_u)^{-1} ( V_u^{p})^+\, d\Blr_u - \int_0^t (\Bbr_u)^{-1} ( V_u^{p})^- \, d\Bbr_u - \pA_t+F^{-C}_{t}
-V_t^{c}(0, \wt{\psi}^l, \wt{\psi}^b ,\etab, \etal, -\pA,-C)\\
&=\int_0^t (\Blr_u)^{-1} ( V_u^{p})^+\, d\Blr_u - \int_0^t (\Bbr_u)^{-1} ( V_u^{p})^- \, d\Bbr_u - \pA_t-F^{C}_{t}-C_{t}\\
&=\int_0^t (\Blr_u)^{-1} ( V_u^{p})^+\, d\Blr_u - \int_0^t (\Bbr_u)^{-1} ( V_u^{p})^- \, d\Bbr_u -\pA_t^{C}.
\end{align*}
Consequently, the process $V:=V(0, \wt{\psi}^l, \wt{\psi}^b ,\etab, \etal, -\pA,-C)=V^{p}+V^{c}$ satisfies
\bde
V_t=\int_0^t (\Blr_u)^{-1} ( V_u-C_{u})^+\, d\Blr_u - \int_0^t (\Bbr_u)^{-1} ( V_u-C_{u})^- \, d\Bbr_u - \pA_t-F^{C}_{t}
\ede
and thus the assertion of the lemma follows.
\endproof

\bhyp \label{assumption for lending cumulative dividend price}
There exists a probability measure $\PT^l $ equivalent to $\P $ such that the
processes $\wt S^{i,l,\textrm{cld}},\, i=1,2, \dots ,d$ are $(\PT^l , \gg)$-local martingales.
\ehyp

\bp \label{proposition for arbitrage free}
Under Assumptions \ref{assumption for absolutely continuous} and \ref{assumption for lending cumulative dividend price}, if $x \geq 0$, then the market model of Section \ref{sect2.3} is arbitrage-free for the hedger in regard to
any contract $(A,C)$.
\ep

\proof
In view of (\ref{value of strategy 3}) and the postulated inequalities: $\rll \leq \rbb$ and $\rll \leq \ribb$ for all $i$, the process $V^{p}:=V^{p}(x, \phi , \pA, C )$  satisfies
\begin{align*}
dV_t^{p}= \, \, & \sum_{i=1}^d \xi^i_t \big(dS^i_t + d\pA^i_t  \big)
- \sum_{i=1}^d \ribb_t( \xi^i_t S^i_t )^+  \, dt + d \pA_t^{C}
 \\ &+   \rll_t \Big( V_t^{p}+ \sum_{i=1}^d ( \xi^i_t S^i_t )^- \Big)^+ \, dt
- \rbb_t \Big( V_t^{p}+ \sum_{i=1}^d ( \xi^i_t S^i_t )^- \Big)^- \, dt \\
\leq \, \, & \sum_{i=1}^d \xi^i_t \big(dS^i_t + d\pA^i_t  \big)
- \sum_{i=1}^d \ribb_t( \xi^i_t S^i_t )^+  \, dt + d\pA_t^{C}
 \\ &+   \rll_t \Big( V_t^{p}  + \sum_{i=1}^d ( \xi^i_t S^i_t )^- \Big)^+ \, dt
- \rll_t \Big( V_t^{p}  + \sum_{i=1}^d ( \xi^i_t S^i_t )^- \Big)^- \, dt \\
= \, \, &\rll_t  V_t^{p} \, dt + \sum_{i=1}^d \xi^i_t \big(dS^i_t + d\pA^i_t  \big) + d\pA_t^{C}
- \sum_{i=1}^d \ribb_t( \xi^i_t S^i_t )^+  \, dt +   \rll_t \sum_{i=1}^d ( \xi^i_t S^i_t )^- \, dt
 \\ \leq \, \, &\rll_t  V_t^{p} \, dt + \sum_{i=1}^d \xi^i_t \big(dS^i_t - \rll_t S^i_t \, dt + d\pA^i_t \big)
  + d\pA_t^{C}.
\end{align*}
Consequently, the discounted wealth $V^{l,p} := (\Blr)^{-1} V^{p}$ satisfies
\begin{align*}
d V^{l,p}_t \leq \sum_{i=1}^d \xi^i_t (\Blr_t)^{-1} \big( dS^i_t - \rll_t S^i_t \, dt + d\pA^i_t  \big) + (\Blr_t)^{-1} \, d\pA_t^{C}
=\sum_{i=1}^d \xi^i_t \, d\wt S^{i,l,{\textrm{cld}}}_t +  dA^{C,l}_t.
\end{align*}
Furthermore, the netted wealth is given by the following expression  (see Lemma \ref{nettled wealth formula})
\bde
\Vnet_t (x,\phi , \pA, C) =  V_t (x,\phi , \pA, C) + U_t(A,C)=V_{t}^{p}-C_{t}+ U_t(A,C)
\ede
where the $\gg$-adapted process of finite variation $U(A,C)$ is given by \eqref{pyy2}.
Hence the discounted netted wealth, which is given by
$$
\Vnettl_t:= (\Blr_t)^{-1}\Vnet_{t}(x,\phi , \pA, C)=V^{l,p}_t-(\Blr_t)^{-1}C_{t}+ (\Blr_t)^{-1}U_t(A,C),
$$
satisfies (for brevity, we write $U(A,C)=U$)
\begin{align*}
&d\Vnettl_t=  dV^{l,p}_t-d((\Blr_t)^{-1}C_{t})+  d(  (\Blr_t)^{-1} U_t) \\
&\leq   \sum_{i=1}^d \xi^i_t \, d\wt S^{i,l,{\textrm{cld}}}_t
+ (\Blr_t)^{-2} ( U_t-C_{t})^+\, d\Blr_t - (\Blr_t)^{-1} (\Bbr_t)^{-1} ( U_t-C_{t})^- \, d\Bbr_t \medskip\\
&\qquad + U_t \, d(\Blr_t)^{-1} +(\Blr_t)^{-1}\, dA_{t}^{C}-(\Blr_t)^{-1}\, dA_{t}-(\Blr_t)^{-1}\, dF^{C}_{t}-d((\Blr_t)^{-1}C_{t})\medskip\\
%&=  \sum_{i=1}^d \xi^i_t \, d\wt S^{i,l,{\textrm{cld}}}_t
%+r^l_t (\Blr_t)^{-1}( U_t-C_{t})^+\, dt - r^b_t (\Blr_t)^{-1} ( U_t-C_{t})^- \, dt  - r^l_t (\Blr_t)^{-1} (U_t -C_{t})\, %dt\medskip\\
& = \sum_{i=1}^d \xi^i_t \, d\wt S^{i,l,{\textrm{cld}}}_t + (r^l_t - r^b_t )(\Blr_t)^{-1} ( U_t-C_{t})^- \, dt
\end{align*}
and thus
\be \label{inequality for proving arbitrage free}
\Vnettl_t- \Vnettl_0\leq \sum_{i=1}^d \int_{(0,t]} \xi^i_u \, d\wt S^{i,l,{\textrm{cld}}}_u .
\ee
First, the assumption that the process $\Vnettl$ is bounded from below, implies that the right-hand side in \eqref{inequality for proving arbitrage free} is a $(\PT^l,\gg)$-supermartingale, which is null at $t=0$. Next, $\VLL_T(x) = \Blr_T x$ (since $x\ge0$). From \eqref{inequality for proving arbitrage free}, we thus obtain
\bde
(\Blr_T)^{-1} \big( \Vnet_T (x,\phi , A ) - \VLL_T (x) \big) \leq \sum_{i=1}^d \int_{(0,T]} \xi^i_t \, d\wt S^{i,l,{\textrm{cld}}}_t.
\ede
Since $\PT^l $ is equivalent to $\P$, we conclude that either $\Vnet_T (x,\phi , A,C) = \VLL_T(x)$
or $\P ( \Vnet_T (x,\phi , A,C) < \VLL_T(x))>0$. This means that an arbitrage opportunity may not arise and thus the market model with partial netting is arbitrage-free for the hedger in regard to any contract $(A,C)$.
\endproof

\bhyp \label{assumption for borrowing cumulative dividend price}
There exists a probability measure $\PT^b $ equivalent to $\P $ such that the
processes $\wt S^{i,b,\textrm{cld}},\, i=1,2, \dots ,d$ are $(\PT^b , \gg)$-local martingales.
\ehyp

\brem  \label{rates assumption for negative wealth}
Similarly as in Remark 3.2 of \cite{BR-2014}, we observe that the statement of Proposition \ref{proposition for arbitrage free} is also true for $x\leq 0$, provided that Assumption \ref{assumption for borrowing cumulative dividend price} is valid and  $r^b \leq r^{i,b}$ for $i=1,2,\ldots,d$. For the hedger, one can then show that
\bde
d\Vnettb_t \leq \sum_{i=1}^d \xi^i_t \, d\wt S^{i,b,{\textrm{cld}}}_t + (r^l_t - r^b_t )(\Bbr_t)^{-1} ( U_t(A,C)-C_{t})^+ \, dt
\ede
and thus, using similar arguments as above, we conclude that there is no arbitrage for the hedger in regard to any
contract $(A,C)$.
\erem

\brem
Let Assumption \ref{additional assumption for collateral account} be valid. Then Definition \ref{arbitrage using nettled wealth}, Proposition \ref{proposition for arbitrage free} and Remark \ref{rates assumption for negative wealth} apply not only to the hedger, but also to the counterparty. Therefore, if both parties have non-negative initial endowments (resp., both have non-positive initial endowments), Assumption \ref{assumption for lending cumulative dividend price} (resp., Assumption \ref{assumption for borrowing cumulative dividend price}) holds, and  $r^l \leq r^{i,b}$ (resp., $r^b \leq r^{i,b}$) for all $i$, then the model is arbitrage-free for both parties. When the initial endowments have opposite signs then if Assumptions \ref{assumption for lending cumulative dividend price} and \ref{assumption for borrowing cumulative dividend price} are valid and $r^b \leq r^{i,b}$) for all $i$, then the model is arbitrage-free for both parties.
\erem

%%%%%%%%%%%%%%%%%%%%%%%%%%%%%%%%%%%%%%%%%%%%%%%%%%%%%
\subsection{Extended Arbitrage Opportunities}   \label{sect3.1b}
%%%%%%%%%%%%%%%%%%%%%%%%%%%%%%%%%%%%%%%%%%%%%%%%%%%%%

Results of Section \ref{sect3.1}  give only a partial answer to the question whether a market model
with partial netting is arbitrage-free.  We will now attempt to give a deeper analysis of the arbitrage-free
property for all contracts under specific assumptions on prices of risky assets.
To this end, we introduce the following definition (see Remark 3.1 in \cite{BR-2014}).

\bd {\rm An {\it extended arbitrage opportunity} with respect to the contract $(A,C)$ for the hedger
with the initial endowment $x$ is a pair $(\Xhat, \Fhat , A)$ and $(\Xtil , \Ftil , -A )$ of admissible strategies
such that $x = \Xhat+\Xtil $ and
\bde
\P ( \Vnet_T \geq \VLL_T (x))=1\quad \text{ and }\quad \P ( \Vnet_T  > \VLL_T (x) ) > 0
\ede
where the {\it netted wealth} $\Vnet = \Vnet ( \Xhat , \Xtil , \Fhat , \Ftil , A ,C)$ is given by}
\bde
\Vnet:=  V(\Xhat,\Fhat , \pA, C) + V(\Xtil,\Ftil , -\pA, -C).
\ede
\ed

The next result gives sufficient conditions for non-existence of extended arbitrage opportunities for the hedger.

\bp \label{remark for non-arbitrage model}
Assume that there exist some $\gg$-adapted processes $\beta^{i}$ satisfying $r^{b}\leq\beta^{i}\leq r^{i,b}$
and a probability measure $\PTb $ equivalent to $\P$ such that the auxiliary processes $\wt S^{i,\textrm{cld}},\,
i=1,2, \dots ,d$, which are given by
\be \label{frrt}
d\wt S^{i,\textrm{cld}}_t = dS^i_t + d\pA^i_t - \beta^{i}_{t}S_{t}^{i}\,dt ,
\ee
are continuous, square-integrable, $(\PTb , \gg)$-martingales. Then no extended arbitrage opportunity exists
for the hedger in respect of any contract $(A,C)$ and any initial endowment $x\in\mathbb{R}$.
\ep

\proof
Note that the process $\widehat{V}^{p}:=V^{p}(\Xhat, \Fhat , \pA, C )$ is governed by
\begin{align*}
d\widehat{V}_t^{p}= \, \, & \sum_{i=1}^d \widehat{\xi}^i_t \big(dS^i_t + d\pA^i_t  \big)
- \sum_{i=1}^d \ribb_t( \widehat{\xi}^i_t S^i_t )^+  \, dt + d \pA_t^{C}
 \\ &+   \rll_t \Big( \widehat{V}_t^p+ \sum_{i=1}^d ( \widehat{\xi}^i_t S^i_t )^- \Big)^+ \, dt
- \rbb_t \Big( \widehat{V}_t^p+ \sum_{i=1}^d ( \widehat{\xi}^i_t S^i_t )^- \Big)^- \, dt.
\end{align*}
and $\widetilde{V}^{p}:=V^{p}(\Xtil, \Ftil , -\pA, -C )$  satisfies
\begin{align*}
d\widetilde{V}_t^{p}= \, \, & \sum_{i=1}^d \widetilde{\xi}^i_t \big(dS^i_t + d\pA^i_t  \big)
- \sum_{i=1}^d \ribb_t( \widetilde{\xi}^i_t S^i_t )^+  \, dt -d \pA_t^{C}
 \\ &+   \rll_t \Big( \widetilde{V}_t^{p}+ \sum_{i=1}^d ( \widetilde{\xi}^i_t S^i_t )^- \Big)^+ \, dt
- \rbb_t \Big( \widetilde{V}_t^{p}+ \sum_{i=1}^d ( \widetilde{\xi}^i_t S^i_t )^- \Big)^- \, dt.
\end{align*}
We observe that the netted wealth satisfies
\bde
\Vnet:=  V(\Xhat,\Fhat , \pA, C) + V(\Xtil,\Ftil , -\pA, -C)=\widehat{V}^{p}-C+ \widetilde{V}^{p}+C=\widehat{V}^{p}+\widetilde{V}^{p}
\ede
and thus
\begin{align*}
d\Vnet_t= \, \, & \sum_{i=1}^d (\widehat{\xi}^i_t+\widetilde{\xi}^i_t)\big(dS^i_t + d\pA^i_t  \big)
- \sum_{i=1}^d \ribb_t( \widehat{\xi}^i_t S^i_t )^+  \, dt - \sum_{i=1}^d \ribb_t(\widetilde{ \xi}^i_t S^i_t )^+  \, dt
 \\ &+   \rll_t \Big( \widehat{V}_t^p+ \sum_{i=1}^d ( \widehat{\xi}^i_t S^i_t )^- \Big)^+ \, dt
- \rbb_t \Big( \widehat{V}_t^p+ \sum_{i=1}^d ( \widehat{\xi}^i_t S^i_t )^- \Big)^- \, dt
 \\ &+   \rll_t \Big( \widetilde{V}_t^{p}+ \sum_{i=1}^d ( \widetilde{\xi}^i_t S^i_t )^- \Big)^+ \, dt
- \rbb_t \Big( \widetilde{V}_t^{p}+ \sum_{i=1}^d ( \widetilde{\xi}^i_t S^i_t )^- \Big)^- \, dt.
\end{align*}
Since $r^{l}\leq r^{b}$, we obtain
\begin{align}\label{1 netted weath inequality}
d\Vnet_t\leq & \sum_{i=1}^d (\widehat{\xi}^i_t+\widetilde{\xi}^i_t) \big(dS^i_t + d\pA^i_t  \big)
- \sum_{i=1}^d \ribb_t( \widehat{\xi}^i_t S^i_t )^+  \, dt - \sum_{i=1}^d \ribb_t(\widetilde{ \xi}^i_t S^i_t )^+  \, dt \nonumber
\\ &\mbox{}+   \rll_t \Big( \widehat{V}_t^p+\widetilde{V}_t^{p}+ \sum_{i=1}^d ( \widehat{\xi}^i_t S^i_t )^-+ \sum_{i=1}^d ( \widetilde{\xi}^i_t S^i_t )^-  \Big) \, dt
\end{align}
and
\begin{align} \label{2 netted weath inequality}
d\Vnet_t\leq & \sum_{i=1}^d (\widehat{\xi}^i_t+\widetilde{\xi}^i_t) \big(dS^i_t + d\pA^i_t  \big)
- \sum_{i=1}^d \ribb_t( \widehat{\xi}^i_t S^i_t )^+  \, dt - \sum_{i=1}^d \ribb_t(\widetilde{ \xi}^i_t S^i_t )^+  \, dt \nonumber
\\ &\mbox{}+   \rbb_t \Big( \widehat{V}_t^p+\widetilde{V}_t^{p}+ \sum_{i=1}^d ( \widehat{\xi}^i_t S^i_t )^-+ \sum_{i=1}^d ( \widetilde{\xi}^i_t S^i_t )^-  \Big) \, dt .
\end{align}
Using \eqref{1 netted weath inequality} and the equality $\Vnet=\widehat{V}^{p}+\widetilde{V}^{p}$, we obtain
for the process $\Vnettl:= (\Blr)^{-1}\Vnet$
\begin{align*}
&d\Vnettl_t=  (\Blr_t)^{-1}d\Vnet_{t}-r^{l}_{t}(\Blr_t)^{-1}\Vnet_{t}\, dt\\
&\leq   (\Blr_t)^{-1}\bigg(\sum_{i=1}^d \widehat{\xi}^i_t \big(dS^i_t + d\pA^i_t  \big)
- \sum_{i=1}^d \ribb_t( \widehat{\xi}^i_t S^i_t )^+  \, dt +  \sum_{i=1}^d r^{l}_{t}( \widehat{\xi}^i_t S^i_t )^-\, dt\bigg)\\
&\quad + (\Blr_t)^{-1}\bigg(\sum_{i=1}^d \widetilde{\xi}^i_t \big(dS^i_t + d\pA^i_t  \big)
- \sum_{i=1}^d \ribb_t( \widetilde{\xi}^i_t S^i_t )^+  \, dt +  \sum_{i=1}^d r^{l}_{t}( \widetilde{\xi}^i_t S^i_t )^-\, dt\bigg)\\
&=(\Blr_t)^{-1}\bigg(\sum_{i=1}^d \widehat{\xi}^i_t \big(dS^i_t + d\pA^i_t-\beta_{t}^{i}S_{t}^{i}\, dt  \big)
- \sum_{i=1}^d \ribb_t( \widehat{\xi}^i_t S^i_t )^+  \, dt +  \sum_{i=1}^d r^{l}_{t}( \widehat{\xi}^i_t S^i_t )^-\, dt+\sum_{i=1}^d \beta_{t}^{i}\widehat{\xi}^i_tS_{t}^{i} \, dt\bigg)\\
&\quad +(\Blr_t)^{-1}\bigg(\sum_{i=1}^d \widetilde{\xi}^i_t \big(dS^i_t + d\pA^i_t-\beta_{t}^{i}S_{t}^{i}\, dt  \big)
- \sum_{i=1}^d \ribb_t( \widetilde{\xi}^i_t S^i_t )^+  \, dt +  \sum_{i=1}^d r^{l}_{t}( \widetilde{\xi}^i_t S^i_t )^-\, dt+\sum_{i=1}^d \beta_{t}^{i}\widetilde{\xi}^i_tS_{t}^{i} \, dt\bigg).
\end{align*}
Similarly, in view \eqref{2 netted weath inequality}, the process $\Vnettb := (\Bbr)^{-1}\Vnet $ satisfies
\begin{align*}
&d\Vnettb_t=  (\Bbr_t)^{-1}d\Vnet_{t}-r^{b}_{t}(\Bbr_t)^{-1}\Vnet_{t}\, dt\\
&=(\Bbr_t)^{-1}\bigg(\sum_{i=1}^d \widehat{\xi}^i_t \big(dS^i_t + d\pA^i_t-\beta_{t}^{i}S_{t}^{i}\, dt  \big)
- \sum_{i=1}^d \ribb_t( \widehat{\xi}^i_t S^i_t )^+  \, dt +  \sum_{i=1}^d r^{b}_{t}( \widehat{\xi}^i_t S^i_t )^-\, dt+\sum_{i=1}^d \beta_{t}^{i}\widehat{\xi}^i_tS_{t}^{i}\, dt\bigg)\\
&\quad +(\Blr_t)^{-1}\bigg(\sum_{i=1}^d \widetilde{\xi}^i_t \big(dS^i_t + d\pA^i_t-\beta_{t}^{i}S_{t}^{i}\, dt  \big)
- \sum_{i=1}^d \ribb_t( \widetilde{\xi}^i_t S^i_t )^+  \, dt +  \sum_{i=1}^d r^{b}_{t}( \widetilde{\xi}^i_t S^i_t )^-\, dt+\sum_{i=1}^d \beta_{t}^{i}\widetilde{\xi}^i_tS_{t}^{i} \, dt\bigg).
\end{align*}
Since the process $\beta^i$ satisfies $r^{l}\leq\beta^{i}\leq r^{i,b}$ for every $i=1,2,\dots ,d$, we obtain
\bde
\sum_{i=1}^d \beta_{t}^{i}\widehat{\xi}^i_tS_{t}^{i}\leq \sum_{i=1}^d r_{t}^{i,b}(\widehat{\xi}^i_tS_{t}^{i})^{+}-\sum_{i=1}^d r_{t}^{l}(\widehat{\xi}^i_tS_{t}^{i})^{-}.
\ede
Under the stronger condition that $r^{b}\leq\beta^{i}\leq r^{i,b}$ is satisfied for every $i=1,2,\dots ,d$, we also have that
\bde
\sum_{i=1}^d \beta_{t}^{i}\widehat{\xi}^i_tS_{t}^{i}\leq \sum_{i=1}^d r_{t}^{i,b}(\widehat{\xi}^i_tS_{t}^{i})^{+}-\sum_{i=1}^d r_{t}^{b}(\widehat{\xi}^i_tS_{t}^{i})^{-}.
\ede
By assumption, there exists a probability measure $\PTb $ equivalent to $\P$ such that the
processes $\wt S^{i,\textrm{cld}},\, i=1,2, \dots ,d$ are continuous, square-integrable, $(\PTb , \gg)$-martingales, where
$\wt S^{i,\textrm{cld}}$ is given by \eqref{frrt} for some $\gg$-adapted processes $\beta^{i}$ satisfying
$r^{b}\leq\beta^{i}\leq r^{i,b}$. Then
\be \label{inequality for proving arbitrage freex}
\Vnettl_t- \Vnettl_0\leq \sum_{i=1}^d \int_{(0,t]} (\widehat{\xi}^i_u+\widetilde{\xi}^i_u) \, d\wt S^{i,{\textrm{cld}}}_u
\ee
and
\be\label{2 inequality for proving arbitrage free}
\Vnettb_t- \Vnettb_0\leq \sum_{i=1}^d \int_{(0,t]} (\widehat{\xi}^i_u+\widetilde{\xi}^i_u) \, d\wt S^{i,{\textrm{cld}}}_u .
\ee
% Let us focus on the case when the initial endowment $x$ is non-negative.
% First, the assumption that the process $\Vnettl$ is bounded from below, implies that the right-hand side in \eqref{inequality for proving arbitrage %freex} is a $(\PTb ,\gg)$-supermartingale, which is null at $t=0$. Since $x\ge0$, we have $\VLL_T(x) = \Blr_T x$. From \eqref{inequality for proving %arbitrage freex}, we thus obtain
%\bde
%(\Blr_T)^{-1} \big( \Vnet_T - \VLL_T (x) \big) \leq \sum_{i=1}^d \int_{(0,T]} (\widehat{\xi}^i_t+\widetilde{\xi}^i_t)\, d\wt %S^{i,l,{\textrm{cld}}}_t.
%\ede
%Since $\PTb $ is equivalent to $\P$, we conclude that either $\Vnet_T  = \VLL_T(x)$
%or $\P ( \Vnet_T < \VLL_T(x))>0$. This means that no extended arbitrage opportunity arises for the hedger with the initial
%endowment $x \geq 0$ in regard to any contract $(A,C)$. For the case where $x\leq0$, similar arguments can be applied
%to the process $\Vnettb$, which was shown to satisfy  \eqref{2 inequality for proving arbitrage free}.
Using standard arguments (see the proof of Proposition \ref{proposition for arbitrage free}), we deduce that the market model with partial netting is arbitrage-free for the hedger in respect of any contract $(A,C)$ and any initial endowment $x\in\mathbb{R}$.
\endproof

Let us now discuss various alternative martingale conditions, which were introduced to analyze the non-existence of (extended) arbitrage
opportunities in the present set-up. First, is easy to see that Proposition \ref{remark for non-arbitrage model} furnishes sufficient conditions ensuring that the market model with partial netting is arbitrage-free with respect to any contract for both parties with arbitrary initial endowments $x_{1},x_{2}\in\mathbb{R}$. This motivates the introduction of Assumptions  \ref{assumption for artifical cumulative dividend price} and \ref{changed assumption for artifical cumulative dividend price} in Section \ref{sect5.3}. It is fair to acknowledge that the condition $r^{b}\leq\beta^{i}\leq r^{i,b}$ is restrictive and thus this result is not fully satisfactory. However, we argue below that the condition $r^{b}\leq\beta^{i}\leq r^{i,b}$ is needed in the abstract set-up where `risky' assets are left
unspecified, so their prices may in fact be modeled through continuous processes of finite variation.

It is also worth noting that the condition $r^b \leq r^{i,b}$ in Remark  \ref{rates assumption for negative wealth}  was not due to the fact that we considered there the case when $x \leq 0$, but rather to the choice of $\Bbr$ as a discount factor. Specifically, we decided to search
for a sufficient condition for arbitrage-free property in terms of a martingale measure for processes
$\wt S^{i,b,{\textrm{cld}}}$. Since we do not make any a priori assumptions about the price processes for risky assets,
it may happen that, for instance, $S^1= S^1_0 \Bbr$ and $A^1=0$. Of course, a martingale measure for the process
$\wt S^{1,b,{\textrm{cld}}}$ exists, but the sub-model $(\Blr ,\Bbr , B^{1,b}, S^1)$ is not arbitrage-free unless
$r^b \leq r^{1,b}$. Indeed, in the present set-up, the rate $r^{1,b}$ (resp. $r^b$) can be seen as a borrowing (resp.,
lending) rate, since the non-risky return $r^b$ can be generated by the hedger by purchasing the stock $S^1$.
This argument shows that the inequality $r^b \leq r^{i,b}$ is necessary to avoid arbitrage if we do not make any other assumptions
about risky asset except for postulating the existence of a martingale measure for $\wt S^{1,b,{\textrm{cld}}}$ (by contrast,
if the stock price $S^1$ is given, say, by the Black and Scholes model then there is no need to
postulate that $r^b \leq r^{i,b}$ since any investment in $S^1$ is risky).

The condition that a martingale measure for $\wt S^{1,b,{\textrm{cld}}}$ exists is, in some sense,
weaker that the postulate that a martingale measure for $\wt S^{1,l,{\textrm{cld}}}$ exists. Indeed, in the
latter case, when the asset price is of finite variation, it equals to $S_0 \Blr$ (rather than $S_0 \Bbr$)
and thus one could conjecture that the condition $r^l \leq r^{1,b}$ is sufficient to preclude arbitrage in
the sub-model $(\Blr ,\Bbr , B^{1,b}, S^1)$. This is indeed true when $x \geq 0$, but when $x <0$ and
$r^l < r^b$, there still exists an arbitrage opportunity, since the hedger may sell stock and reduce
interest payments on his debt.

Finally, one could postulate that the process $B^{i,b}$ could be chosen as a discount factor for the $i$th risky
asset. In that case, to preclude an arbitrage opportunity of the same kind as above when $x<0$, one would need to postulate
that $r^{i,b} \geq r^b$.

In our opinion, the condition that a martingale measure for $\wt S^{i,l,{\textrm{cld}}}$ exists is more natural but,
as was explained above, it is not a sufficient condition for no-arbitrage if a `risky' asset may in fact by non-risky.
Of course, in a non-trivial model where the prices of risky assets have non-vanishing volatilities, the above-mentioned
martingale conditions will be equivalent, under mild technical assumptions, and conditions $r^l \leq r^b$
and $r^l \leq r^{i,b}$ that underpin Proposition \ref{proposition for arbitrage free} should suffice to ensure that a model is arbitrage-free for both parties with arbitrary initial endowments.

%%%%%%%%%%%%%%%%%%%%%%%%%%%%%%%%%%%%%%%%%%%%%%%%%%%%%%%%%%%%%%%%%%%
\subsection{Fair and Profitable Bilateral Prices}   \label{sect3.2}
%%%%%%%%%%%%%%%%%%%%%%%%%%%%%%%%%%%%%%%%%%%%%%%%%%%%%%%%%%%%%%%%%%%

Our next goal is to describe the range of arbitrage prices of a contract with cash flows $A$ and collateral $C$. It is rather clear from the next definition that a {\it hedger's fair price} may depend on the hedger's initial endowment $x$ and it may fail to be unique, in general.

\bd \label{definition for hedger's fair price}
We say that a real number $p^{A,C} = A_0$ is a {\it hedger's fair price} for $(A,C)$ at time 0 whenever for any self-financing trading strategy $(x, \phi ,A, C)$, such that the discounted wealth process $\wh V(x , \phi , A, C)$ is bounded from below, we have that
\be \label{non arbitrage condition}
\P \big( V_T (x, \phi, A, C) = \VLL_T (x) \big) = 1\quad \text{ or }\quad \P \big( V_T (x, \phi , A, C) < \VLL_T (x) \big) > 0 .
\ee
\ed

One may observe that the two conditions in Definition \ref{definition for hedger's fair price} are analogous to conditions of Definition \ref{arbitrage using nettled wealth}, although they have different financial meaning. Recall that Definition \ref{arbitrage using nettled wealth} deals with a possibility of offsetting a dynamically hedged contract $(A,C)$ with an
arbitrary market price by an unhedged contract $(-A,-C)$, whereas Definition \ref{definition for hedger's fair price} is concerned with finding a unilateral fair price for $(A,C)$ from the perspective of the hedger. For a more detailed discussion, the interested reader may consult \cite{BR-2014}.

Let us recall the generic definition of replication of a contract on $[t,T]$ (see Definition 5.1 in \cite{BR-2014}).

\bd \lab{def:replicate}
For a fixed $t \in [0,T]$, a self-financing trading strategy $(\VLL_{t}(x)+p_t^{A,C}, \phi , A- A_t , \pC )$,
where  $p_t^{A,C}$ is a ${\cal G}_{t}$-measurable random variable, is said to {\it replicate the collateralized
contract} $(A,C)$ on $[t,T]$ whenever $V_T(\VLL_{t}(x)+p_t^{A,C}, \phi , A-A_t, C) = \VLL_T (x)$.
\ed

We henceforth assume that the initial endowment of the hedger (resp., the counterparty) is $x_{1}$ (resp., $x_{2}$)
where $x_{1},x_{2}\in\mathbb{R}$. We consider the situation when the hedger with the initial endowment $x_1$ at
time 0 enters the contract $A$ at time $t$ and the contract can be replicated by the hedger.

\bd \label{definition of ex-dividend price}
Any ${\cal G}_{t}$-measurable random variable for which a replicating strategy for $A$ over $[t,T]$ exists is called the {\it hedger's ex-dividend price} at time $t$ for the contract $(A,C)$ and it is denoted by $P^{h}_t(x_1, A, C)$, so that for some
$\phi $ replicating $(A,C)$
\bde
V_T(\VLL_{t}(x_1)+P^{h}_t(x_1, A, C), \phi , A -A_t, C) = \VLL_T (x_1).
\ede
\ed

%It is worth noting that for $t=0$ we always have that $p_0=A_0$ and thus, for any portfolio $\phi $, the strategies $(x+p_0, \phi , A- A_0 , \pC )$ and $(x, \phi , A , \pC )$ are in fact identical. Therefore, we may simply say that a self-financing trading strategy $(x, \phi , A , \pC )$ {\it replicates} $(A,C)$ on $[0,T]$ whenever the equality $V_T(x , \phi , A ,C) = \VLL_T (x)$ holds. This equality is consistent with the first equation in definition (\ref{non arbitrage condition}) of a hedger's fair price, so we conclude that any ex-dividend price $p_0$ of $(A,C)$ at time 0 is also a hedger's fair price $\ppff$ for $(A,C)$ at time 0 (though
%the converse does not hold, in general).

\bd \label{remark for counterparty's ex-dividend price}
For an arbitrary level $x_2$ of the counterparty's initial endowment and a strategy $\phi$ replicating $(-A,-C)$,
the {\it counterparty's ex-dividend price} $P^{c}_t(x_2, -A, -C)$ at time $t$ for the contract $(-A,-C)$ is given by the equality
\bde
V_T(\VLL_{t}(x_2)-P^{c}_t(x_2, -A, -C), \phi , -A+A_t, -C) = \VLL_T (x_2).
\ede
\ed

It is clear that in Definitions \ref{definition of ex-dividend price} and  \ref{remark for counterparty's ex-dividend price},
we deal with unilateral prices, as evaluated by the hedger and the counterparty, respectively.
Note that if $x_1=x_2=x$, then $P^{h}_t(x, A, C)=p_{t}^{A,C}$ and $P^{c}_t(x, -A, -C)=-p_{t}^{-A,-C}$. Due to this convention,
the equality $P^{h}_t(x_1, A, C) = P^{c}_t(x_1, -A, -C)$ will be satisfied when Definitions \ref{definition of ex-dividend price} and \ref{remark for counterparty's ex-dividend price} applied to a standard market model with a single cash account in which the
prices are known to be independent of initial endowments $x_1$ and $x_2$. The next definition is consistent with this convention.
Note that Definition \ref{defbbg} is based on an implicit assumption that prices are uniquely defined; we address this important issue in the foregoing section.

\bd \label{defbbg}
The hedger is willing to {\it sell} (resp., to {\it buy}) a contract $(A,C)$ if $P^{h}_t(x_1, A, C) \geq 0 $
(resp., $P^{h}_t(x_1, A, C) \leq 0 $). The counterparty is willing to {\it sell} (resp., to {\it buy}) a contract
$(-A,-C)$ if $P^{c}_t(x_2, -A, -C) \leq 0$ (resp., $P^{c}_t(x_2, -A, -C)\geq 0$).
\ed

Since we place ourselves in a nonlinear framework, a natural asymmetry arises between the hedger and his counterparty, so that the price discrepancy may occur, meaning that it may happen that $P^{h}_t(x_1, A, C) \ne P^{c}_t(x_2, -A, -C)$.  However, it is expected that the two prices will typically yield a no-arbitrage range determined by the (higher) seller's price and the (lower) buyer's price, though it may also happen that both parties are willing to be sellers (or both are willing to be buyers) of a given contract. In addition, since a positive excess cash generated by one contract may be offset (partly or totally) by a negative excess cash associated with another contract, we expect that the seller's (resp., buyer's) price
for the combination of two contracts should be lower (resp., higher) than the sum of the seller's (resp., buyer's)
prices of individual contracts.

\bex \label{European call option}
Let us consider a contract $(A,C)$ with $C=0$ and $A_t = p \, \I_{[0,T]}(t) + X \I_{[T]}(t).$
If $X = - (S^i_T-K)^+$, then we deal with a European call option written by the hedger.
A natural guess is that the prices $P^{h}_0(x_{1},A,C)$ and $P^{c}_0 (x_{2},-A,-C)$ should be positive. Similarly,
if $X = (S^i_T-K)^+$, that the counterparty is the option's writer, it is natural to expect that $P^{h}_0 (x_{1},A,C)$
and $P^{c}_0 (x_{2},-A,-C)$ should be negative.  Furthermore, if $C=0$ and
$A_t = p \, \I_{[0,T]}(t) - (S^i_T-K)^+\I_{[T]}(t)$,
then we guess that the price $P^{h}_0(x_{1},A,C)$ should be independent of $x_{1}$, provided that $x_1 \ge 0$.
Indeed, as a consequence of the last constraint in (\ref{portfolio choose}), the hedger cannot use his initial endowment to buy shares for the purpose of hedging. In view of this constraint, the postulated model does not cover the standard case of different borrowing and lending rates when $r^{i,b}=r^b > r^l$ and trading is assumed to be unrestricted, so that the hedger's initial endowment can be used for hedging.

In the standard case, it is natural to expect that the hedger's price of the call option will depend on the hedger's initial endowment $x_1$. To sum up, for each particular market circumstances, the properties of ex-dividend prices may be quite different.
Nevertheless, we will argue that most of their properties can be analyzed using general results on BSDEs as a convenient tool.
\eex

Recall that $x_{1}$ and $x_{2}$ stand for the initial endowments of the hedger and the counterparty, respectively.
Due to a generic nature of a contract $(A,C)$, it is impossible to make any plausible a priori conjectures about
relative sizes and/or signs of prices. The equality $P^{h}_t(x, A, C) = P^{c}_t(x, -A, -C)$ means that both parties
agree on a common price for the contract. Otherwise, that is, if the equality $P^{h}_t(x, A, C) = P^{c}_t(x, -A, -C)$
fails to hold, then the following situations may arise:

\noindent {\bf (H.1)} $\ 0 \le P^{c}_t(x_2, -A, -C) < P^{h}_t(x_1, A, C)$,

\noindent {\bf (H.2)} $\ P^{c}_t(x_1, A, C) \le 0 < P^{h}_t(x_2, -A, -C)$,

\noindent {\bf (H.3)} $\ P^{c}_t(x_2, -A, -C) < P^{h}_t(x_1, A, C) \le 0$,

\noindent and, symmetrically,

\noindent {\bf (C.1)} $\ 0 \le P^{h}_t(x_1, A, C) < P^{c}_t(x_2, -A, -C)$,

\noindent {\bf (C.2)} $\ P^{h}_t(x_1, A, C) \le 0 < P^{c}_t(x_2, -A, -C)$,

\noindent {\bf (C.3)} $\ P^{h}_t(x_1, A, C) < P^{c}_t(x_2, -A, -C) \le 0$.

%fails to hold, then the following situations may arise: \hfill \break
%(H.1) $0 \le P^{c}_t(x_2, -A, -C) < P^{h}_t(x_1, A, C)$, \hfill \break
%(H.2) $P^{c}_t(x_1, A, C) \le 0 < P^{h}_t(x_2, -A, -C)$, \hfill \break
%(H.3) $P^{c}_t(x_2, -A, -C) < P^{h}_t(x_1, A, C) \le 0$, \hfill \break
%and, symmetrically, \hfill \break
%(C.1) $0 \le P^{h}_t(x_1, A, C) < P^{c}_t(x_2, -A, -C)$, \hfill \break
%(C.2) $P^{h}_t(x_1, A, C) \le 0 < P^{c}_t(x_2, -A, -C)$, \hfill \break
%(C.3) $P^{h}_t(x_1, A, C) < P^{c}_t(x_2, -A, -C) \le 0$.

Before analyzing each situation, let us recall that the cash flows of a contract $(A,C)$ are invariably considered
from the perspective of the hedger, so it makes sense to observe that the counterparty faces the cash flows given by
$(-A,-C)$. Consequently, in case (H.1), we may say that the hedger is the seller of $(A,C)$ and the counterparty is the buyer of
$(-A,-C)$,  but the counterparty is not willing to pay the amount demanded by the hedger.
In case (H.2), both parties are willing to be sellers of the contract, meaning in practice that the hedger is ready to sell $(A,C)$
and the counterparty is willing to sell $(-A,-C)$.
Finally, case (H.3) refers to the situation the counterparty is willing to be the seller of $(-A,-C)$, whereas the hedger can now be seen as a buyer of $(A,C)$, but he is not willing to pay the price that is needed by the counterparty to replicate the contract.

Assume that the market model is arbitrage-free for both parties in the sense of Definition \ref{arbitrage using nettled wealth}.
Then in all three cases, (H.1)--(H.3), any $\G_t$-measurable random variable $P^f_t$ satisfying
\be
P^f_t \in \big[ P^{c}_t(x_{2},-A,-C), P^{h}_t (x_{1},A,C) \big]
\ee
can be considered to be a {\it fair price} for both the hedger and his counterparty, in the sense that
a  bilateral transaction done at the price $P^f_t$ will not generate an arbitrage opportunity for neither of them.
Hence the interval $[P^{c}(x_{2},-A,-C),P^{h}_t (x_{1},A,C)]$ represents the range of fair prices of the
contract $(A,C)$ for both parties, as seen from the perspective of the hedger (a special case of this interval was dubbed the {\it arbitrage-band} by Bergman \cite{B-1995}).

\bd
The  $\G_t$-measurable interval ${\cal R}^f_t (x_1,x_2) := \big[ P^{c}_t(x_{2},-A,-C), P^{h}_t (x_{1},A,C) \big]$
%\bde
%{\cal R}^f_t (x_1,x_2) := \big[ P^{c}_t(x_{2},-A,-C), P^{h}_t (x_{1},A,C) \big]
%\ede
is called the {\it range of fair bilateral prices} at time $t$ of an OTC contract $(A,C)$ between the hedger and the counterparty.
\ed

Although the analysis for the cases (C.1)--(C.3) can be done analogously, the financial interpretation and conclusions are quite different. In case (C.1), the hedger is willing to be the seller of $(A,C)$ and the counterparty is willing to be the buyer
and he is ready to pay even more than it is requested by the hedger. In case (C.2), both parties are disposed to be buyers
at their respective prices, meaning that each party is ready to pay a positive premium to another.  Finally, in case (C.3), the counterparty is willing to be the seller, whereas the hedger can now be seen as a buyer of $(A,C)$
and he is ready to pay more than it is demanded by the counterparty. Hence for any $\G_t$-measurable random variable $P^p_t$ satisfying
\be \label{inequality to prove}
P^p_t \in \big[ P^{h}_t(x_{1},A,C), P^{c}_t (x_{2},-A,-C) \big]
\ee
can be seen as a price at which both parties would be disposed to make the deal with each other. Note that, unless $P^{h}_t(x_{1},A,C)=P^{c}_t (x_{2},-A,-C)$, the price $P^p_t$ is not a fair bilateral price, in the sense explained above, since an arbitrage opportunity arises for at least one party involved when an OTC contract $(A,C)$ is traded between them at the price $P^p_t$. This simple observation motivates the following definition.

\bd  \label{defarba}
Assume that the inequality $P^{h}_t(x_{1},A,C) \ne P^{c}_t (x_{2},-A,-C)$ holds. Then the $\G_t$-measurable interval
${\cal R}^p_t (x_1,x_2) := \big[ P^{h}_t(x_{1},A,C), P^{c}_t (x_{2},-A,-C) \big]$
%\bde
%{\cal R}^p_t (x_1,x_2) := \big[ P^{h}_t(x_{1},A,C), P^{c}_t (x_{2},-A,-C) \big]
%\ede
is called the {\it range of bilaterally profitable prices} at time $t$ of an OTC contract $(A,C)$ between the hedger
and the counterparty.
\ed

Note that in our discussion, we dealt in fact with at least three different concepts of arbitrage:

\noindent {\bf (A.1)} the classic definition of an arbitrage opportunity that may arise by trading in primary assets,

\noindent {\bf (A.2)} an arbitrage opportunity associated with a long hedged position in some contract combined with
a short unhedged position in the same contract; in that case, the contract's price at time 0 is considered
to be exogenously given by the market, that is, it is driven by the law of demand and supply (see Definition \ref{arbitrage using nettled wealth} of an arbitrage in regard to a given contract),

\noindent {\bf (A.3)} an arbitrage opportunity related to the fact that the hedger and the counterparty may require different
premia to implement their respective replicating strategies; if this kind of an arbitrage opportunity arises, then
it is simultaneously available to both parties involved in an OTC contract with a price negotiated between them (as in Definition \ref{defarba}).

%Note that in our discussion, we deal in fact with at least three different concepts of arbitrage: \hfill \break
%{\bf (A.1)} an arbitrage opportunity that may arise by trading in primary assets only, \hfill \break
%{\bf (A.2)} an arbitrage opportunity associated with a long hedged position in some contract combined with
%a short unhedged position in the same contract; in that case, the contract's price at time 0 is considered
%to be exogenously given by the market (see Definition \ref{arbitrage using nettled wealth}), \hfill \break
%{\bf (A.3)} an arbitrage opportunity related to the fact that the hedger and the counterparty may require different
%premia to implement their respective replicating strategies; if this kind of an arbitrage opportunity arises,
%it is simultaneously available to both parties in an OTC contract (as in cases (C.1)--(C.3)).

Note that in case (C.2) an immediate {\it reselling arbitrage opportunity} arises for a third party, that is, a trader who could simultaneously `purchase' a contract from one party and `resell' to the other. Specifically, if $P^{h}_t (x_{1},A,C) \leq 0$ and $P^{c}_t (x_{2},-A,-C)> 0$, then a third party can make a deal with the hedger to face $(-A,-C)$ and receive $-P^{h}_t (x_{1},A,C)\ge 0$  and, at the same time, enter the contract with the counterparty to face $(A,C)$ and get $P^{c}_t (x_{2},-A,-C) > 0$. This offsetting strategy produces an immediate profit of $P^{c}_t (x_{2},-A,-C)-P^{h}_t (x_{1},A,C)>0$ for the third party.

\section{Pricing BSDEs and Replicating Strategies}    \label{sect4}
%%%%%%%%%%%%%%%%%%%%%%%%%%%%%%%%%%%%%%%%%%%%%%%%%%%%%%%%%%%%%%%%%%%%
%%%%%%%%%%%%%%%%%%%%%%%%%%%%%%%%%%%%%%%%%%%%%%%%%%%%%%%%%%%%%%%%%%%%

Our next aim is to show that the hedger's and counterparty's prices and their replicating strategies can be found
by solving suitable BSDEs. For this purpose, we will use some auxiliary
results on BSDEs driven by multi-dimensional continuous martingales (see \cite{NR3} and the references therein).
%For the reader's convenience,
%they are collected in the appendix (see Section \ref{sect7.2}).
In Propositions \ref{hedger ex-dividend price} and \ref{counterparty ex-dividend price}, we will show that
if $x_{1}x_{2}\ge0$, then the prices $P^{h}_t (x_{1},A,C)$ and $P^{c}_t (x_{2},-A,-C)$ are given by the
solutions of two BSDEs that are driven by either the common $(\PT^l , \gg)$-local martingale $\wt S^{l,\textrm{cld}}$
(when $x_1 \geq 0,\, x_2 \geq 0$) or the common $(\PT^b , \gg)$-local martingale $\wt S^{b,\textrm{cld}}$
(when $x_1 \leq 0,\, x_2 \leq 0$).
By contrast, when the inequality $x_{1}x_{2}<0$ holds, say $x_1>0$ and $x_2<0$,
then the prices are associated with solutions to the two BSDEs driven by $\wt S^{l,\textrm{cld}}$ and $\wt S^{b,\textrm{cld}}$, respectively. Therefore, to find the range of fair (or profitable) bilateral prices using the comparison theorem for BSDEs, we will first need to find a suitable variant of the pricing BSDE for both parties, which will be driven by a common continuous local martingale (see Section \ref{sect5.3}).

%%%%%%%%%%%%%%%%%%%%%%%%%%%%%%%%%%%%%%%%%%%%%%%%%%%%%%%%
\subsection{Modeling of Risky Assets}        \label{sect4.1}
%%%%%%%%%%%%%%%%%%%%%%%%%%%%%%%%%%%%%%%%%%%%%%%%%%%%%%%%

To show the existence of a solution to the pricing BSDE, we need to complement Assumptions \ref{assumption for lending cumulative dividend price} and \ref{assumption for borrowing cumulative dividend price} by imposing specific conditions on the underlying market model. For any $d\times d$ matrix $m$, the norm of $m$  is given by $\norm m \norm^2 :=\text{Tr}(m m^{\ast})$.
In the next assumption, the superscript $k$ stands for either $l$ or $b$.

\bhyp \label{additional assumption for lending cumulative dividend price}
We postulate that: \hfill \break
(i) the process $\wt S^{,\textrm{cld}}$ is a continuous, square-integrable, $(\PT^k , \gg)$-martingale and has the predictable representation property (PRP) with respect to the filtration $\gg$ under~$\PT^k$, \hfill \break
(ii) there exists an $\mathbb{R}^{d\times d}$-valued, $\gg$-adapted process $m^{k}$ such that
\be \label{vfvf1}
\langle \wt S^{k,\textrm{cld}}\rangle_{t}=\int_{0}^{t}m^{k}_{u}(m_{u}^{k})^{\ast}\,du
\ee
where $m^{k}(m^{k})^{\ast}$ is invertible and there exists a constant $K_m>0$ such that, for all $t\in[0,T]$,
\be \label{mmc2}
\norm m^k_{t}\norm+\norm(m^k_{t}(m^k_{t})^{\ast})^{-\frac{1}{2}}\norm\leq K_m,
\ee
(iii) the price processes $S^{i},\, i=1,2,\ldots,d$ of risky assets are bounded.
\ehyp

Note that condition \eqref{mmc2} means that the process $m^l$ satisfies Assumption 5.2 in \cite{NR3}.
Although the postulate that the prices $S^{i},\, i=1,2,\ldots,d$ are bounded can be seen as a quite reasonable real-world requirement, it is rarely satisfied in commonly used financial models, including the classic Black and Scholes model. It is also worth noting that condition \eqref{mmc2} could appear to be too restrictive. In order to relax this condition, we will need to impose stronger conditions on the process $\langle \wt S^{l,\textrm{cld}}\rangle$. Specifically, we define the matrix-valued process
$\mathbb{S}$
\[
\mathbb{S}_{t}:=
\begin{pmatrix}
S^{1}_{t} & 0 & \ldots & 0 \\
0 & S^{2}_{t} & \ldots & 0 \\
\vdots & \vdots & \ddots & \vdots\\
0 & 0 & \ldots & S^{d}_{t}
\end{pmatrix}.
\]
\bd \label{ellipticity}
We say that $\gamma $ satisfies the {\it ellipticity} condition if there exists a constant $\Lambda>0$
\be \label{elli}
\sum_{i,j=1}^{d}\left(\gamma_{t}\gamma^{\ast}_{t}\right)_{ij}a_{i}a_{j}\ge \Lambda|a|^{2}=\Lambda a^{\ast}a,\ \text{ for all } a\in\mathbb{R}^{d} \text{ and } t\in[0,T].
\ee
\ed

We consider the following assumption, which should be seen as an alternative to Assumption
\ref{additional assumption for lending cumulative dividend price}. Once again, the superscript $k$
is equal either to $l$ or $b$.

\bhyp \label{changed assumption for lending cumulative dividend price}
We postulate that: \hfill \break
(i) the process $\wt S^{k,\textrm{cld}}$ is a continuous, square-integrable, $(\PT^k , \gg)$-martingale
and has the PRP with respect to the filtration $\gg$ under~$\PT^k$, \hfill \break
(ii) equality \eqref{vfvf1} holds with the $\gg$-adapted process $m^{k}$ such that $m^{k}(m^{k})^{\ast}$ is invertible and satisfies $m^{k}(m^{k})^{\ast}=\mathbb{S}\gamma\gamma^{\ast}\mathbb{S}$ where a $d$-dimensional square matrix $\gamma$ of $\gg$-adapted processes satisfies the ellipticity condition~\eqref{elli}.
\ehyp

\brem \label{remark for diffusion type market model 1}
We will show that Assumption \ref{additional assumption for lending cumulative dividend price} or \ref{changed assumption for lending cumulative dividend price} can be easily met when the prices of risky assets are given by the diffusion-type model.
For example, we may assume that each risky asset $S^i,\, i=1, 2, \dots  ,d$ has the ex-dividend price dynamics under $\P$ given by
\[
dS^i_t = S^i_t \bigg( \mu^i_t \, dt + \sum_{j=1}^d \sigma^{ij}_t \, dW^j_t \bigg), \quad S^i_0>0 ,
\]
or, equivalently, the $d$-dimensional process $S=(S^{1},\ldots,S^{d})^{\ast}$ satisfies
\[
dS_{t}=\mathbb{S}_{t}(\mu_{t} \, dt+\sigma_{t} \, dW_{t})
\]
where $W = (W^1, \dots , W^d)^{\ast}$ is the $d$-dimensional Brownian motion, $\mu = (\mu^1, \dots , \mu^d)^{\ast}$ is an $\mathbb{R}^d$-valued, $\ff^W$-adapted process, $\sigma = [\sigma^{ij}]$ is a $d$-dimensional square matrix of $\ff^W$-adapted processes satisfying the {\it ellipticity} condition.
We now set $\gg = \ff^W$ and we recall that the $d$-dimensional Brownian motion $W$ enjoys the predictable representation
property with respect to its natural filtration $\ff^W$; this property is shared by the process $\wt{W}$ defined
 \eqref{wtw2}.

Assuming that the corresponding dividend processes are given by $\pA^i_t = \int_0^t \kappa^i_{u} S^i_{u} \, du $,
we obtain
\bde
d\wt S_t^{i,l,\textrm{cld}}=(B_{t}^{l})^{-1}\big(dS^i_t+ d\pA^i_t-r_{t}^{l}S_{t}^{i}\, dt\big)
=(B_{t}^{l})^{-1}S^{i}_{t}\bigg( \big(\mu^i_t+\kappa^i_t-r_{t}^{l}\big)\,dt+ \sum_{j=1}^d \sigma^{ij}_t \, dW^j_t \bigg).
\ede
If we denote $S^{l,\textrm{cld}}=(S^{1,l,\textrm{cld}},\ldots,S^{d,l,\textrm{cld}})^{\ast}$ and
$\mu+\kappa-r^{l}= (\mu^1+\kappa^{1}-r^{l}, \dots , \mu^d+\kappa^{d}-r^{l})^{\ast}$, then
\bde
d\wt S_t^{l,\textrm{cld}}=(B_{t}^{l})^{-1}\mathbb{S}_{t}\Big( \big(\mu_t+\kappa_t-r_{t}^{l}\big)\,dt+ \sigma_t \, dW_t \Big).
\ede
We set $a_{t}:=\sigma_t^{-1}(\mu_t+\kappa_t-r^{l}_{t})$ for all $t \in [0,T]$ and we define the probability
measure $\PT^{l}$ on $(\Omega , {\cal F}^W_T)$ by
\be
\frac{d\PT^{l}}{d\P}=\exp\bigg\{-\int_{0}^{T}a_{t}\, dW_{t}-\frac{1}{2}\int_{0}^{T}|a_{t}|^{2}\, dt\bigg\}.
\ee
Then $\PT^{l}$ is equivalent to $\P$ and, from the Girsanov theorem, the process $\widetilde{W}:=(\widetilde{W}^{1},\widetilde{W}^{2},\ldots,\widetilde{W}^{d})^{\ast}$ is a Brownian motion under $\PT^{l}$, where
\be \lab{wtw2}
d\widetilde{W}_{t} := dW_{t}+a_{t}\, dt = dW_{t}+\sigma_t^{-1}(\mu_t+\kappa_t-r^{l}_{t})\,dt.
\ee
It is clear that under $\PT^{l}$
\bde
d\wt S_t^{l,\textrm{cld}} = (B_{t}^{l})^{-1}\mathbb{S}_{t} \sigma_t \, d\wt{W}_t .
\ede
Therefore, if the processes $\mu,\, \sigma$ and $\kappa$ are bounded, then the processes
$\wt S^{i,l,\textrm{cld}},\, i=1,2, \dots ,d$ are continuous, square-integrable, $(\PT^{l}, \gg)$-martingales. Furthermore,
the quadratic variation of $\wt S^{l,\textrm{cld}}$ equals
\bde
\langle \wt S^{l,\textrm{cld}}\rangle_{t}=\int_{0}^{t}m^{l}_{u}(m_{u}^{l})^{\ast}\,du
\ede
where $m^{l}(m^{l})^{\ast}=\mathbb{S} \gamma \gamma \mathbb{S}$ and $\gamma :=(B^{l})^{-1}\sigma$. Obviously, $m^{l}(m^{l})^{\ast}$ is invertible and thus Assumption \ref{changed assumption for lending cumulative dividend price} is satisfied.
Moreover, if the processes $S^{i}$, $i=1,2, \dots ,d$ are bounded, then the process $m^{l}$ satisfies condition \eqref{mmc2}
and thus Assumption \ref{additional assumption for lending cumulative dividend price} is valid with $k=l$.
\erem

\subsection{Hedgers's Prices and Replicating Strategies}    \label{sect4.3}
%%%%%%%%%%%%%%%%%%%%%%%%%%%%%%%%%%%%%%%%%%%%%%%%%%%%%%%%%%%%%%%%%%%%%%%%%%%%

From now on, we work under the standing assumption that $Q_t = t$ for every $t \in [0,T]$ in Assumption 3.1 in \cite{NR3}
and thus also in all results in Sections 3--5 of \cite{NR3} (in particular, in the definition
of the norm for the space $\wHlamd$). Note that this postulate is consistent with either of Assumptions  \ref{additional assumption for lending cumulative dividend price} and \ref{changed assumption for lending cumulative dividend price}.
Moreover, we henceforth postulate that the processes $r^{l},\,r^{b}$ and $r^{i,b}$ for $i=1,2,\ldots,d$
are nonnegative and bounded.

The following result describes the prices and hedging strategies for the hedger.
Recall that  $A^{C}:=A+C+F^{C}$ and
\bde
F^{C}_t =-\int_0^t (B^{\pCc,b}_{u})^{-1}\pC_{u}^+\, dB^{\pCc,b}_{u} + \int_0^t (B^{\pCc,l}_{u})^{-1} \pC_{u}^- \, dB^{\pCc,l}_{u}.
\ede
%We also denote
%\bde
%A^{C,l}_t= \int_{(0,t]}(\Blr_{u})^{-1}\, dA^C_{u}, \quad A^{C,b}_t= \int_{(0,t]}(\Bbr_{u})^{-1}\, dA^C_{u}.
%\ede
Following \cite{NR3}, but with $Q_t=t$, we denote by $\wHzerd $ the subspace of all $\mathbb{R}^{d}$-valued, $\gg$-adapted processes $X$ with
\be \label{defhh}
|X|_{\wHzerd}^{2}:=\EP \bigg[ \int_{0}^{T}\|X_{t}\|^{2}\,dt \bigg] <\infty .
\ee
Also, let $\widehat{L}^{2}_{0}$ stand for the space of all real-valued, $\mathcal{G}_{T}$-measurable
random variables $\eta$ such that $|\eta|_{\widehat{L}^2_{0}}^{2}=\EP (\eta^{2})<\infty $.

\bd
A contract $(A,C)$ is {\it admissible under} $\PT^l$ if the process $A^{C,l}$ belongs to $\wHzero$
and the random variable $A^{C,l}_{T}$ belongs to $\widehat{L}^{2}_{0}$ under~$\PT^l$.
A contract $(A,C)$ is {\it admissible under} $\PT^b$ if the process $A^{C,b}$ belongs to $\wHzero$
and the random variable $A^{C,b}_{T}$ belongs to $\widehat{L}^{2}_{0}$ under~$\PT^b$.
\ed

\bp \label{hedger ex-dividend price}
(i) Let either Assumption \ref{additional assumption for lending cumulative dividend price} or Assumption \ref{changed assumption for lending cumulative dividend price} with $k=l$ be satisfied. Then for any real number $x \geq 0$ and any
contract $(A,C)$ admissible under $\PT^l$, the hedger's ex-dividend price satisfies $P^{h}(x,A,C) =  \Blr (Y^{h,l,x} - x)  - C$ where $(Y^{h,l,x}, Z^{h,l,x})$ is the unique solution to the BSDE
\begin{equation}\label{BSDE with positive x for hedger}
\left\{
\begin{array}
[c]{l}
dY^{h,l,x}_t = Z^{h,l,x,\ast}_t \, d \wt S^{l,{\textrm{cld}}}_t
+\wt{f}_l \big(t, Y^{h,l,x}_t, Z^{h,l,x}_t \big)\, dt + dA^{C,l}_t, \medskip\\
Y^{h,l,x}_T=x.
\end{array}
\right.
\end{equation}
The unique replicating strategy equals
$\phi = \big(\xi^1,\dots ,\xi^d, \psi^{l}, \psi^{b},\psi^{1,b},\dots ,\psi^{d,b}, \etab, \etal\big)$
where, for every $t\in[0,T]$ and $i=1,2,\ldots,d,$
\bde
\xi^i_{t}= Z^{h,l,x,i}_{t},  \quad
\psi^{i,b}_t =  -(\Bibr_t)^{-1} (\xi^i_t S^i_t)^+, \quad
\etab_t =-  (B^{\pCc,b}_t)^{-1}\pC_t^+, \quad
\etal_t =(B^{\pCc,l}_t)^{-1} \pC_t^-,
\ede
and
\begin{align*}
&\psi^{l}_t = (\Blr_t)^{-1} \Big( \Blr_tY^{h,l,x}_{t} + \sumik_{i=1}^d ( \xi^i_t S^i_t )^- \Big)^+, \\
&\psi^{b}_t = - (\Bbr_t)^{-1} \Big(\Blr_tY^{h,l,x}_{t}+ \sumik_{i=1}^d ( \xi^i_t S^i_t )^- \Big)^-.
\end{align*}
(ii) Let either Assumption \ref{additional assumption for lending cumulative dividend price} or Assumption \ref{changed assumption for lending cumulative dividend price} with $k=b$ be satisfied. Then for any real number $x\leq0$ and any contract $(A,C)$ admissible under $\PT^b$, the hedger's ex-dividend price satisfies $P^{h}(x,A,C) =  \Bbr (Y^{h,b,x} - x)- C$ where $(Y^{h,b,x}, Z^{h,b,x})$ is the unique solution to the BSDE
\begin{equation}\label{BSDE with negative x for hedger}
\left\{
\begin{array}
[c]{l}
dY^{h,b,x}_t = Z^{h,b,x,\ast}_t \, d \wt S^{b,{\textrm{cld}}}_t
+\wt{f}_b \big(t, Y^{h,b,x}_t, Z^{h,b,x}_t \big)\, dt + dA^{C,b}_t, \medskip\\
Y^{h,b,x}_T=x.
\end{array}
\right.
\end{equation}
The unique replicating strategy equals $\phi = \big(\xi^1,\dots ,\xi^d, \psi^{l},
\psi^{b},\psi^{1,b},\dots ,\psi^{d,b}, \etab, \etal\big)$ where, for every $t\in[0,T]$ and $i=1,2,\ldots,d,$
\bde
\xi^i_{t}= Z^{h,b,x,i}_{t},  \quad
\psi^{i,b}_t =  -(\Bibr_t)^{-1} (\xi^i_t S^i_t)^+, \quad
\etab_t =-  (B^{\pCc,b}_t)^{-1}\pC_t^+, \quad
\etal_t =(B^{\pCc,l}_t)^{-1} \pC_t^-,
\ede
and
\begin{align*}
&\psi^{l}_t = (\Blr_t)^{-1} \Big( \Bbr_tY^{h,b,x}_{t} + \sumik_{i=1}^d ( \xi^i_t S^i_t )^- \Big)^+, \\
&\psi^{b}_t = - (\Bbr_t)^{-1} \Big(\Bbr_tY^{h,b,x}_{t}+ \sumik_{i=1}^d ( \xi^i_t S^i_t )^- \Big)^-.
\end{align*}
\ep

\begin{proof}
Assume first that $x\ge0$. Then, from Theorems 4.1 and 5.1 in \cite{NR3}, we know that if either Assumption \ref{additional assumption for lending cumulative dividend price} or Assumption \ref{changed assumption for lending cumulative dividend price} with
$k=l$ is satisfied and $A^{C,l}\in\wHzero $ and $A^{C,l}_{T}\in \widehat{L}^{2}_{0}$ under $\PT^l$, then BSDE (\ref{BSDE with positive x for hedger}) has a unique solution $(Y^{h,l,x}, Z^{h,l,x})$. Thus, from Proposition 5.2 in \cite{BR-2014} we obtain $P^{h}(x,A,C) =  \Blr(Y^{h,l,x} - x)  - C$.  Moreover, the replicating strategy $\varphi$ can be constructed uniquely, as was explained in Section \ref{sect2.3}. In view of Remark 5.3 in \cite{BR-2014}, an analogous analysis can be done when the initial endowment satisfies $x\leq0$ .
\end{proof}

\brem
Let us give some comments on the uniqueness of a replicating strategy in Proposition \ref{hedger ex-dividend price}. We only consider the case when $x\ge0$, since similar arguments apply to the case $x\leq0$. The uniqueness of the solution of BSDE (\ref{BSDE with positive x for hedger}) means that if $(Y^{1},Z^{1})$ and $(Y^{2},Z^{2})$ are two solutions of BSDE (\ref{BSDE with positive x for hedger}), then
\be \label{lol}
\mathbb{E}_{\wt{\mathbb{P}}_{l}}\bigg[\int_{0}^{T}|Y^{1}_{t}-Y^{2}_{t}|^{2}\, dt
+\int_{0}^{T}\|(m^{l}_{t})^{\ast}Z^{1}_{t}-(m^{l}_{t})^{\ast}Z^{2}_{t}\|^{2}\, dt\bigg]=0.
\ee
Under Assumption \ref{additional assumption for lending cumulative dividend price}, there exists a constant $k_m$ such that $\norm (m^{l}_{t})^{\ast}\norm \ge k_m$ % (see Section \ref{section special case})
and a constant $K$ such that $\norm \mathbb{S}_{t}\norm\leq K$. Therefore, from \eqref{lol} we deduce that
\be \label{lol1}
\mathbb{E}_{\wt{\mathbb{P}}_{l}}\bigg[\int_{0}^{T}\|\mathbb{S}_{t}Z^{1}_{t}-\mathbb{S}_{t}Z^{2}_{t}\|^{2}\, dt\bigg]=0.
\ee
Under Assumption \ref{changed assumption for lending cumulative dividend price} with $k=l$, we have that $m^{l}(m^{l})^{\ast}=\mathbb{S}\gamma\gamma^{\ast}\mathbb{S}$ and thus
\[
\mathbb{E}_{\wt{\mathbb{P}}_{l}}\bigg[\int_{0}^{T}\|(m^{l}_{t})^{\ast}(Z^{1}_{t}-Z^{2}_{t}) \|^{2}\, dt\bigg]=
\mathbb{E}_{\wt{\mathbb{P}}_{l}}\bigg[\int_{0}^{T}(Z^{1}_{t}
-Z^{2}_{t})^{\ast}\mathbb{S}\gamma\gamma^{\ast}\mathbb{S}(Z^{1}_{t}-Z^{2}_{t})\, dt\bigg].
\]
Since $\gamma$ satisfies the ellipticity condition, there exists a constant $\Lambda>0$ such that
\[
\mathbb{E}_{\wt{\mathbb{P}}_{l}}\bigg[\int_{0}^{T}(Z^{1}_{t}-Z^{2}_{t})^{\ast}\mathbb{S}_{t}\gamma
\gamma^{\ast}\mathbb{S}_{t}(Z^{1}_{t}-Z^{2}_{t})\, dt\bigg]
\ge \Lambda \, \mathbb{E}_{\wt{\mathbb{P}}_{l}}\bigg[\int_{0}^{T}\|\mathbb{S}_{t}Z^{1}_{t}-\mathbb{S}_{t}Z^{2}_{t}\|^{2}\, dt\bigg].
\]
We conclude that under either of Assumptions \ref{additional assumption for lending cumulative dividend price} and
\ref{changed assumption for lending cumulative dividend price} with $k=l$, equality \eqref{lol1} is satisfied by any two
solutions of BSDE (\ref{BSDE with positive x for hedger}).

From the above arguments and the structure of the replicating strategy (see Proposition \ref{hedger ex-dividend price})
\[
\phi = \big(\xi^1,\dots ,\xi^d, \psi^{l}, \psi^{b},\psi^{1,b},\dots ,\psi^{d,b}, \etab, \etal\big),
\]
we know that the uniqueness is in the sense of equivalence with respect to $\P_{l}\otimes \Leb$.
Moreover,  for $\xi=(\xi^1,\dots ,\xi^d)^{\ast}$, $\psi^{l}$, $\psi^{b}$ and $\psi=(\psi^{1,b},\dots ,\psi^{d,b})^{\ast}$ the uniqueness holds in the following norm
\[
\|\varphi\|:=\mathbb{E}_{\wt{\mathbb{P}}_{l}}\bigg[\int_{0}^{T}\|\mathbb{S}_{t}\xi_{t}\|^{2}dt+\int_{0}^{T}(|\psi^{l}_{t}|^{2}+|\psi^{b}_{t}|^{2})dt
+\int_{0}^{T}\|\psi_{t}\|^{2}dt\bigg].
\]
\erem

%%%%%%%%%%%%%%%%%%%%%%%%%%%%%%%%%%%%%%%%%%%%%%%%%%%%%%%%%%%%%%%%%%%%%%%%%%%%%%%%
\subsection{Counterparty's Prices and Replicating Strategies}   \label{sect4.4}
%%%%%%%%%%%%%%%%%%%%%%%%%%%%%%%%%%%%%%%%%%%%%%%%%%%%%%%%%%%%%%%%%%%%%%%%%%%%%%%%

Let us first observe that, in view of Assumption \ref{additional assumption for collateral account}, we have $(-A)^{-C}=-A^{C}$.
Using Definition \ref{remark for counterparty's ex-dividend price}, one can prove the following result for the counterparty.

\bp \label{counterparty ex-dividend price}
Let the assumptions of part (i) or (ii) in Proposition \ref{hedger ex-dividend price} be satisfied for $x \ge 0$
and $x \le 0$, respectively.
Then the counterparty's ex-dividend price satisfies, for every $t \in [0,T)$,
\bde
P^{c}_t (x,-A,-C) =-\left(\Blr_t (Y^{c,l,x}_t - x)+C_t\right)\I_{\{x\ge0\}}-\left(\Bbr_t (Y^{c,b,x}_t - x)+C_t\right)\I_{\{x\leq0\}}
\ede
where $(Y^{c,l,x}, Z^{c,l,x})$ and $(Y^{c,b,x}, Z^{c,b,x})$ is respectively the unique solution to the BSDE
\begin{equation}\label{BSDE with positive x for counterparty}
\left\{
\begin{array}
[c]{l}
dY^{c,l,x}_t = Z^{c,l,x,\ast}_t \, d \wt S^{l,{\textrm{cld}}}_t
+\wt{f}_l \big(t, Y^{c,l,x}_t, Z^{c,l,x}_t \big)\, dt - dA^{C,l}_t, \medskip\\
Y^{c,l,x}_T=x,
\end{array}
\right.
\end{equation}
and
\begin{equation}\label{BSDE with negative x for counterparty}
\left\{
\begin{array}
[c]{l}
dY^{c,b,x}_t = Z^{c,b,x,\ast}_t \, d \wt S^{b,{\textrm{cld}}}_t
+\wt{f}_b \big(t, Y^{c,b,x}_t, Z^{c,b,x}_t \big)\, dt - dA^{C,b}_t, \medskip\\
Y^{c,b,x}_T=x.
\end{array}
\right.
\end{equation}
The unique replicating strategy equals
$\phi = \big(\xi^1,\dots ,\xi^d, \psi^{l}, \psi^{b},\psi^{1,b},\dots ,\psi^{d,b}, \etab, \etal\big)$
where, for every $t\in[0,T]$ and $i=1,2,\ldots,d,$
\bde
\xi_{t}= Z^{c,l,x}_{t}\I_{\{x\ge0\}}+Z^{c,b,x}_{t}\I_{\{x\leq0\}},\
\psi^{i,b}_t =  -(\Bibr_t)^{-1} (\xi^i_t S^i_t)^+, \
\etab_t =-  (B^{\pCc,b}_t)^{-1}\pC_t^-, \
\etal_t =(B^{\pCc,l}_t)^{-1} \pC_t^+ ,
\ede
and
\begin{align*}
&\psi^{l}_t = (\Blr_t)^{-1} \Big( \Blr_tY^{c,l,x}_{t}\I_{\{x\ge0\}}+\Bbr_tY^{c,b,x}_{t}\I_{\{x\leq0\}}+ \sumik_{i=1}^d ( \xi^i_t S^i_t )^- \Big)^+, \\
&\psi^{b}_t = - (\Bbr_t)^{-1} \Big(\Blr_tY^{c,l,x}_{t}\I_{\{x\ge0\}}+\Bbr_tY^{c,b,x}_{t}\I_{\{x\leq0\}}+ \sumik_{i=1}^d ( \xi^i_t S^i_t )^- \Big)^-.
\end{align*}
\ep

%%%%%%%%%%%%%%%%%%%%%%%%%%%%%%%%%%%%%%%%%%%%%%%%%%%%%%%%%%%%%%%%%
%%%%%%%%%%%%%%%%%%%%%%%%%%%%%%%%%%%%%%%%%%%%%%%%%%%%%%%%%%%%%%%%%
\section{Properties of Arbitrage Prices} \label{sect5}
%%%%%%%%%%%%%%%%%%%%%%%%%%%%%%%%%%%%%%%%%%%%%%%%%%%%%%%%%%%%%%%%%
%%%%%%%%%%%%%%%%%%%%%%%%%%%%%%%%%%%%%%%%%%%%%%%%%%%%%%%%%%%%%%%%%

Recall that we consider the special case of an exogenous margin account with rehypothecated cash collateral.
The exogenous property implies that $C$ does not depend on a strategy $\varphi$ and the value of the strategy.
We denote the initial endowment of the hedger (resp., counterparty) by $x_{1}$ (resp., $x_{2}$). We will examine
 the pricing and hedging problems for both parties in the following situations:

\noindent  -- the initial endowments satisfy $x_{1}\ge0$ and $x_{2}\ge0$,

\noindent  -- the initial endowments satisfy $x_{1}\leq0$ and $x_{2}\leq0$,

\noindent -- the initial endowments satisfy $x_{1} x_{2}\leq 0$.

%\noindent {\bf Case A} -- the initial endowments satisfy $x_{1}\ge0$ and $x_{2}\ge0$,
%
%\noindent {\bf Case B} -- the initial endowments satisfy $x_{1}\leq0$ and $x_{2}\leq0$,
%
%\noindent {\bf Case C} -- the initial endowments satisfy $x_{1} x_{2}\leq 0$.

Our goal is to establish inequalities for unilateral prices for each of the three above-mentioned cases
(see Propositions \ref{inequality proposition for both positive initial wealth}, \ref{inequality proposition for both negative initial wealth} and \ref{inequality proposition for positive negative initial wealth}, respectively) and thus also to derive the ranges of fair bilateral prices. In the last case, that is, when $x_{1} x_{2}\leq 0$ we also examine the properties of the class of {\it monotone} contracts (see Section \ref{sect5.3.2}). Finally, we study the monotonicity of unilateral prices with respect to the initial endowment and we derive the pricing PDE in the Markovian framework.

%%%%%%%%%%%%%%%%%%%%%%%%%%%%%%%%%%%%%%%%%%%%%%%%%%%%%%%%%%%%%%%%%%%%
\subsection{Initial Endowments of Equal Signs}  \label{sect5.1}
%%%%%%%%%%%%%%%%%%%%%%%%%%%%%%%%%%%%%%%%%%%%%%%%%%%%%%%%%%%%%%%%%%%%

We first assume that both parties have positive initial endowments, that is,
$x_{1}\ge0$ and $x_{2}\ge0$. The proof of Proposition \ref{inequality proposition for both positive initial wealth}
is provided in Nie and Rutkowski \cite{NR3} (see Theorem 5.2 in \cite{NR3}), where a suitable comparison
theorem for BSDEs driven by a multi-dimensional martingale is also proven.

\bp \label{inequality proposition for both positive initial wealth}
Let either Assumption \ref{additional assumption for lending cumulative dividend price} or Assumption \ref{changed assumption for lending cumulative dividend price} with $k=l$ hold. If $x_{1}\ge0$ and $x_{2}\ge0$, then for any contract $(A,C)$ admissible under $\PT^l$ we have, for all $t\in[0,T]$,
\be \label{eqq1}
P^{c}_t (x_{2},-A,-C)\leq P^{h}_t (x_{1},A,C),  \quad \PT^l-\aass ,
\ee
so that the range of fair bilateral prices ${\cal R}^f_t (x_1,x_2)$ is non-empty almost surely.
\ep

%%%%%%%%%%%%%%%%%%%%%%%%%%%%%%%%%%%%%%%%%%%%%%%%%%%%%%%%%%%%%%%%%%%%%
% \subsection{Negative Initial Endowments}  \label{sect5.2}
%%%%%%%%%%%%%%%%%%%%%%%%%%%%%%%%%%%%%%%%%%%%%%%%%%%%%%%%%%%%%%%%%%%%%

In the second step, we postulate that both parties have positive initial endowments, that is,
$x_{1}\leq0$ and $x_{2}\leq0$. As was explained in Remark \ref{rates assumption for negative wealth}, we now need assume that $\rbb_t\leq\ribb_t$ for $i=1,2,\ldots,d$. The proof of Proposition \ref{inequality proposition for both negative initial wealth}  is postponed to the appendix.

\bp \label{inequality proposition for both negative initial wealth}
Let either Assumption \ref{additional assumption for lending cumulative dividend price} or Assumption
\ref{changed assumption for lending cumulative dividend price} with $k=b$ hold.
If $x_{1}\leq 0,\, x_{2}\leq 0$ and $r^b \leq r^{i,b}$ for $i=1,2,\ldots,d$,
then for any contract $(A,C)$ admissible under $\PT^b$ we have, for all $t\in[0,T]$,
\be \label{eqq2}
P^{c}_t (x_{2},-A,-C)\leq P^{h}_t (x_{1},A,C),\quad \PT^b-\aass ,
\ee
so that the range of fair bilateral prices ${\cal R}^f_t (x_1,x_2)$ is non-empty almost surely.
\ep

%%%%%%%%%%%%%%%%%%%%%%%%%%%%%%%%%%%%%%%%%%%%%%%%%%%%%%%%%%%%%%%%%%%%%%%
\subsection{Initial Endowments of Opposite Signs}  \label{sect5.3}
%%%%%%%%%%%%%%%%%%%%%%%%%%%%%%%%%%%%%%%%%%%%%%%%%%%%%%%%%%%%%%%%%%%%%%%

We now consider the case when the initial endowments of the two parties have opposite signs, specifically, we postulate that $x_{1}\ge0$ and $x_{2}\leq0$. From Propositions \ref{hedger ex-dividend price} and \ref{counterparty ex-dividend price}, it follows that $P^{h} (x_{1},A,C) =  \Blr (Y^{h,l,x_{1}}- x_{1})  - C$ where $(Y^{h,l,x_{1}}, Z^{h,l,x_{1}})$ is the unique solution of the BSDE
\be
\left\{
\begin{array}
[c]{l}
dY^{h,l,x_{1}}_t = Z^{h,l,x_{1},\ast}_t \, d \wt S^{l,{\textrm{cld}}}_t
+\wt{f}_l \big(t, Y^{h,l,x_{1}}_t, Z^{h,l,x_{1}}_t \big)\, dt + dA^{C,l}_t, \medskip\\
Y^{h,l,x_{1}}_T=x_{1},\nonumber
\end{array}
\right.
\ee
and $P^{c} (x_{2},-A,-C) =-(\Bbr (Y^{c,b,x_{2}}- x_{2})+C)$ where $(Y^{c,b,x_{2}}, Z^{c,b,x_{2}})$ is the unique solution of the BSDE
\be
\left\{
\begin{array}
[c]{l}
dY^{c,b,x_{2}}_t = Z^{c,b,x_{2},\ast}_t \, d \wt S^{b,{\textrm{cld}}}_t
+\wt{f}_b \big(t, Y^{c,b,x_{2}}_t, Z^{c,b,x_{2}}_t \big)\, dt - dA^{C,b}_t, \medskip\\
Y^{c,b,x_{2}}_T=x_{2}.\nonumber
\end{array}
\right.
\ee
Note that the BSDE for $Y^{h,l,x_{1}}$ is driven by the $(\PT^l , \gg)$-local martingale $\wt S^{l,\textrm{cld}}$, but the BSDE for $Y^{c,b,x_{2}}$ is driven by the $(\PT^b , \gg)$-local martingale $\wt S^{b,\textrm{cld}}$. We will now attempt to find another probability measure $\PT $ equivalent to $\mathbb{P}$ and a $(\PT , \gg)$-local martingale $\wt S^{\textrm{cld}}$ such that the BSDEs related to $P^{h}(x_{1},A,C)$ and $P^{c}(x_{2},A,C)$ are both driven by a common $(\PT^b , \gg)$-local martingale $\wt S^{\textrm{cld}}$.
If we denote $\widetilde{Y}^{h,l,x_{1}}=  \Blr (Y^{h,l,x_{1}} - x_{1})$, then $P^{h} (x_{1},A,C) = \widetilde{Y}^{h,l,x_{1}}-C$. In view of (\ref{cumulative dividend risk asset price1}) and (\ref{drift function lending}),
we obtain
\be
\begin{array}
[c]{ll}
&d\widetilde{Y}^{h,l,x_{1}}_t=-x_{1}\,d\Blr_t+Y^{h,l,x_{1}}_t\,d\Blr_t+\Blr_t\,dY^{h,l,x_{1}}_t\medskip\\
&=\mbox{} -x_{1}\rll_t\Blr_t\,dt+\rll_t\Blr_tY^{h,l,x_{1}}_t\,dt+\Blr_tZ^{h,l,x_{1},\ast}_t\,d \wt S^{l,{\textrm{cld}}}_t
+\Blr_t\wt{f}_l \big(t, Y^{h,l,x_{1}}_t, Z^{h,l,x_{1}}_t \big)\,dt +dA^C_t\medskip\\
&=\mbox{} -x_{1}\rll_t\Blr_t\,dt+\rll_t\Blr_tY^{h,l,x_{1}}_t\,dt+\sum_{i=1}^dZ^{h,l,x_{1},i}_t\left(dS^i_t - \rll_t S^i_t \, dt + d\pA^i_t\right)
+f_l \big(t, \Blr_tY^{h,l,x_{1}}_t, Z^{h,l,x_{1}}_t \big)\,dt\medskip\\
&\quad \mbox{} -\rll_t\Blr_tY^{h,l,x_{1}}_t\,dt +dA^C_t\medskip\\
&=\mbox{} -x_{1}\rll_t\Blr_t\,dt+\sum_{i=1}^dZ^{h,l,x_{1},i}_t\left(dS^i_t - \rll_t S^i_t \, dt + d\pA^i_t\right)\medskip\\
&\quad \mbox{} + \sum_{i=1}^d \rll_t Z^{h,l,x_{1},i}_{t} S^i_t\,dt
- \sum_{i=1}^d \ribb_t(Z^{h,l,x_{1},i}_{t}S^i_t )^+\, dt +\rll_t \Big( \Blr_t Y^{h,l,x_{1}}_t + \sum_{i=1}^d ( Z^{h,l,x_{1},i}_{t} S^i_t )^- \Big)^+dt\medskip\\
&\quad \mbox{} - \rbb_t \Big( \Blr_t Y^{h,l,x_{1}}_t+ \sum_{i=1}^d ( Z^{h,l,x_{1},i}_{t} S^i_t )^- \Big)^-dt+dA^C_t \medskip\\
&= \mbox{} - x_{1}\rll_t\Blr_t\,dt+\sum_{i=1}^dZ^{h,l,x_{1},i}_t\left(dS^i_t+ d\pA^i_t\right)+g(t,\Blr_t Y^{h,l,x_{1}}_t,Z^{h,l,x_{1}}_t)\,dt+dA^C_t
\nonumber
\end{array}
\ee
where
\be\label{drift driver for positive and negative initial wealth}
g(t,y,z)=-\sumik_{i=1}^d \ribb_t(z^{i}S^i_t )^++\rll_t \Big(y+ \sumik_{i=1}^d ( z^{i}S^i_t )^- \Big)^+
- \rbb_t \Big( y+ \sumik_{i=1}^d ( z^{i}S^i_t )^- \Big)^-.
\ee
Upon denoting $\widetilde{Z}^{h,l,x_{1}}= Z^{h,l,x_{1}}$, we obtain
\bde
d\widetilde{Y}^{h,l,x_{1}}_t =\sumik_{i=1}^d\widetilde{Z}^{h,l,x_{1},i}_t\left(dS^i_t+ d\pA^i_t\right)-x_{1}\rll_t\Blr_tdt+g\big(t, \widetilde{Y}^{h,l,x_{1}}_t+x_{1}\Blr_t,\widetilde{Z}^{h,l,x_{1}}_t\big)\,dt+dA^C_t .
\ede
Similarly, if we denote $\widetilde{Y}^{c,b,x_{2}}= -\Bbr (Y^{c,b,x_{2}} - x_{2})$ and  $\widetilde{Z}^{c,b,x_{2}}= -Z^{c,b,x_{2}} $, then the counterparty's price equals $P^{c}(x_{2},A,C) = \widetilde{Y}^{c,b,x_{2}}-C$. In view of (\ref{cumulative dividend risk asset price2}), (\ref{drift function borrowing}) and (\ref{drift driver for positive and negative initial wealth}), we obtain
\be
\begin{array}
[c]{ll}
&d\widetilde{Y}^{c,b,x_{2}}_t=x_{2}\, d\Bbr_t-Y^{c,b,x_{2}}_t\,d\Bbr_t-\Bbr_t\, dY^{c,b,x_{2}}_t\medskip\\
&=x_{2}\rbb_t\Bbr_t\,dt-\rbb_t\Bbr_tY^{c,b,x_{2}}_t\,dt-\Bbr_tZ^{c,b,x_{2},\ast}_t\,d \wt S^{b,{\textrm{cld}}}_t
-\Blr_t\wt{f}_b \big(t, Y^{c,b,x_{2}}_t, Z^{c,b,x_{2}}_t \big)\,dt +dA^C_t\medskip\\
&=x_{2}\rbb_t\Bbr_t\,dt-\rbb_t\Bbr_tY^{c,b,x_{2}}_t\,dt-\sum_{i=1}^dZ^{c,b,x_{2},i}_t\left(dS^i_t - \rbb_t S^i_t \, dt + d\pA^i_t\right)
-f_b \big(t, \Bbr_tY^{c,b,x_{2}}_t, Z^{c,b,x_{2}}_t \big)\,dt\medskip\\
&\quad\mbox{} +\rbb_t\Bbr_tY^{h,l,x_{1}}_t\,dt+dA^C_t\medskip\\
&=x_{2}\rbb_t\Bbr_t\,dt-\sum_{i=1}^dZ^{c,b,x_{2},i}_t\left(dS^i_t+ d\pA^i_t\right)-g(t,\Blr_t Y^{c,b,x_{2}}_t,Z^{c,b,x_{2}}_t)\,dt+dA^C_t\medskip\\
&=\sum_{i=1}^d\widetilde{Z}^{c,b,x_{2},i}_t\left(dS^i_t+ d\pA^i_t\right)+x_{2}\rbb_t\Bbr_t\,dt-g\big(t, -\widetilde{Y}^{c,b,x_{2}}_t+x_{2}\Bbr_t,-\widetilde{Z}^{c,b,x_{2}}_t\big)\,dt+dA^C_t.\nonumber
\nonumber
\end{array}
\ee
The following assumptions are motivated by Assumptions \ref{additional assumption for lending cumulative dividend price}
and \ref{changed assumption for lending cumulative dividend price}, respectively.

\bhyp \label{assumption for artifical cumulative dividend price}
We postulate that: \hfill \break
(i) there exists a probability measure $\PTb $ equivalent to $\P$ such that the
processes $\wt S^{i,\textrm{cld}},\, i=1,2, \dots ,d$
given by \be\label{auxiliary processes}
d\wt S^{i,\textrm{cld}}_t = dS^i_t + d\pA^i_t - \beta^{i}_{t}S_{t}^{i}\,dt
\ee
for some $\gg$-adapted processes $\beta^{i}$ satisfying $r^{b}\leq\beta^{i}\leq r^{i,b}$,
are continuous, square-integrable $(\PTb , \gg)$-martingales, and have the PRP with respect to the filtration $\gg$ under $\PTb$,\hfill \break
(ii) there exists an $\mathbb{R}^{d\times d}$-valued, $\gg$-adapted process $m$ such that
\be\label{auxiliary processes quadratic variation}
\langle \wt S^{\textrm{cld}}\rangle_{t}=\int_{0}^{t}m_{u}m_{u}^{\ast}\,du,
\ee
where $m(m)^{\ast}$ is invertible and there exists a constant $K_m>0$ such that, for all $t\in[0,T]$,
\be \label{mmc2x}
\norm m_{t}\norm+\norm(m_{t}(m_{t})^{\ast})^{-\frac{1}{2}}\norm\leq K_m,
\ee
(iii) the price processes $S^{i},\, i=1,2,\ldots,d$ of risky assets are bounded.
\ehyp

\bhyp \label{changed assumption for artifical cumulative dividend price}
We postulate that: \hfill \break
(i) there exists a probability measure $\PTb $ equivalent to $\P$ such that the processes $\wt S^{i,\textrm{cld}},\, i=1,2, \dots ,d$ given by (\ref{auxiliary processes}) are $(\PTb , \gg)$-continuous square integrable martingales, and have the PRP with respect to the filtration $\gg$ under $\PTb$,\hfill \break
(ii) condition (\ref{auxiliary processes quadratic variation}) holds with the $\gg$-adapted process $m$ such that $mm^{\ast}$ is invertible and given by $mm^{\ast}= \mathbb{S}\gamma\gamma^{\ast}\mathbb{S}$
where a $d$-dimensional square matrix $\gamma$ of $\gg$-adapted processes satisfies the ellipticity condition~\eqref{elli}.
\ehyp

\brem \label{remark for different assumptions for cumulative dividend price}
Recall that
\bde
d\wt S^{i,l,{\textrm{cld}}}_t=(\Blr_t)^{-1}\left(dS^i_t - \rll_t S^i_t \, dt + d\pA^i_t\right), \quad
d\wt S^{i,b,{\textrm{cld}}}_t=(\Bbr_t)^{-1}\left(dS^i_t - \rbb_t S^i_t \, dt + d\pA^i_t\right).
\ede
Since $r^l$ is non-negative and bounded, we obtain $C_{0}<(\Blr)^{-1}<1$ for some constant $C_{0}$.
Then Assumption \ref{assumption for artifical cumulative dividend price} with $\beta^{i}=\rll$ is equivalent to Assumption \ref{assumption for lending cumulative dividend price}. Similar comments apply to other assumptions. We mention that, from
Proposition \ref{remark for non-arbitrage model} and $r^{b}\leq\beta^{i}\leq r^{i,b}$, we know that under Assumption \ref{assumption for artifical cumulative dividend price} or Assumption \ref{changed assumption for artifical cumulative dividend price}, our partial netting model is arbitrage-free with respect to any contract for both the hedger and counterparty with $x_{1},x_{2}\in\mathbb{R}$.
\erem

\brem
The above assumptions can be easily satisfied for the diffusion-type market model similarly to the one in Remark \ref{remark for diffusion type market model 1}; the details are left to the reader.
\erem

\bd
We say that $(A,C)$ is {\it admissible under} $\PTb$ when $A^C \in\wHzero $ and $A^C_T\in \widehat{L}^{2}_{0}$ under $\PTb$.
\ed

For $g$ given by (\ref{drift driver for positive and negative initial wealth}),  let us define
\bde
g^{h}(t,x,y,z):=\sum_{i=1}^dz^{i}_t\beta^{i}_{t}S_{t}^{i}+(-x\rll_t\Blr_t+g(t,y+x\Blr_t,z))\I_{\{x\ge0\}}+(-x\rbb_t\Bbr_t+g(t, y+x\Bbr_t,z))\I_{\{x\leq0\}}
\ede
and
\bde
g^{c}(t,x,y,z):=\sum_{i=1}^dz^{i}_t\beta^{i}_{t}S_{t}^{i}+(x\rll_t\Blr_t-g(t,-y+x\Blr_t,-z))\I_{\{x\ge0\}}+(x\rbb_t\Bbr_t-g(t, -y+x\Bbr_t,-z))\I_{\{x\leq0\}}.
\ede

The next result is a counterpart of Propositions  \ref{hedger ex-dividend price} and \ref{counterparty ex-dividend price}.
In view of the discussion at the beginning of this subsection, Proposition \ref{general pricing proposition} is a rather
straightforward consequence of Theorem 4.1 in \cite{NR3} and thus its proof is omitted.

\bp \label{general pricing proposition}
Let either Assumption \ref{assumption for artifical cumulative dividend price} or Assumption \ref{changed assumption for artifical cumulative dividend price} be valid. Consider an arbitrary contract $(A,C)$ admissible
under $\PTb$. Then $P^{h}(x_{1},A,C)=\widetilde{Y}^{h,x_{1}}-C$ and $P^{c}(x_{2},-A,-C)=\widetilde{Y}^{c,x_{2}}-C$
where  $(\widetilde{Y}^{h,x_{1}},\widetilde{Z}^{h,x_{1}})$ is the unique solution of the BSDE
\be\label{artifical BSDE for hedger}
\left\{
\begin{array}
[c]{ll}
d\widetilde{Y}^{h,x_{1}}_t =\widetilde{Z}^{h,x_{1},\ast}_t\, d\wt S_t^{\textrm{cld}}+g^{h}(t,x_{1}, \widetilde{Y}^{h,x_{1}}_t,\widetilde{Z}^{h,x_{1}}_t)\,dt+dA^C_t,\medskip\\
\widetilde{Y}^{h,x_{1}}_T=0,
\end{array}
\right.
\ee
and $(\widetilde{Y}^{c,x_{2}},\widetilde{Z}^{c,x_{2}})$ is the unique solution of the BSDE
\be\label{artifical BSDE for counterparty}
\left\{
\begin{array}
[c]{ll}
d\widetilde{Y}^{c,x_{2}}_t=\widetilde{Z}^{c,x_{2},\ast}_t\, d\wt S_t^{\textrm{cld}}+g^{c}(t, x_{2}, \widetilde{Y}^{c,x_{2}}_t,\widetilde{Z}^{c,x_{2}}_t)\,dt+dA^C_t,\medskip\\
\widetilde{Y}^{c,x_{2}}_T=0,
\end{array}
\right.
\ee
Moreover, the unique replicating strategy for the hedger equals
$\phi = \big(\xi^1,\dots ,\xi^d, \psi^{l}, \psi^{b},\psi^{1,b},\dots ,\psi^{d,b}, \etab, \etal\big)$ where for every $t\in[0,T]$ and $i=1,2,\ldots,d$
\bde
\xi^i_{t}= \widetilde{Z}^{h,x_{1},i}_{t},  \quad \psi^{i,b}_t =  -(\Bibr_t)^{-1} (\xi^i_t S^i_t)^+, \quad \etab_t =-  (B^{\pCc,b}_t)^{-1}\pC_t^+, \quad \etal_t =(B^{\pCc,l}_t)^{-1} \pC_t^-,
\ede
and
\bde
\begin{array}
[c]{ll}
&\psi^{l}_t = (\Blr_t)^{-1} \Big( \widetilde{Y}^{h,x_{1}}_{t}+x_{1}\Blr_t\I_{\{x_{1}\ge0\}}+x_{1}\Bbr_t\I_{\{x_{1}\leq0\}}+ \sum_{i=1}^d ( \xi^i_t S^i_t )^- \Big)^+, \medskip\\
&\psi^{b}_t = - (\Bbr_t)^{-1} \Big( \widetilde{Y}^{h,x_{1}}_{t}+x_{1}\Blr_t\I_{\{x_{1}\ge0\}}+x_{1}\Bbr_t\I_{\{x_{1}\leq0\}}+ \sum_{i=1}^d ( \xi^i_t S^i_t )^- \Big)^-.
\end{array}
\ede
The unique replicating strategy for the counterparty equals
$\phi = \big(\xi^1,\dots ,\xi^d, \psi^{l}, \psi^{b},\psi^{1,b},\dots ,\psi^{d,b}, \etab, \etal\big)$ where for every $t\in[0,T]$ and $i=1,2,\ldots,d$
\bde
\xi^i_{t}=-\widetilde{Z}^{c,x_{2},i}_{t},  \quad \psi^{i,b}_t =  -(\Bibr_t)^{-1} (\xi^i_t S^i_t)^+, \quad \etab_t =-  (B^{\pCc,b}_t)^{-1}\pC_t^-, \quad \etal_t =(B^{\pCc,l}_t)^{-1} \pC_t^+.
\ede
and
\bde
\begin{array}
[c]{ll}
&\psi^{l}_t = (\Blr_t)^{-1} \Big(-\widetilde{Y}^{c,x_{2}}_{t}+x_{2}\Blr_t\I_{\{x_{2}\ge0\}}+x_{2}\Bbr_t\I_{\{x_{2}\leq0\}}+ \sum_{i=1}^d ( \xi^i_t S^i_t )^- \Big)^+, \medskip\\
&\psi^{b}_t = - (\Bbr_t)^{-1} \Big( -\widetilde{Y}^{c,x_{2}}_{t}+x_{2}\Blr_t\I_{\{x_{2}\ge0\}}+x_{2}\Bbr_t\I_{\{x_{2}\leq0\}}+ \sum_{i=1}^d ( \xi^i_t S^i_t )^- \Big)^-.
\end{array}
\ede
\ep

One can check that $g^{h}(t,x,0,0)=g^{c}(t,x,0,0)=0$ for all  $x\in\mathbb{R}$.
Consider any contract $(A,C)$ admissible under $\PTb$. If, in addition, $A^{C}$ is a decreasing process, then for any $x_{1}, x_{2}\in\mathbb{R}$, $\widetilde{Y}^{h,x_{1}}\ge0$ and $\widetilde{Y}^{c,x_{2}}\ge0$, where  $(\widetilde{Y}^{h,x_{1}},\widetilde{Z}^{h,x_{1}})$ is the unique solution of BSDE (\ref{artifical BSDE for hedger}) and $(\widetilde{Y}^{c,x_{2}},\widetilde{Z}^{c,x_{2}})$ is the unique solution of BSDE
(\ref{artifical BSDE for counterparty}). Consequently, $P^{h}(x_{1},A,C) \ge -C$ and $P^{c}(x_{2},-A,-C) \ge -C$.
If the process $A^{C}$ is increasing, then for any $x_{1}, x_{2}\in\mathbb{R}$ we have $\widetilde{Y}^{h,x_{1}}\leq0$ and $\widetilde{Y}^{c,x_{2}}\leq0$, so that $P^{h}(x_{1},A,C) \leq -C$ and $P^{c}(x_{2},-A,-C) \leq -C$.

\bex
In Example \ref{European call option}, we considered a contract $(A,C)$ with $A_t = p \, \I_{[0,T]}(t) + X \I_{[T]}(t)$  and $C=0$. Let us first assume that $X\leq0$; for instance, for a European call option $X = - (S^i_T-K)^+$ and for a European put option $X = - (K-S^i_T)^+$. Then, obviously, the process $A^{C}-A_0 =A - A_0$ is decreasing. Then, for any $x\in\mathbb{R}$, both $P^{h}_t (x,A,C)$ and $P^{c}_t (x,-A,-C)$ are positive, meaning that the hedger is the seller and the counterparty is the buyer.
Similarly, if $X\ge0$, for instance, for a European call option $X = (S^i_T-K)^+$  and for European put option $X = (K-S^i_T)^+$,
then, for any $x\in\mathbb{R}$, both $P^{h}_t (x,A,C)$ and $P^{c}_t (x,-A,-C)$ are negative, meaning that the counterparty
is the seller and the hedger is the buyer. Needless to say that such properties of unilateral options prices were expected.
\eex

%%%%%%%%%%%%%%%%%%%%%%%%%%%%%%%%%%%%%%%%%%%%%%%%%%%%%%%%%%%%%
\subsubsection{General Contracts}       \label{sect5.3.1}
%%%%%%%%%%%%%%%%%%%%%%%%%%%%%%%%%%%%%%%%%%%%%%%%%%%%%%%%%%%%%

Since $\wt S^{i,\textrm{cld}},\, i=1,2, \dots ,d$ are $(\PTb , \gg)$-local martingales under Assumption \ref{assumption for artifical cumulative dividend price}, we can apply the comparison theorem to BSDEs (\ref{artifical BSDE for hedger}) and (\ref{artifical BSDE for counterparty}) in order to establish the following proposition (for the proof, see Section \ref{sect7}).

\bp \label{inequality proposition for positive negative initial wealth}
Let either Assumption \ref{assumption for artifical cumulative dividend price} or Assumption \ref{changed assumption for artifical cumulative dividend price} be valid. Assume that $x_{1}\ge0,\, x_{2}\leq0$ and $r^b \leq r^{i,b}$ for $i=1,2,\ldots,d$.
Then the following statement are valid. \hfill \break
(i) If $x_{1}x_{2}=0$, then for an arbitrary contract $(A,C)$ admissible under $\PTb$ and all $t\in[0,T]$,
\be \label{eqnew2x}
P^{c}_t (x_{2},-A,-C)\leq P^{h}_t (x_{1},A,C),  \quad \PTb-\aass ,
\ee
so that the range of fair bilateral prices ${\cal R}^f_t (x_1,x_2)$ is non-empty almost surely. \hfill \break
(ii) Let $r^{l}$ and $r^{b}$ be deterministic and satisfy $r^{l}_{t}<r^{b}_{t}$ for all $t\in[0,T]$.
Then inequality \eqref{eqnew2x} holds for all contracts $(A,C)$ admissible under $\PTb$ and all $t\in[0,T]$
if and only if $x_{1}x_{2}=0$.
\ep

Notice that if $x_{1}x_{2}=0$ then, from Propositions \ref{inequality proposition for both positive initial wealth} and \ref{inequality proposition for both negative initial wealth}, we know that the desired inequality holds under respective assumptions. However, in the current proposition, we are working under Assumption \ref{assumption for artifical cumulative dividend price}, so that it is not clear whether the pricing inequality still holds.

Finally, for the case $x_{1}\leq0,\, x_{2}\ge0$, one can show show that the following
result is valid. The proof of Proposition \ref{inequality proposition for negative positive initial wealth}
is similar to that of Proposition \ref{inequality proposition for positive negative initial wealth}
and thus it is omitted.

\bp \label{inequality proposition for negative positive initial wealth}
Let either Assumption \ref{assumption for artifical cumulative dividend price} or
Assumption \ref{changed assumption for artifical cumulative dividend price} be valid.
Assume that $x_{1}\leq0,\, x_{2}\ge0$ and $r^b \leq r^{i,b}$ for $i=1,2,\ldots,d$. Then the following statements
are valid.  \hfill \break
(i) If $x_{1}x_{2}=0$, then for an arbitrary contract $(A,C)$ admissible
under $\PTb$ and all $t\in[0,T]$
\be \label{eqnne1}
P^{c}_t (x_{2},-A,-C)\leq P^{h}_t (x_{1},A,C),  \quad \PTb-\aass ,
\ee
so that the range of fair bilateral prices ${\cal R}^f_t (x_1,x_2)$ is non-empty almost surely. \hfill \break
(ii) Let $r^{l}$ and $r^{b}$ be deterministic and satisfy $r^{l}_{t}<r^{b}_{t}$ for all $t\in[0,T]$.
Then  inequality \eqref{eqnne1} holds for all contracts $(A,C)$ admissible under $\PTb$ and all $t\in[0,T]$
if and only if $x_{1}x_{2}=0$.
\ep

\brem \label{remark for the proof under different assumptions for cumulative dividend price}
Under the assumptions of Proposition \ref{inequality proposition for positive negative initial wealth}, one can prove that Propositions \ref{inequality proposition for both positive initial wealth} and \ref{inequality proposition for both negative initial wealth} hold under $\PTb$, that is, if $x_{1}x_{2}\ge0$, then for any $t\in[0,T]$,
\be \label{desew}
P^{c}_t (x_{2},-A,-C)\leq P^{h}_t (x_{1},A,C),  \quad \PTb-\aass
\ee
Indeed, using similar arguments as in the proof of Proposition \ref{inequality proposition for positive negative initial wealth},
one can show that
\bde
g^{h}(t,x_{1},y,z)-g^{c}(t,x_{2},y,z)
\leq \sumik_{i=1}^{d} |z^{i}S_{t}^{i}| \big( (r^{l}_{t}-r^{i,b}_{t})\I_{\{x_{1}\ge0,x_{2}\ge0\}}+
(r^{b}_{t}-r^{i,b}_{t})\I_{\{x_{1}\leq0,x_{2}\leq0\}} \big)
\leq 0.
\ede
Consequently, using Proposition \ref{general pricing proposition} and the comparison theorem for BSDEs, we obtain the desired
inequality \eqref{desew}.
\erem

%%%%%%%%%%%%%%%%%%%%%%%%%%%%%%%%%%%%%%%%%%%%%%%%%%%%%%%%%%%%%%%%%
\subsubsection{Contracts with Monotone Cash Flows}   \label{sect5.3.2}
%%%%%%%%%%%%%%%%%%%%%%%%%%%%%%%%%%%%%%%%%%%%%%%%%%%%%%%%%%%%%%%%%

If $x_{1}x_{2}< 0$ then, from the proof of Proposition \ref{inequality proposition for positive negative initial wealth}, we know that for some contracts $(A,C)$ we have $P^{c}_{\wh{t}} (x_{2},-A,-C)\ge P^{h}_{\wh{t}} (x_{1},A,C)$ for some $\wh{t}\in[0,T]$. The next theorem show that, for some special classes of contracts $(A,C)$, the inequality  $P^{c}_t (x_{2},-A,-C)\leq P^{h}_t (x_{1},A,C)$ is satisfied for all $t\in[0,T]$.

\begin{theorem} \label{special contract pricing}
Let either Assumption \ref{assumption for artifical cumulative dividend price} or
Assumption \ref{changed assumption for artifical cumulative dividend price} be valid.
If $x_{1}\ge0,\, x_{2}\leq0$, then for an arbitrary contract $(A,C)$ admissible under $\PTb$ and such that the process $A^{C}$ is decreasing on $(0,T]$ we have, for every $t\in[0,T]$,
\bde
P^{c}_t (x_{2},-A,-C)\leq P^{h}_t (x_{1},A,C),  \quad \PTb-\aass ,
\ede
so that the range of fair bilateral prices ${\cal R}^f_t (x_1,x_2)$ is non-empty almost surely.
\end{theorem}

\begin{proof}
From Proposition \ref{general pricing proposition} and the inequalities $x_{1}\ge0,\, x_{2}\leq0$, we know
that for any contract $(A,C)$ admissible under $\PTb$ we have $P^{h}(x_{1},A,C)=\widetilde{Y}^{h,l,x_{1}}-C$
and $P^{c}(x_{2},-A,-C)=\widetilde{Y}^{c,b,x_{2}}-C$
where  $(\widetilde{Y}^{h,l,x_{1}},\widetilde{Z}^{h,l,x_{1}})$ is the unique solution of BSDE (\ref{artifical BSDE for positive hedger}) and $(\widetilde{Y}^{c,b,x_{2}},\widetilde{Z}^{c,b,x_{2}})$ is the unique solution of BSDE (\ref{artifical BSDE for negative counterparty}). Since
\bde
g^{h,l}(t,x_{1},0,0)=g^{c,b}(t,x_{2},0,0)=0,
\ede
and $A^{C}$ is a decreasing process then, from the comparison theorem for BSDEs, we have $\widetilde{Y}^{h,l,x_{1}}\ge 0$ and $\widetilde{Y}^{c,b,x_{2}}\ge0$. Since $x_{1}\ge0$, BSDE (\ref{artifical BSDE for hedger}) becomes
\be\label{artifical BSDE 1 for positive hedger}
\left\{
\begin{array}
[c]{ll}
d\widetilde{Y}^{h,l,x_{1}}_t =\widetilde{Z}^{h,l,x_{1},\ast}_td\wt S_t^{\textrm{cld}}+\widetilde{g}^{h,l}(t, x_{1}, \widetilde{Y}^{h,l,x_{1}}_t,\widetilde{Z}^{h,l,x_{1}}_t)\,dt+dA^C_t,\medskip\\
\widetilde{Y}^{h,l,x_{1}}_T=0,
\end{array}
\right.
\ee
where the generator $\tilde{g}^{h,l}(t,x,y,z)$ does not depend on $x$ and it is given by (recall that $\zzb = z^i S^i_t$)
\bde
\tilde{g}^{h,l}(t,x,y,z):=\sumik_{i=1}^d \beta^{i}_{t} \zzb
-\sumik_{i=1}^d \ribb_t( \zzb )^++\rll_ty+ \rll_t\sumik_{i=1}^d ( \zzb )^- .
\ede
Since
\bde
\begin{array}
[c]{ll}
g^{c,b}(t,x,y,z)&=\sum_{i=1}^d \beta^{i}_{t} \zzb
+\sum_{i=1}^d \ribb_t(- \zzb )^++xr^{b}_{t}B_{t}^{b}\medskip\\
&\quad-\rll_t \Big(-y+xB_{t}^{b}+\sum_{i=1}^d (- \zzb )^-\Big)^{+}+\rbb_t \Big(-y+xB_{t}^{b}+\sum_{i=1}^d (-\zzb )^- \Big)^{-}\medskip\\
&\ge \sum_{i=1}^d \beta^{i}_{t} \zzb
+\sum_{i=1}^d \ribb_t(- \zzb )^++xr^{b}_{t}B_{t}^{b}-\rbb_t \Big(-y+xB_{t}^{b}+\sum_{i=1}^d (-\zzb )^- \Big)\medskip\\
&=\sum_{i=1}^d \beta^{i}_{t} \zzb
+\sum_{i=1}^d \ribb_t(- \zzb )^++\rbb_ty-\rbb_t\sum_{i=1}^d (- \zzb  )^-,
\end{array}
\ede
we obtain
\bde
\begin{array}
[c]{ll}
&\tilde{g}^{h,l}(t,x,\widetilde{Y}_t^{h,l,x_{1}},\widetilde{Z}_t^{h,l,x_{1}})
-g^{c,b}(t,x,\widetilde{Y}_t^{h,l,x_{1}},\widetilde{Z}_t^{h,l,x_{1}})\medskip\\
&\leq(r_t^{l}-r_t^{b})\widetilde{Y}_t^{h,l,x_{1}}-\sum_{i=1}^d \ribb_t|\widetilde{Z}_t^{h,l,x_{1},i}S^i_t|
+\rll_t \sum_{i=1}^d (\widetilde{Z}_t^{h,l,x_{1},i}S^i_t )^- +\rbb_t \sum_{i=1}^d (-\widetilde{Z}_t^{h,l,x_{1},i}S^i_t )^-\medskip\\
&=(r_t^{l}-r_t^{b})\widetilde{Y}_t^{h,l,x_{1}}
+ \sum_{i=1}^d(\rll_t-\ribb_t)(\widetilde{Z}_t^{h,l,x_{1},i}S^i_t )^-
+\sum_{i=1}^d (\rbb_t-\ribb_t)(-\widetilde{Z}_t^{h,l,x_{1},i}S^i_t )^-\leq 0.
\end{array}
\ede
The comparison theorem for BSDEs gives $\widetilde{Y}_t^{h,l,x_{1}}\ge\widetilde{Y}_t^{c,b,x_{2}}$ and thus $P^{c}_t (x_{2},-A,-C)\leq P^{h}_t (x_{1},A,C),$ $\PTb$-a.s. for every $t\in[0,T]$.
\end{proof}

\begin{theorem} \label{special contract pricing 1}
Let either Assumption \ref{assumption for artifical cumulative dividend price} or
Assumption \ref{changed assumption for artifical cumulative dividend price} be valid.
If $x_{1}\le 0,\, x_{2}\ge 0$, then for an arbitrary contract $(A,C)$ admissible under $\PTb$ and such that the process $A^{C}$ is increasing on $(0,T]$ we have, for every $t\in[0,T]$,
\bde
P^{c}_t (x_{2},-A,-C)\leq P^{h}_t (x_{1},A,C),  \quad \PTb-\aass ,
\ede
so that the range of fair bilateral prices ${\cal R}^f_t (x_1,x_2)$ is non-empty almost surely.
\end{theorem}

\begin{proof}
From Proposition \ref{general pricing proposition} and $x_{1}\leq0,\, x_{2}\ge0$, we know that
$P^{h} (x_{1},A,C)=\widetilde{Y}^{h,b,x_{1}}-C$ and $P^{c} (x_{2},-A,-C)=\widetilde{Y}^{c,l,x_{2}}-C$
where  $(\widetilde{Y}^{h,b,x_{1}},\widetilde{Z}^{h,b,x_{1}})$ is the unique solution of the following BSDE
\be\label{artifical BSDE for negative hedger}
\left\{
\begin{array}
[c]{ll}
d\widetilde{Y}^{h,b,x_{1}}_t =\widetilde{Z}^{h,b,x_{1},\ast}_td\wt S_t^{\textrm{cld}}+g^{h,b}(t, x_{1}, \widetilde{Y}^{h,b,x_{1}}_t,\widetilde{Z}^{h,b,x_{1}}_t)\,dt+dA^C_t,\medskip\\
\widetilde{Y}^{h,b,x_{1}}_T=0,
\end{array}
\right.
\ee
where
\bde
g^{h,b}(t,x,y,z):=\sumik_{i=1}^dz^{i}_t\beta^{i}_{t}S_{t}^{i}
-xr^{b}_{t}B_{t}^{b}+g(t,y+xB_{t}^{b},z).
\ede
and $(\widetilde{Y}^{c,l,x_{2}},\widetilde{Z}^{c,l,x_{2}})$ is the unique solution of the following BSDE
\be\label{artifical BSDE for positive counterparty}
\left\{
\begin{array}
[c]{ll}
d\widetilde{Y}^{c,l,x_{2}}_t =\widetilde{Z}^{c,l,x_{2},\ast}_td\wt S_t^{\textrm{cld}}+g^{c,l}(t, x_{2}, \widetilde{Y}^{c,l,x_{2}}_t,\widetilde{Z}^{c,l,x_{2}}_t)\,dt+dA^C_t,\medskip\\
\widetilde{Y}^{c,l,x_{2}}_T=0,
\end{array}
\right.
\ee
where
\bde
g^{c,l}(t,x,y,z):=\sumik_{i=1}^dz^{i}_t\beta^{i}_{t}S_{t}^{i}
+xr^{l}_{t}B_{t}^{l}-g(t,-y+xB_{t}^{l},-z).
\ede
Since
\bde
g^{h,b}(t,x_{1},0,0)=g^{c,l}(t,x_{2},0,0)=0,
\ede
and the process $A^{C}$ is assumed to be increasing, from Theorem 3.3 in \cite{NR3}, % \ref{comparison theorem 2x},
we obtain $\widetilde{Y}^{h,b,x_{1}}\leq 0$ and $\widetilde{Y}^{c,l,x_{2}}\leq0$. Therefore, since $x_{2}\ge0$, we see that $g^{c,l}(t,x,y,z)$ does not depend on $x$ and
\bde
g^{c,l}(t,x,y,z)=\tilde{g}^{c,l}(t,x,y,z):=\sumik_{i=1}^d\beta^{i}_{t} \zzb
+\sumik_{i=1}^d \ribb_t(- \zzb )^++\rll_t y-\rll_t \sumik_{i=1}^d (- \zzb )^-
\ede
where, as usual, we denote $\zzb = z^i S^i_t$. Furthermore, the function $g^{h,b}(t,x,y,z)$ satisfies
\bde
\begin{array}
[c]{ll}
g^{h,b}(t,x,y,z)&=\sum_{i=1}^d \beta^{i}_{t} \zzb
-\sum_{i=1}^d \ribb_t( \zzb )^+-xr^{b}_{t}B_{t}^{b}\medskip\\
&\quad \mbox{}+\rll_t \Big(y+xB_{t}^{b}+\sum_{i=1}^d (\zzb )^-\Big)^{+}-\rbb_t \Big(y+xB_{t}^{b}+\sum_{i=1}^d (\zzb )^- \Big)^{-}\medskip\\
&\leq \sum_{i=1}^dz^{i}_t\beta^{i}_{t}S_{t}^{i}
-\sum_{i=1}^d \ribb_t( \zzb  )^+-xr^{b}_{t}B_{t}^{b}+\rbb_t \Big(y+xB_{t}^{b}+\sum_{i=1}^d (\zzb )^- \Big)\medskip\\
&=\sum_{i=1}^d \beta^{i}_{t} \zzb
-\sum_{i=1}^d \ribb_t( \zzb )^++\rbb_t y + \rbb_t \sum_{i=1}^d (- \zzb )^-
\end{array}
\ede
and thus
\bde
\begin{array}
[c]{ll}
&g^{h,b}(t,x,\widetilde{Y}_t^{h,b,x_{1}},\widetilde{Z}_t^{h,b,x_{1}})
-\widetilde{g}^{c,l}(t,x,\widetilde{Y}_t^{h,b,x_{1}},\widetilde{Z}_t^{h,b,x_{1}})\medskip\\
&\leq(r_t^{b}-r_t^{l})\widetilde{Y}_t^{h,b,x_{1}}-\sum_{i=1}^d \ribb_t|\widetilde{Z}_t^{h,b,x_{1},i}S^i_t|
+\rll_t \sum_{i=1}^d (-\widetilde{Z}_t^{h,b,x_{1},i}S^i_t )^- +\rbb_t \sum_{i=1}^d (\widetilde{Z}_t^{h,b,x_{1},i}S^i_t )^-\medskip\\
&=(r_t^{b}-r_t^{l})\widetilde{Y}_t^{h,b,x_{1}}
+ \sum_{i=1}^d(\rll_t-\ribb_t)(-\widetilde{Z}_t^{h,b,x_{1},i}S^i_t )^-
+\sum_{i=1}^d (\rbb_t-\ribb_t)(\widetilde{Z}_t^{h,b,x_{1},i}S^i_t )^- \leq 0.
\end{array}
\ede
The comparison theorem for BSDEs gives $\widetilde{Y}^{h,b,x_{1}}\ge\widetilde{Y}^{c,l,x_{2}}$ and thus $P^{c}_t (x_{2},-A,-C)\leq P^{h}_t (x_{1},A,C)$ $ \PTb$-a.s., for every $t\in[0,T]$.
\end{proof}

\brem \label{remark for specail contract pricing}
Consider a contract $(A,C)$ such that $A^{C}$ is a decreasing process on $(0,T]$. If $x_{1}\ge0$ then, from the proof
of the proposition, we see that $P^{h}(x_{1},A,C)$ does not depend on the initial wealth $x_{1}$, that is, for every $x,y\in\mathbb{R}_{+}$ we have $P^{h}(x,A,C)=P^{h}(y,A,C)$.
This follows from the equality $P^{h}(x_{1},A,C)=\widetilde{Y}^{h,l,x_{1}}-C$,
where  $(\widetilde{Y}^{h,l,x_{1}},\widetilde{Z}^{h,l,x_{1}})$ is the unique solution of BSDE (\ref{artifical BSDE 1 for positive hedger}), which is independent of $x_{1}$. Note that the above relation hinges on the condition $x_{1}\ge0$. Indeed, when $x_{1}\leq0$, then $P^{h}(x_{1},A,C)$ does not enjoy the independence property. Furthermore, for any $x_{2}\in\mathbb{R}$,
the price $P^{c}(x_{2},-A,-C)$ does not have such property. Finally, for a contract $(A,C)$ such that $A^{C}$ is an increasing process on $(0,T]$, if $x_{2}\ge0$, then $P^{c}(x_{2},-A,-C)$ does not depend on the initial wealth $x_{2}$, but $P^{h}(x_{1},A,C)$ does not have this property.

The above-mentioned property is intuitively clear from its financial interpretation. In essence, the independence of the hedger's price of his non-negative positive wealth is a consequence of the last constraint in equation (\ref{portfolio choose}), which states that the hedger cannot use his initial endowment to buy shares for the purpose of hedging. Of course, when he sells shares
to replicate an option, as is the case for the put option, then, obviously, the fact that his initial endowment is positive
is also irrelevant.
\erem

\brem\label{remark for increasing contract}
Assume that $x_{1}> 0$ and $x_{2}< 0$. We claim that if $r^{l}$ and $r^{b}$ are deterministic and satisfy $r^{l}_{t}<r^{b}_{t}$ for all $t\in[0,T]$, then we can find a date $\widehat{t}\in[0,T]$ and a contract $(A,C)$ with an increasing process $A^{C}$ such that
\bde
P^{c}_{\widehat{t}} (x_{2},-A,-C)>P^{h}_{\widehat{t}} (x_{1},A,C),  \quad \PTb-\aass
\ede
To this end, it suffices consider a contract $(A,C)$ with $C=0$ and $A_t = p \, \I_{[0,T]}(t)+\alpha\I_{[t_{0},T]}(t)$
where $t_{0}\in(0,T)$, $r_{t}\in(r^{l}_{t},r^{b}_{t})$ for every $t\in[0,T]$ and $\alpha$ satisfies
\bde
0 < \alpha\leq\min\left\{ x_{1}\Blr_{t_0}, -x_{2}\Bbr_{t_0}\right\}.
\ede
We set $x=x_{1}-\alpha (\Blr_{t_0})^{-1} \ge0$ and we define the strategy $\phi = \big(\xi^1,\dots ,\xi^d, \psi^{l}, \psi^{b},\psi^{1,b},\dots ,\psi^{d,b}, \etab, \etal\big)$ where $\xi^i=\psi^{i,b}=\psi^{b}=\etab=\etal=0$
for all $i=1,2,\ldots,d$ and
\bde
\psi^l_t = x \I_{[0,t_0)} + (\Blr_{t_0})^{-1} \big(x \Blr_{t_0}+\cac\big) \I_{[t_0,T]}.
\ede
Then we have
\bde
V_T (x, \varphi , A,C)=  x \Blr_{T} + \cac  e^{\int_{t_{0}}^{T}r^l_{u}\,du}
=\big(x_{1} - \alpha  (\Blr_{t_0})^{-1}\big) \Blr_{T}
+\cac  e^{\int_{t_{0}}^{T}r^l_{u}\,du} =x_{1}  (\Blr_{T})=V_{T}^{0}(x_{1}).
\ede
Hence the hedger's self-financing strategy $(x, \phi , A, \pC )$ replicates the contract $(A,C)$ on $[0,T]$
and, in fact, this is the unique replicating strategy.
From Definition \ref{definition of ex-dividend price}, it follows that $P^{h}_{0} (x_{1},A,C)=x-x_{1}=-\alpha (\Blr_{t_0})^{-1}$.
Let us now consider the contract from the perspective of the counterparty. For $\tilde{x}=x_{2}+\alpha (\Bbr_{t_0})^{-1} \leq 0$,
we define the strategy $\tilde{\phi} = \big(\tilde{\xi}^1,\dots ,\tilde{\xi}^d, \tilde{\psi}^{l}, \tilde{\psi}^{b},\tilde{\psi}^{1,b},\dots ,\tilde{\psi}^{d,b}, \tilde{\eta}^{b}, \tilde{\eta}^{l}\big)$
where $\tilde{\xi}^i=\tilde{\psi}^{i,b}=\tilde{\psi}^{l}=\tilde{\eta}^b=\tilde{\eta}^l=0$
for all $i=1,2,\ldots,d$ and
\bde
\tilde{\psi}^b_t = \tilde{x} \I_{[0,t_0)} +(\Bbr_{t_0})^{-1}\big(\tilde{x} \Bbr_{t_0} +\cac \big)  \I_{[t_0,T]}.
\ede
Then we have
\bde
V_T (\tilde{x}, \tilde{\varphi} , A,C)
=  \tilde{x}  \Bbr_{T} +\cac  e^{\int_{t_{0}}^{T}r^b_{u}\,du}
=x_{2} \Bbr_{T} =V_{T}^{0}(x_{2}).
\ede
Therefore, the self-financing strategy $(\tilde{x},\tilde{ \phi} , -A, -\pC )$ is the unique replicating strategy for the contract $(-A,-C)$ on $[0,T]$ and, from Definition \ref{remark for counterparty's ex-dividend price}, it follows that
$P^{c}_{0} (x_{2},-A,-C)=x_{2}-\tilde{x}=-\alpha (\Bbr_{t_0})^{-1}$. Moreover, since $r^l <r^{b}$ and $\alpha>0$, we have that
\bde
P^{h}_{0} (x_{1},A,C) = -\alpha  (\Blr_{t_0})^{-1} <-\alpha  (\Bbr_{t_0})^{-1} = P^{c}_{0} (x_{2},-A,-C).
\ede
Consequently, under the assumption that $r^{l}<r^{b}$ we have found a contract $(A,C)$ and a date
$\widehat{t}=0$ such that $P^{c}_{0} (x_{2},-A,-C)> P^{h}_{0} (x_{1},A,C)$. This means that the range of bilaterally
profitable prices ${\cal R}^p_0(x_1,x_2)$ for $(A,C)$ is non-empty.
\erem

%%%%%%%%%%%%%%%%%%%%%%%%%%%%%%%%%%%%%%%%%%%%%%%%%%%%%%%%%%%%%%%%%%%%%%%%%%
\subsection{Monotonicity of Prices with Respect to the Initial Endowment}     \label{sect5.4}
%%%%%%%%%%%%%%%%%%%%%%%%%%%%%%%%%%%%%%%%%%%%%%%%%%%%%%%%%%%%%%%%%%%%%%%%%%

As shown in the preceding subsection, the initial endowment plays an important role in the pricing inequality.
In the following, we examine in more details the impact of the initial endowment on the ex-dividend price. In view of Remark \ref{remark for the proof under different assumptions for cumulative dividend price}, we only need to work under Assumption \ref{assumption for artifical cumulative dividend price}.

\bp \label{monotonicity proposition}
Let either Assumption \ref{assumption for artifical cumulative dividend price} or
Assumption \ref{changed assumption for artifical cumulative dividend price} be valid and
let a contract $(A,C)$ be admissible under $\PTb$. Then the hedger's price satisfies: \hfill \break
(i) if $\bar{x}\ge x\ge0$, then
\be \label{increasing of hedger ex-dividend price for positive wealth}
P^{h}_t (\bar{x},A,C)\leq P^{h}_t (x,A,C),
\ee
(ii) if $0\ge\bar{x}\ge x$, then
\be \label{increasing of hedger ex-dividend price for negative wealth}
P^{h}_t (\bar{x},A,C)\ge P^{h}_t (x,A,C),
\ee
and the counterparty's price satisfies: \hfill \break
(i) if $\bar{x}\ge x\ge0$, then
\be \label{increasing of counterparty ex-dividend price for positive wealth}
P^{c}_t (\bar{x},-A,-C)\ge P^{c}_t (x,-A,-C),
\ee
(ii) if $0\ge\bar{x}\ge x$, then
\be \label{increasing of counterparty ex-dividend price for negative wealth}
P^{c}_t (\bar{x},-A,-C)\leq P^{c}_t (x,-A,-C).
\ee
\ep

\begin{proof}
Let us denote
\bde
g^{l,h}(x):=-x\rll_t\Blr_t+g(t, y+x\Blr_t,z), \quad g^{b,h}(x):=-x\rbb_t\Bbr_t+g(t, y+x\Bbr_t,z),
\ede
and
\bde
g^{l,c}(x):= x\rll_t\Blr_t-g(t, -y+x\Blr_t,-z), \quad g^{b,c}(x):=x\rbb_t\Bbr_t-g(t, -y+x\Bbr_t,-z),
\ede
where (see (\ref{drift driver for positive and negative initial wealth}))
\bde
g(t,y,z)=-\sumik_{i=1}^d \ribb_t(z^{i}S^i_t )^++\rll_t \Big(y+ \sumik_{i=1}^d ( z^{i}S^i_t )^- \Big)^+
- \rbb_t \Big( y+ \sumik_{i=1}^d ( z^{i}S^i_t )^- \Big)^-.
\ede
If we denote $K:=y+ \sum_{i=1}^d ( z^{i}S^i_t)^{-}$ and $\tilde{K}:=-y+ \sum_{i=1}^d (-z^{i}S^i_t)^{-}$, then
\bde
\begin{array}
[c]{ll}
g^{l,h}(x)&=-x\rll_t\Blr_t+\rll_t(x\Blr_t +K)^+- \rbb_t(x\Blr_t +K)^-\medskip\\
&=-\rll_t(x\Blr_t+K)+\rll_t(x\Blr_t +K)^+- \rbb_t(x\Blr_t +K)^-+\rll_{t}K\medskip\\
&=\rll_t(x\Blr_t +K)^-- \rbb_t(x\Blr_t +K)^-+\rll_{t}K\medskip\\\
&=(\rll_t- \rbb_t)(x\Blr_t +K)^-+\rll_{t}K
\end{array}
\ede
and
\bde
\begin{array}
[c]{ll}
g^{b,h}(x)&=-x\rbb_t\Blr_t+\rll_t(x\Bbr_t +K)^+- \rbb_t(x\Bbr_t +K)^-\medskip\\
&=-\rbb_t(x\Bbr_t+K)+\rll_t(x\Bbr_t +K)^+- \rbb_t(x\Bbr_t +K)^-+\rbb_{t}K\medskip\\
&=(\rll_t- \rbb_t)(x\Blr_t +K)^++\rbb_{t}K.
\end{array}
\ede
Similarly, \bde
%\begin{array}[c]{ll}
g^{l,c}(x)
%&=x\rll_t\Blr_t-\rll_t(x\Blr_t +\tilde{K})^++\rbb_t(x\Blr_t +\tilde{K})^-\medskip\\
%&=\rll_t(x\Blr_t+\tilde{K})-\rll_t(x\Blr_t +\tilde{K})^++\rbb_t(x\Blr_t +\tilde{K})^--\rll_{t}\tilde{K}\medskip\\
=(\rbb_t- \rll_t)(x\Blr_t +\tilde{K})^--\rll_{t}\tilde{K}
%\end{array}
\ede
and
\bde
%\begin{array}[c]{ll}
g^{b,c}(x)
%&=x\rbb_t\Blr_t-\rll_t(x\Bbr_t +\tilde{K})^++\rbb_t(x\Bbr_t +\tilde{K})^-\medskip\\
%&=\rbb_t(x\Bbr_t+\tilde{K})-\rll_t(x\Bbr_t +\tilde{K})^++\rbb_t(x\Bbr_t +\tilde{K})^--\rbb_{t}K\medskip\\
=(\rbb_t- \rll_t)(x\Blr_t +K)^+-\rbb_{t}K.
%\end{array}
\ede
Therefore, the functions $\widetilde{g}^{l,h}(x)$ and $\widetilde{g}^{b,c}(x)$ are increasing with respect to $x$, whereas
the functions $\widetilde{g}^{b,h}(x)$ and $\widetilde{g}^{l,c}(x)$ are decreasing with respect to $x$.
Consequently, from the comparison theorem for BSDEs, if $\bar{x}\ge x\ge0$, then $\widetilde{Y}^{h,l,x}\leq\widetilde{Y}^{h,l,\bar{x}}$ where $(\widetilde{Y}^{h,l,x},\widetilde{Z}^{h,l,x})$ is the unique solution of BSDE (\ref{artifical BSDE for positive hedger}). Moreover,  $\widetilde{Y}^{c,l,x}\ge\widetilde{Y}^{c,l,\bar{x}}$ where $(\widetilde{Y}^{c,l,x},\widetilde{Z}^{c,l,x})$ is the unique solution of BSDE
(\ref{artifical BSDE for positive counterparty}).  Then from Remark \ref{remark for the proof under different assumptions for cumulative dividend price}, we deduce that (\ref{increasing of hedger ex-dividend price for positive wealth}) and (\ref{increasing of counterparty ex-dividend price for positive wealth}) hold.
For $0\ge\bar{x}\ge x$ one can show, using similar arguments, that (\ref{increasing of hedger ex-dividend price for negative wealth}) and (\ref{increasing of counterparty ex-dividend price for negative wealth}) are valid.
\end{proof}

By combining Propositions \ref{inequality proposition for both negative initial wealth}--\ref{monotonicity proposition}, we
obtain the following result, which summarizes the properties of unilateral prices.

\begin{theorem} \label{main theorem}
Let either Assumption \ref{assumption for artifical cumulative dividend price} or
Assumption \ref{changed assumption for artifical cumulative dividend price} be valid.
Then for any contract $(A,C)$ admissible under $\PTb$ the following statements are valid:  \hfill \break
(i) if $\bar{x}\ge x\ge0$, then for all $t\in[0,T]$
\be\label{monotonicity ex-dividend price for positive wealth}
P^{c}_t (x,-A,-C)\leq P^{c}_t (\bar{x},-A,-C)\leq P^{h}_t (\bar{x},A,C)\leq P^{h}_t (x,A,C).
\ee
(ii) if $0\ge\bar{x}\ge x$, then for all $t\in[0,T]$
\be\label{monotonicity ex-dividend price for negative wealth}
P^{h}_t (\bar{x},A,C)\ge P^{h}_t (x,A,C)\ge P^{c}_t (x,-A,-C)\ge P^{c}_t (\bar{x},-A,-C).
\ee
Moreover, if $r^{l}$ and $r^{b}$ are deterministic and satisfy $r^{l}_{t}<r^{b}_{t}$ for all $t\in[0,T]$, then for $\bar{x}>0>x$, there exists $(\wh{t},A,C)$ such that
\bde
P^{c}_{\wh{t}} (x,-A,-C)>P^{h}_{\wh{t}} (\bar{x},A,C)\ge P^{c}_{\wh{t}} (\bar{x},-A,-C)
\ede
and there also exists $(\wh{t},A,C)$ such that
\bde
P^{h}_{\wh{t}} (\bar{x},A,C)\ge P^{c}_{\wh{t}} (\bar{x},-A,-C)> P^{h}_{\wh{t}} (x,A,C).
\ede
\end{theorem}

\bcor \label{corollary of price bound}
Under the assumptions of Proposition \ref{monotonicity proposition}, for any contract $(A,C)$ and any date $t\in[0,T]$
\be\label{price bound}
P^{c}_t (0,-A,-C)\leq P^{c}_t (x,-A,-C)\leq P^{h}_t (x,A,C)\leq P^{h}_t (0,A,C),
\ee
so that ${\cal R}^f_t (x,x) \subset {\cal R}^f_t (0,0)$.
\ecor

The above corollary shows that an investor with either a positive or a negative initial endowment
has a potential advantage over an investor with null initial wealth to enter any contract $(A,C)$ at any time $t$.
This conclusion is plausible, since the borrowing rate is higher than the lending rate.

 Indeed, for the same strategy, when an
investor who has zero initial endowment needs to borrow money in order to hedge a contract, an investor with a positive initial endowment may use money from his initial wealth for the same purpose.  Similarly, when an investor with null initial endowment
needs to lend money in order to implement his hedging strategy, an investor with a negative initial
endowment can use instead a surplus of cash to repay his debt. These features create a comparative advantage.

Using Corollary \ref{corollary of price bound} and Proposition \ref{monotonicity proposition},
we can examine the asymptotic properties of $P^{h}_{t}(x,A,C)$ and $P^{c}_{t}(x,-A,-C)$ when the initial
endowment $x$ tends to either $\infty $ or $-\infty$.

\bp
Let the assumptions of Proposition \ref{monotonicity proposition} be valid. For any contract $(A,C)$ and any date $t\in[0,T]$,
there exist $\gg$-adapted processes, denoted  by $P^{h,A,C,+}_{t}$, $P^{h,A,C,-}_{t}$, $P^{c,-A,-C,+}_{t}$ and $P^{c,-A,-C,-}_{t}$,
such that
\bde
P^{h,A,C,+}_{t},\, P^{h,A,C,-}_{t},\, P^{c,-A,-C,+}_{t},\, P^{c,-A,-C,-}_{t} \in [P^{c}_t (0,-A,-C) ,  P^{h}_t (0,A,C)]
 = {\cal R}^f_0 (0,0)
\ede
and % (taking the limit in probability sense),
\bde
\lim\limits_{x\rightarrow +\infty}P^{h}_t (x,A,C)=P^{h,A,C,+}_{t}\ge P^{c,-A,-C,+}_{t}=\lim\limits_{x\rightarrow +\infty}P^{c}_t (x,-A,-C),
\ede
\bde
\lim\limits_{x\rightarrow -\infty}P^{h}_t (x,A,C)=P^{h,A,C,-}_{t}\ge P^{c,-A,-C,-}_{t}=\lim\limits_{x\rightarrow -\infty}P^{c}_t (x,-A,-C).
\ede
\ep

\begin{proof}
The statement easily follows  from Proposition \ref{monotonicity proposition} and Corollary \ref{corollary of price bound}.
\end{proof}

We can only have $P^{h,A,C,+}_{t}\ge P^{c,A,C,+}_{t}$ and $P^{h,A,C,-}_{t}\ge P^{c,A,C,-}_{t}$. Other comparison results between these four processes are still unclear. Indeed, if $r^{l}$ and $r^{b}$ are deterministic and such that $r^{l}_{t}<r^{b}_{t}$ for every $t\in[0,T]$, then there exists $(\wh{t},A,C)$ such that
\bde
P^{h,A,C,+}_{\wh{t}}\ge P^{c,A,C,+}_{\wh{t}}>P^{h,A,C,-}_{\wh{t}}\ge P^{c,A,C,-}_{\wh{t}},
\ede
as well as there exists $(\wh{t},A,C)$ such that
\bde
P^{h,A,C,-}_{\wh{t}}\ge P^{c,A,C,-}_{\wh{t}}>P^{h,A,C,+}_{\wh{t}}\ge P^{c,A,C,+}_{\wh{t}}.
\ede
Now, we consider a special case of a contract $(A,C)$ and $t\in[0,T]$ such that
\bde
P^{h,A,C}_{t}:=\min\left\{P^{h,A,C,+}_{t},P^{h,A,C,-}_{t}\right\}
\ge\max\left\{P^{c,-A,-C,+}_{t},P^{c,-A,-C,-}_{t}\right\} =: P^{c,-A,-C}_{t}.
\ede
Then $\big[P_t^{c,-A,-C},P_t^{h,A,C}\big]$ is the bilateral fair pricing range for all investors with identical, but otherwise
arbitrary, initial endowment, meaning that
\bde
\big[P^{c,-A,-C}_{t},P^{h,A,C}_{t}\big]=\bigcap_{x\in\mathbb{R}}[P^{c}_{t}(x,-A,-C),P^{h}_{t}(x,A,C)]
 = \bigcap_{x\in\mathbb{R}} {\cal R}^f_t (x,x).
\ede

The following stability of unilateral ex-dividend prices with respect to the initial endowment can also be
established using Proposition 3.1 in \cite{NR3}. % \ref{stability of BSDE}.
For the reader's convenience, we recall two alternative versions of the Lipschitz condition, which were employed in \cite{NR3}
(see Definitions 2.1 and 3.1 in \cite{NR3}).
Let $h:\Omega \times[0,T] \times \rr \times \rr^{d} \rightarrow \rr$ be a ${\cal G}\otimes\mathcal{B}([0,T])\otimes\mathcal{B}(\rr)\otimes\mathcal{B}(\rr^d)$-measurable function such that $h(\cdot,\cdot,y,z)$ is a $\gg$-adapted process for any fixed $(y,z)\in\rr \times \rr^d$,
and let $m$ be the process introduced in either Assumption \ref{assumption for artifical cumulative dividend price}
or Assumption \ref{changed assumption for artifical cumulative dividend price}.

\bd \label{definition uniformly Lipschitz}
We say that $h$ satisfies the {\it uniform Lipschitz condition} if there exists a constant $L$ such that, for all $t\in[0,T]$ and $y_{1},y_{2}\in\mathbb{R},\, z_{1},z_{2}\in\mathbb{R}^d$,
\be  \label{uniformly Lipschitz for the driver}
|h(t,y_{1},z_{1})-h(t,y_{2},z_{2})|\leq L\left(|y_{1}-y_{2}|+\|z_{1}-z_{2}\|\right), \quad \P-\aass
\ee
We say that $h$ satisfies the {\it $m$-Lipschitz condition} if there exist two strictly positive and $\gg$-adapted processes $\rho $ and $\theta$ such that, for all $t\in[0,T]$ and $y_{1},y_{2}\in\mathbb{R},\, z_{1},z_{2}\in\mathbb{R}^d$,
\be \label{Lipschitz for the driver}
|h(t,y_{1},z_{1})-h(t,y_{2},z_{2})|\leq \rho_{t}|y_{1}-y_{2}|+\theta_{t}\| m_{t}^{\ast}(z_{1}-z_{2})\|.
\ee
\ed

\begin{theorem}\label{stability property of price}
Let either Assumption \ref{assumption for artifical cumulative dividend price} or
Assumption \ref{changed assumption for artifical cumulative dividend price} be valid.
Then for any contract $(A,C)$ admissible under $\PTb$, there exists a constant $K_{0}$ such that
\bde
\EP \left[\sup_{t\in[0,T]}|P^{h}_{t} (x_{1},A,C)-P^{h}_{t} (x_{2},A,C)|+\sup_{t\in[0,T]}|P^{c}_{t} (x_{1},-A,-C)-P^{c}_{t} (x_{2},-A,-C)|\right]\leq K_{0}|x_{1}-x_{2}|.
\ede
\end{theorem}

\begin{proof}
From Remark \ref{remark for the proof under different assumptions for cumulative dividend price}, we have
$P^{h}_t (x_{i},A,C) =\widetilde{Y}^{h,x_{i}}_t-C_t$ for every $t \in [0,T)$,
where $(\widetilde{Y}^{h,x_{i}},\widetilde{Z}^{h,x_{i}})$ is the solution of the following BSDE
\be\label{additional BSDE for hedger with positive wealth}
\left\{
\begin{array}
[c]{ll}
d\widetilde{Y}^{h,x_{i}}_t =\widetilde{Z}^{h,x_{i},\ast}_td\wt S_t^{\textrm{cld}}+g^{h}(t,x_{i},\widetilde{Y}^{h,x_{i}}_t,\widetilde{Z}^{h,x_{i},\ast}_t)\, dt +dA^C_t,\medskip\\
\widetilde{Y}^{h,x_{i}}_T=0,
\end{array}
\right.
\ee
where
\bde
g^{h}(t,x,y,z):=\sumik_{i=1}^d \beta^{i}_{t} z^{i}_t S_{t}^{i}+\big( g(t,y+x\Blr_t,z)-x\rll_t\Blr_t \big)\I_{\{x\ge0\}}
+\big( g(t, y+x\Bbr_t,z)-x\rbb_t\Bbr_t \big) \I_{\{x\leq0\}}.
\ede
It is not hard to check that if $x_{1}x_{2}\ge0$, then there exists a constant $K$, which only depends on
the bound for $\rll$ and $\rbb$, such that
\bde
|g^{h}(t,x_{1},y,z)-g^{h}(t,x_{2},y,z)|\leq K|x_{1}-x_{2}|.
\ede
Consequently, if $x_{1}x_{2}<0$, then
\bde
\begin{array}
[c]{ll}
|g^{h}(t,x_{1},y,z)-g^{h}(t,x_{2},y,z)|&\leq |g^{h}(t,x_{1},y,z)-g^{h}(t,0,y,z)|+|g^{h}(t,0,y,z)-g^{h}(t,x_{2},y,z)|\medskip\\
&\leq K|x_{1}|+K|x_{2}|=K|x_{1}-x_{2}|.
\end{array}
\ede
We conclude that there exists a constant $K$, which depends only on the bound for $\rll$ and $\rbb$, such that
\bde
|g^{h}(t,x_{1},y,z)-g^{h}(t,x_{2},y,z)|\leq K|x_{1}-x_{2}|, \text{ for all } x_{1},x_{2}\in\mathbb{R}.
\ede
Under Assumption \ref{assumption for artifical cumulative dividend price}
(resp., Assumption \ref{changed assumption for artifical cumulative dividend price}), for a fixed $x\in\mathbb{R}$, $g^{h}(t,x,y,z)$ satisfies (\ref{Lipschitz for the driver}) with $\rho=\theta= \wh{L}$, where a constant $\wh{L}$ depends on the bound for $r^{l},r^{b},r^{i,b}$ and $S^{i}$ for $i=1,2,\ldots,d$, as well as the lower bound for $|m|$ (resp., a constant $\wh{L}$ depends on the bound $r^{l},r^{b}$ and $r^{i,b}$).  Consequently, there always exists a constant $\wh{L}$, such that the driver satisfies
the Lipschitz condition (\ref{Lipschitz for the driver}) with processes $\rho=\theta= \wh{L}$.
Consequently, as in Section 3.2 in \cite{NR3},  we deduce that the spaces $\wHlamo$ and $ \wHzero$, (resp., the spaces $\widehat{L}^{2}_{\lambda}$ and $\widehat{L}^{2}_{0}$) may be identified, since the related norms are equivalent.
Moreover, one can check $\alpha^{-1}g^{h}(t,x,0,0)\in\wHzero$.

By an application of Proposition 3.2 in \cite{NR3}, there exists a constant $K_{0}$ such that
\begin{align*}
&\EP \bigg[\sup_{t\in[0,T]}|P^{h}_{t}(x_{1},A,C)-P^{h}_{t}(x_{2},A,C)|^{2}\bigg] = \EP \bigg[\sup_{t\in[0,T]}|\widetilde{Y}^{h,x_{1}}_{t}-\widetilde{Y}^{h,x_{2}}_{t}|^{2}\bigg]\\
& \leq  K_{0}\left|\alpha^{-1}g^{h}(t,x_{1},\widetilde{Y}^{h,x_{2}}_{t}-A_{t}^{C},\widetilde{Z}^{h,x_{2}}_{t})
-\alpha^{-1}g^{h}(t,x_{2},\widetilde{Y}^{h,x_{2}}_{t}-A_{t}^{C},\widetilde{Z}^{h,x_{2}}_{t})
\right|_{\wHzero}^{2}\\
&\leq K_{0}|x_{1}-x_{2}|^{2}.
\end{align*}
%Hence
%\bde
%\EP \bigg[\sup_{t\in[0,T]}|P^{h}_{t}(x_{1},A,C)-P^{h}_{t}(x_{2},A,C)|\bigg]
%=\EP \bigg[\sup_{t\in[0,T]}|\widetilde{Y}^{h,x_{1}}_{t}-\widetilde{Y}^{h,x_{2}}_{t}|\bigg]\leq K_{0}|x_{1}-x_{2}|.
%\ede
Similarly, one can check that the same inequality holds for the counterparty's price. \end{proof}
%\bde
%\EP \bigg[\sup_{t\in[0,T]}|P^{c}_{t}(x_{1},-A,-C)-P^{c}_{t}(x_{2},-A,-C)|\bigg]
%\leq K_{0}|x_{1}-x_{2}|,
%\ede
%which completes the proof.
%\end{proof}

% \newpage

%%%%%%%%%%%%%%%%%%%%%%%%%%%%%%%%%%%%%%%%%%%%%%%%%%%%%%%%%%%%%%%%%%%%%%%%%%%
\subsection{Price Independence of the Initial Endowment}
%%%%%%%%%%%%%%%%%%%%%%%%%%%%%%%%%%%%%%%%%%%%%%%%%%%%%%%%%%%%%%%%%%%%%%%%%%%

We will now show that for a certain class of contracts the price is independent of the initial endowment.
It is worth noting that an analogous result does not hold in Bergman's model studied in \cite{NR4}.

%\bd
%For any fixed $t_0\in[0,T]$, we consider a contract $(A,C)$ starting from $t_0$ where $C$ is exogenously given, we say that $(A,C)$ is {\it admissible under} $\PT^{\beta}$ when $A^C-A^C_{t_0} \in\Classaptb([t_0,T]) $.
%\ed

\bp \lab{pro_new1}
Let $x_{1}\ge0$ and either  Assumption \ref{assumption for artifical cumulative dividend price} or Assumption  \ref{changed assumption for artifical cumulative dividend price} be valid. Consider an arbitrary  contract $(A,C)$ admissible under $\PTb$. If the process $A^{C}-A^C_0$ is decreasing, then the price $P^{h}_{t} (x_{1},A,C)$ is independent of $x_1$, so that
$P^{h}_t (x_{1},A,C)=P^{h}_t (0,A,C)$ for all $x_1 \ge0$.
\ep

\begin{proof}
Since $x_{1}\ge0$, it follows from Proposition \ref{general pricing proposition} that the hedger's price
of any contract $(A,C)$ admissible under $\PTb$ satisfies $P^{h}_{t}(x_{1},A,C)=\widetilde{Y}^{h,l,x_{1}}_{t}-C_{t}$
where  $(\widetilde{Y}^{h,l,x_{1}},\widetilde{Z}^{h,l,x_{1}})$ is the unique solution of the following BSDE
\be \label{aifical BSDE for positive hedger}
\left\{
\begin{array}
[c]{ll}
d\widetilde{Y}^{h,l,x_{1}}_t =\widetilde{Z}^{h,l,x_{1},\ast}_t\, d\wt S_t^{\textrm{cld}}+g^{h,l}(t, x_{1}, \widetilde{Y}^{h,l,x_{1}}_t,\widetilde{Z}^{h,l,x_{1}}_t)\,dt+dA^C_t,\medskip\\
\widetilde{Y}^{h,l,x_{1}}_T=0,
\end{array}
\right.
\ee
where
\bde
\begin{array}
[c]{rl}
g^{h,l}(t,x_1,y,z):=&\sumik_{i=1}^d \beta^{i}_{t} z^iS_{t}^{i}-x_1\rll_t\Blr_t-\sumik_{i=1}^d \ribb_t(z^{i}S^i_t )^+\medskip\\
&\mbox{}+\rll_t \Big(y+x_1B_t^l+\sumik_{i=1}^d ( z^{i}S^i_t )^- \Big)^+
- \rbb_t \Big( y+ x_1B_t^l+\sumik_{i=1}^d ( z^{i}S^i_t )^- \Big)^-.
\end{array}
\ede
Since $g^{h,l}(t,x_{1},0,0)=0$ and the process $A^{C}-A^C_{0}$ is decreasing, we deduce from the comparison theorem for BSDEs
(see, for instance, Theorem 3.3 in \cite{NR3} with $U^1=A^{C}-A^C_{0}$ and $U^2=0$) that $\widetilde{Y}^{h,l,x_{1}}\ge 0$. Since $x_{1}\ge0$,  BSDE \eqref{aifical BSDE for positive hedger} can thus be represented as follows
\be\label{aifical BSDE 1 for positive hedger}
\left\{
\begin{array}
[c]{ll}
d\widetilde{Y}^{h,l,x_{1}}_t =\widetilde{Z}^{h,l,x_{1},\ast}_td\wt S_t^{\textrm{cld}}+\widetilde{g}^{h,l}(t, x_{1}, \widetilde{Y}^{h,l,x_{1}}_t,\widetilde{Z}^{h,l,x_{1}}_t)\,dt+dA^C_t,\medskip\\
\widetilde{Y}^{h,l,x_{1}}_T=0,
\end{array}
\right.
\ee
where the generator $\tilde{g}^{h,l}(t,x_1,y,z)$ is independent of $x_1$ and equals (recall that $\zzb = z^i S^i_t$)
\bde
\tilde{g}^{h,l}(t,x,y,z):=\sumik_{i=1}^d \beta^{i}_{t} \zzb
-\sumik_{i=1}^d \ribb_t( \zzb )^++\rll_ty+ \rll_t\sumik_{i=1}^d ( \zzb )^- .
\ede
 Obviously, the unique solution to BSDE (\ref{aifical BSDE 1 for positive hedger}) is independent of $x_1$ and thus the price $P^{h}_{t}(x_{1},A,C)=\widetilde{Y}^{h,l,x_{1}}_{t}-C_{t}$ enjoys the same property.
\end{proof}

%The proof of the next result is similar and thus it is omitted.
%
%\bp
%Let $x_{2}\le 0 $ and either  Assumption \ref{assumption for artifical cumulative dividend price} or Assumption  \ref{changed assumption for artifical cumulative dividend price} be valid. For an arbitrary  contract $(A,C)$ admissible under $\PTb$, if the process $A^{C}-A^C_0$ is increasing, then the price $P^{c}_{t} (x_{2},-A,-C)$ is independent of $x_2$.
%\ep
%
%\brem
%For $x_1 < 0$ or $x_2 > 0$, we do not have similar propositions.
%\erem

%%%%%%%%%%%%%%%%%%%%%%%%%%%%%%%%%%%%%%%%%%%%%%%%%%%%%%%%%%%%%%%%%
\subsection{Positive Homogeneity of the Hedger's Price}
%%%%%%%%%%%%%%%%%%%%%%%%%%%%%%%%%%%%%%%%%%%%%%%%%%%%%%%%%%%%%%%%%

We consider once again the hedger's price and we show that it is positively homogeneous with respect to
the size of the contract and the non-negative initial endowment. Observe that this property is no longer true if only
the size of the contract (but not the initial endowment) is scaled by a non-negative number $\lambda $ (of course, unless
the price is independent of the initial endowment, as in Proposition \ref{pro_new1}).

\bp \lab{pro_new2}
Let $x_{1}\ge0$ and either  Assumption \ref{assumption for artifical cumulative dividend price} or Assumption  \ref{changed assumption for artifical cumulative dividend price} be valid. Consider an arbitrary contract $(A,C)$ admissible under $\PTb$. If the process $C\in\widehat{\mathcal{H}}_0^2$, then for all $\lambda\in\mathbb{R}_+$
\be\label{homogeneous 1}
P^{h}_t (\lambda x_{1},\lambda A,\lambda C)=\lambda P^{h}_t (x_1,A,C).
\ee
\ep

\begin{proof}
It is obvious that (\ref{homogeneous 1}) holds for $\lambda=0$. Suppose that $\lambda > 0$.
Once again, from Proposition \ref{general pricing proposition},  we know that $P^{h}(x_{1},A,C)=\widetilde{Y}^{h,l,x_{1}}-C$ where  $(\widetilde{Y}^{h,l,x_{1}},\widetilde{Z}^{h,l,x_{1}})$ is the unique solution to \eqref{aifical BSDE for positive hedger}. Moreover, $P^{h}(\lambda x_{1},\lambda A,\lambda C)=\widetilde{Y}^{h,l, \lambda x_{1}}-\lambda C$ where  $(\widetilde{Y}^{h,l,\lambda x_{1}},\widetilde{Z}^{h,l,\lambda x_{1}})$ is the unique solution of the following BSDE
\bde
\left\{
\begin{array}
[c]{ll}
d\widetilde{Y}^{h,l,\lambda x_{1}}_t =\widetilde{Z}^{h,l,\lambda x_{1},\ast}_t\, d\wt S_t^{\textrm{cld}}+g^{h,l}(t, \lambda x_{1}, \widetilde{Y}^{h,l,\lambda x_{1}}_t,\widetilde{Z}^{h,l,\lambda x_{1}}_t)\,dt+\lambda \, dA^C_t,\medskip\\
\widetilde{Y}^{h,l,\lambda x_{1}}_T=0 .
\end{array}
\right.
\ede
Recall that $A^C=A+C+F^C$ where  $F^{C}_t := -\int_{0}^t  r^c_u \pC_u \, du$.  Then $P^{h}(x_{1},A,C)=Y^1$
where  $(Y^1,Z^1)$ is the unique solution of the following BSDE (since  $(A,C)$ is admissible and $C\in\widehat{\mathcal{H}}_0^2$,  the well-posedness of this BSDE is easy to check)
\bde
\left\{
\begin{array}
[c]{ll}
dY^1_t =Z^{1,\ast}_t\, d\wt S_t^{\textrm{cld}}+g^{h,l}(t, x_{1}, Y^{1}_t+C_t,Z^{1}_t)\,dt+d(A_t+F_t^C),\medskip\\
Y^{1}_T=0.
\end{array}
\right.
\ede
Similarly, $P^{h}(\lambda x_{1},\lambda A,\lambda C)=Y^{2}$ where  $(Y^{2},Z^{2})$ is the unique solution of the following BSDE
\be \label{BSDE 1}
\left\{
\begin{array}
[c]{ll}
dY^{2}_t =Z^{2,\ast}_t\, d\wt S_t^{\textrm{cld}}+g^{h,l}(t, \lambda x_{1},Y^{2}_t+\lambda C_t,Z^{2}_t)\,dt+\lambda \, d(A_t+F_t^C),\medskip\\
Y^{2}_T=0.
\end{array}
\right.
\ee
For $Y:=\lambda Y^1$ and $Z=\lambda Z^1$, we have
\be \label{BSDE 2}
\left\{
\begin{array}
[c]{ll}
dY_t =Z^{\ast}_t\, d\wt S_t^{\textrm{cld}}+\lambda g^{h,l}(t, x_{1}, \lambda^{-1}Y_t+C_t,\lambda^{-1}Z_t)\,dt+\lambda \, d(A_t+F_t^C),\medskip\\
Y_T=0.
\end{array}
\right.
\ee
Hence to complete the proof, it suffices to observe that the equality
\[
\lambda g^{h,l}(t, x_{1}, \lambda^{-1}y+C_t,\lambda^{-1}z)=g^{h,l}(t, \lambda x_{1},y+\lambda C_t,z)
\]
is satisfied for all $\lambda >0$.
\end{proof}
%
%\bp \lab{pro_new3}
%Let $x_{2}\ge0$ and either Assumption \ref{assumption for artifical cumulative dividend price} or Assumption  \ref{changed assumption for artifical cumulative dividend price} be valid. Consider an arbitrary contract $(A,C)$ admissible under $\PTb$. If the process $C\in\widehat{\mathcal{H}}_0^2$, then for every $ \lambda\in\mathbb{R}_+$
%\be\label{homogeneous 2}
%P^{c}_t (\lambda x_{2},-\lambda A,-\lambda C)=\lambda P^{c}_t (x_2,-A,-C).
%\ee
%\ep
%
%\brem
%Similar results holds for $x_1\leq0$ (resp. $x_2\leq 0$).
%\erem
%

%%%%%%%%%%%%%%%%%%%%%%%%%%%%%%%%%%%%%%%%%%%%%%%%%%%%%%%%%%%%%%%%%%%%%%%%%%%%%
\subsection{European Claims and Related Pricing PDEs}   \label{sect5.5}
%%%%%%%%%%%%%%%%%%%%%%%%%%%%%%%%%%%%%%%%%%%%%%%%%%%%%%%%%%%%%%%%%%%%%%%%%%%%%

For simplicity of presentation,  we assume that $d=1$, so that there is only one risky asset $S=S^1$. It is clear,
however, that the results obtained in this subsection can be easily extended to a multi-asset case. Moreover, we postulate that the interest rates $r^{l}$ and $r^{b}$ are deterministic. We examine valuation and hedging of an uncollateralized European contingent claim starting from a fixed time $t\in[0,T]$, that is, we set $C=0$. A generic path-independent claim of European style pays a single cash flow $H(S_{T})$ on the expiration date $T>0$, so that
\bde
A_t - A_0 = -H(S_{T})\I_{[T,T]}(t).
\ede
For any fixed $t <T$, the risky asset $S$ has the ex-dividend price dynamics under $\P$ given by the following expression, for $u \in [t,T]$,
\be\label{Partial netting model stock price}
dS_u = \mu(u,S_{u})\, du +\sigma(u,S_{u})\, dW_u , \quad S_t=\ssx \in \mathcal{O},
\ee
where $W$ is a one-dimensional Brownian motion and $\mathcal{O}$ is the domain of real values that are attainable by the diffusion process $S$ (usually $\mathcal{O}=\mathbb{R}_{+}$). Moreover, the  coefficients $\mu$ and $\sigma$ are such that SDE (\ref{Partial netting model stock price}) has a unique strong solution. We also assume that the volatility coefficient $\sigma$ is bounded and bounded away from zero. Finally, the dividend process equals $\pA^1_t = \int_0^t \kappa( u, S_u) \, du $.

Our first goal is to derive the hedger's pricing PDE for a path-independent European claim. We observe that
\bde
d\wt S^{\textrm{cld}}_u =dS_u + d\pA^1_u -\beta(u, S_{u})\, du=
\big( \mu(u, S_{u})+\kappa(u, S_{u})-\beta(u, S_{u}) \big) du + \sigma (u, S_{u})\, dW_u .
\ede
From the Girsanov theorem, if we denote
\bde
a_{u}:=(\sigma(u, S_{u}))^{-1}\big( \mu(u, S_{u})+\kappa(u, S_{u})-\beta(u, S_{u})\big)
\ede
and define the probability measure $\PT$ as
\bde
\frac{d\PT}{d\P}=\exp\bigg\{-\int_{t}^{T}a_{u}\, dW_{u}
-\frac{1}{2}\int_{t}^{T}|a_{u}|^{2}\, du \bigg\},
\ede
then $\PT $ is equivalent to $\P$ and the process $\widetilde{W}$ is the Brownian motion under $\PT$, where
$d\widetilde{W}_{u}:=dW_{u}+a_{u}\, du$. It is easy to see that
\bde
d\wt S^{\textrm{cld}}_{u}=\sigma (u, S_{u})\, d\widetilde{W}_u
\ede
and thus we conclude that $\wt S^{\textrm{cld}}$ is a $(\PT , \gg)$-martingale and $\langle \wt S^{\textrm{cld}}\rangle_{u}=\int_{t}^{u}|\sigma(v, S_{v})|^{2}\, dv$. Therefore, either Assumption \ref{assumption for artifical cumulative dividend price} or Assumption \ref{changed assumption for artifical cumulative dividend price} holds, provided
that we assume that the Brownian motion $\widetilde{W}$ has the PRP under $(\gg,\PT)$. Of course, the latter assumption is not restrictive in the present set-up.

Since $A$ has only a single cash flow at time $T$ and $C=0$, we deduce from Proposition \ref{general pricing proposition} that,
for any initial endowment $x_{1}\in\mathbb{R}$, the hedger's prices satisfies $P^{h}(x_{1},A,C)=\widetilde{Y}^{h,x_{1}}$,
where  $(\widetilde{Y}^{h,x_{1}},\widetilde{Z}^{h,x_{1}})$ is the unique solution of following BSDE driven by the Brownian motion $\widetilde{W}$
\be\label{Partial netting model Brownian BSDE for hedger}
\left\{
\begin{array}
[c]{ll}
d\widetilde{Y}^{h,x_{1}}_u =\widetilde{Z}^{h,x_{1}}_u\sigma (u, S_{u})\, d\widetilde{W}_u+g^{h}(u,x_{1},S_{u}, \widetilde{Y}^{h,x_{1}}_u,\widetilde{Z}^{h,x_{1}}_u)\, du ,\medskip\\
\widetilde{Y}^{h,x_{1}}_T=H(S_{T}),
\end{array}
\right.
\ee
where for $x_{1}\ge0$
\bde
g^{h}(u,x_{1},\ssy,y,z):=z\beta(u,\ssy)-x_{1}\rll_u\Blr_t-r^{1,b}_{u}(z\ssy)^{+}+\rll_u \Big(y+x_{1}\Blr_u+(z\ssy)^{-}\Big)^+- \rbb_u \Big( y+x_{1}\Blr_u+(z\ssy)^{-}\Big)^-
\ede
and for $x_{1}\leq0$
\bde
g^{h}(u,x_{1},\ssy,y,z):=z\beta(u,\ssy)-x_{1}\rbb_u\Bbr_u-r^{1,b}_{u}(z\ssy)^{+}+\rll_u \Big(y+x_{1}\Bbr_u+(z\ssy)^{-}\Big)^+- \rbb_u \Big( y+x_{1}\Bbr_u+(z\ssy)^{-}\Big)^-.
\ede
The well-posedness of BSDE \eqref{Partial netting model Brownian BSDE for hedger} is well known under mild assumptions, since we assumed that $\widetilde{W}$ has the PRP under $(\gg,\PT)$. The unique replicating strategy for the hedger equals $\phi = \big(\xi, \psi^{l}, \psi^{b},\psi^{1,b}\big)$ where $\xi_{u}= \widetilde{Z}^{h,x_{1}}_{u},\, \psi^{1,b}_{u}=-(B^{1,b}_{u})^{-1}(\xi_u S_u)^{+}$
and
\bde
\begin{array}
[c]{ll}
\psi^{l}_u &= (\Blr_u)^{-1} \Big( \widetilde{Y}^{h,x_{1}}_{u}+x_{1}\Blr_u\I_{\{x_{1}\ge0\}}+x_{1}\Bbr_u\I_{\{x_{1}\leq0\}}+(\xi_u S_u)^{-} \Big)^+, \medskip\\
\psi^{b}_u &= - (\Bbr_u)^{-1} \Big( \widetilde{Y}^{h,x_{1}}_{u}+x_{1}\Blr_u\I_{\{x_{1}\ge0\}}+x_{1}\Bbr_u\I_{\{x_{1}\leq0\}}+(\xi_u S_u)^{-} \Big)^-.
\end{array}
\ede

In the next step, we fix a date $t \in [0,T)$ and we assume that $\Sst_t=\ssx\in \mathcal{O}$. Note that under $\PT$, for
all $u \in [t,T]$,
\bde
d\Sst_u = ( \beta(u, \Sst_{u})-\kappa(u,\Sst_u))\, du+\sigma(u,\Sst_{u})\, d\widetilde{W}_u .
\ede
It is clear that the solution $(\widetilde{Y}^{h,x_{1}},\widetilde{Z}^{h,x_{1}})$ will now depend on the initial value $\ssx$ at time $t$ of the stock price; to emphasize this feature, we write $(\widetilde{Y}^{h,x_{1},\ssx},\widetilde{Z}^{h,x_{1},\ssx})$. Furthermore, if we set $(Y^{h,x_{1},\ssx}_{u},Z^{h,x_{1},\ssx}_{u}):=(\widetilde{Y}^{h,x_{1},\ssx}_{u},\widetilde{Z}^{h,x_{1},\ssx}_{u}\sigma(u,\Sst_{u}))$ and
\bde
\overline{g}^{h}(u,x_{1},\ssx, y,z)=g^{h}(u,x_{1},\ssx, y,z\sigma^{-1}(u, \ssx)),
\ede
then BSDE (\ref{Partial netting model Brownian BSDE for hedger}) yields
\be\label{Partial netting model Brownian BSDE 2 for hedger}
\left\{
\begin{array}
[c]{ll}
dY^{h,x_{1},\ssx}_u =Z^{h,x_{1},\ssx}_u\, d\widetilde{W}_u+\overline{g}^{h}(u,x_{1},\Sst_{u}, Y^{h,x_{1},\ssx}_u,Z^{h,x_{1},\ssx}_u)\, du ,\medskip\\
Y^{h,x_{1},\ssx}_T=H(\Sst_{T}).
\end{array}
\right.
\ee

Using the non-linear Feynman-Kac formula (see \cite{PP-1992,P-1991}), we argue that under suitable smoothness conditions imposed on the coefficients $\mu,\sigma,\kappa$ and $\beta$, the {\it hedger's pricing function} $v(t,\ssx):=Y_{t}^{h,x_{1},\ssx}$ belongs to the class $C^{1,2}([0,T]\times\mathcal{O})$ and solves the following {\it pricing PDE}
\be\label{Partial netting model hedger PDE 1}
\left\{
\begin{array}
[c]{ll}
\frac{\partial v}{\partial t}(t,\ssx)+\mathcal{L}v(t,\ssx)=\overline{g}^{h}\big(t,x_{1},\ssx,v(t,\ssx),\sigma(t,\ssx)\frac{\partial v}{\partial \ssx}\big),\quad (t,\ssx)\in[0,T]\times\mathcal{O},\medskip\\
v(T,\ssx)=H(\ssx), \quad \ssx\in\mathcal{O},
\end{array}
\right.
\ee
where the differential operator $\mathcal{L}$ is given by the following expression
\bde
\mathcal{L}:=\frac{1}{2}\sigma^{2}(t,\ssx)\frac{\partial^{2} }{\partial \ssx^{2}}+(\beta-\kappa)(t,\ssx)\frac{\partial}{\partial \ssx}.
\ede
In view of the definition of $\overline{g}^{h}$, it is clear that PDE (\ref{Partial netting model hedger PDE 1}) is in turn equivalent to
\be\label{Partial netting model hedger PDE}
\left\{
\begin{array}
[c]{ll}
\frac{\partial v}{\partial t}(t,\ssx)+\frac{1}{2}\sigma^{2}(t,\ssx)\frac{\partial^{2}v}{\partial \ssx^{2}}(t,\ssx)=
\kappa(t,\ssx) \frac{\partial v}{\partial \ssx}(t,\ssx)-x_{1}\rll_t\Blr_t\I_{\{x_{1}\ge0\}}-x_{1}\rbb_t\Bbr_t\I_{\{x_{1}\leq0\}}-r^{1,b}_{t}\big(\ssx \frac{\partial v}{\partial \ssx}(t,\ssx)\big)^{+}\medskip\\
\quad\mbox{} +\rll_t \Big(v(t,\ssx)+x_{1}\Blr_t\I_{\{x_{1}\ge0\}}+x_{1}\Bbr_t\I_{\{x_{1}\leq0\}}+\big(\ssx \frac{\partial v}{\partial \ssx}(t,\ssx)\big)^{-}\Big)^+\medskip\\
\quad\mbox{}- \rbb_t \Big(v(t,\ssx)+x_{1}\Blr_t\I_{\{x_{1}\ge0\}}+x_{1}\Bbr_t\I_{\{x_{1}\leq0\}}+\big(\ssx \frac{\partial v}{\partial \ssx}(t,\ssx)\big)^{-}\Big)^-, \quad (t,\ssx)\in[0,T]\times\mathcal{O},\medskip\\
v(T,\ssx)=H(\ssx), \quad \ssx\in\mathcal{O}.
\end{array}
\right.
\ee

\brem
It is worth stressing that the coefficient $\beta$ does not appear in the pricing PDE (\ref{Partial netting model hedger PDE}). Therefore, in order to derive the PDE, $\beta$ can be chosen arbitrarily, except for constraint ensuring that the model is arbitrage-free (see Proposition \ref{remark for non-arbitrage model}). Consequently, without changing the probability measure (i.e., by choosing $\beta$ such $a_t =0$ for all $t \in[0,T]$),  we can still derive PDE (\ref{Partial netting model hedger PDE}).
\erem

Conversely, if $v\in C^{1,2}([0,T]\times\mathcal{O})$ solves PDE (\ref{Partial netting model hedger PDE}), then the pair $(v(u,S_{u}),\sigma(u,S_{u})\frac{\partial v}{\partial \ssx}(u,S_{u}))$ solves BSDE (\ref{Partial netting model Brownian BSDE 2 for hedger}) on $u\in[t,T]$ where, for brevity, we write $S = \Sst$. From the above discussions,  $(v(u,S_{u}),\frac{\partial v}{\partial \ssx}(u,S_{u}))$ is also a solution to BSDE (\ref{Partial netting model Brownian BSDE for hedger}) on $u\in[t,T]$ for an arbitrary initial stock price $S_t=s$. Consequently, the unique replicating strategy for the hedger equals $\phi = \big(\xi, \psi^{l}, \psi^{b},\psi^{1,b}\big)$ where, for all $u\in[t,T]$,
\be\label{Partial netting model replicating strategy for the hedger}
\begin{array}
[c]{ll}
\xi_{u} =\frac{\partial v}{\partial \ssx}(u,S_{u}),\quad \psi^{1,b}_{t}=-(B^{1,b}_{u})^{-1}\big(S_u\frac{\partial v}{\partial \ssx}(u,S_{u})  \big)^{+},\medskip\\
\psi^{l}_u = (\Blr_u)^{-1} \Big(v(u,S_{u})+x_{1}\Blr_u\I_{\{x_{1}\ge0\}}+x_{1}\Bbr_u\I_{\{x_{1}\leq0\}}+\big(S_u\frac{\partial v}{\partial \ssx}(u,S_{u})\big)^{-} \Big)^+, \medskip\\
\psi^{b}_u = - (\Bbr_u)^{-1} \Big(v( u,S_{u})+x_{1}\Blr_u\I_{\{x_{1}\ge0\}}+x_{1}\Bbr_u\I_{\{x_{1}\leq0\}}+\big(S_u\frac{\partial v}{\partial \ssx}(u,S_{u})\big)^{-} \Big)^-.
\end{array}
\ee

Let us now focus on the  pricing PDE for the counterparty. Recall that
$P^{c}(x_{2},-A,-C)=\widetilde{Y}^{c,x_{2}}$, where  $(\widetilde{Y}^{c,x_{2}},\widetilde{Z}^{c,x_{2}})$
is the unique solution to the following BSDE
\be\label{Partial netting model Brownian BSDE for counterparty}
\left\{
\begin{array}
[c]{ll}
d\widetilde{Y}^{c,x_{2}}_u =\widetilde{Z}^{c,x_{2}}_u\sigma (u, S_{u})\, d\widetilde{W}_u
+g^{c}(u,x_{2}, S_{u}, \widetilde{Y}^{c,x_{2}}_u,\widetilde{Z}^{c,x_{2}}_u)\, du,\medskip\\
\widetilde{Y}^{c,x_{2}}_T=H( \Sst_{T}),
\end{array}
\right.
\ee
where for $x_{1}\ge0$,
\bde
g^{c}(u,x_{2},\ssy,y,z):=z\beta(u,\ssy)+x_{2}\rll_u\Blr_u+r^{1,b}_{u}(z\ssy)^{+}-\rll_u \Big(-y+x_{2}\Blr_u+(-z\ssy)^{-}\Big)^++\rbb_u \Big(-y+x_{2}\Blr_u+(-z\ssy)^{-}\Big)^-
\ede
and for $x_{1}\leq0$
\bde
g^{c}(u,x_{2},\ssy,y,z):=z\beta(u,\ssy)+x_{2}\rbb_u\Bbr_u+r^{1,b}_{u}(z\ssy)^{+}-\rll_u \Big(-y+x_{2}\Bbr_u+(-z\ssy)^{-}\Big)^++\rbb_u \Big(-y+x_{2}\Bbr_u+(-z\ssy)^{-}\Big)^-.
\ede
The unique replicating strategy for the counterparty equals $\phi = \big(\xi, \psi^{l}, \psi^{b},\psi^{1,b}\big)$
where $\xi_{u}= -\widetilde{Z}^{c,x_{2}}_{u}$, $\psi^{1,b}_{u}=-(B^{1,b}_{u})^{-1}(\xi_u S_u)^{+}$
and
\bde
\begin{array}
[c]{ll}
\psi^{l}_u = (\Blr_u)^{-1} \Big( -\widetilde{Y}^{h,x_{2}}_{u}+x_{2}\Blr_u\I_{\{x_{2}\ge0\}}+x_{2}\Bbr_u\I_{\{x_{2}\leq0\}}+(\xi_u S_u)^{-} \Big)^+, \medskip\\
\psi^{b}_u = - (\Bbr_u)^{-1} \Big( -\widetilde{Y}^{h,x_{2}}_{u}+x_{2}\Blr_u\I_{\{x_{2}\ge0\}}+x_{2}\Bbr_u\I_{\{x_{2}\leq0\}}+(\xi_u S_u)^{-} \Big)^-.
\end{array}
\ede

For a fixed $(t,s) \in [0,T) \times \mathcal{O}$, we denote $(Y^{c,x_{2},\ssx}_{u},Z^{c,x_{2},\ssx}_{u}):=(\widetilde{Y}^{c,x_{2},\ssx}_{u},\widetilde{Z}^{c,x_{2},\ssx}_{u}\sigma(u,\Sst_{u}))$ and
\bde
\overline{g}^{c}(u,x_{2},\ssx, y,z)=g^{c}(u,x_{2},\ssx, y,z\sigma^{-1}(u, \ssx)).
\ede
Then BSDE (\ref{Partial netting model Brownian BSDE for counterparty}) becomes
\be \label{Partial netting model Brownian BSDE 2 for counterparty}
\left\{
\begin{array}
[c]{ll}
dY^{c,x_{2},\ssx}_u =Z^{c,x_{2},\ssx}_u\, d\widetilde{W}_u +\overline{g}^{c}(u,x_{2}, \Sst_{u},Y^{c,x_{2},\ssx}_u ,Z^{c,x_{2},\ssx}_u)\, du ,\medskip\\
Y^{c,x_{2},\ssx}_T=H( \Sst_{T}).
\end{array}
\right.
\ee
Using the same argument as for the hedger, we deduce that the pricing function $v(t,\ssx):=Y_{t}^{c,x_{2},\ssx}$ belongs to $C^{1,2}([0,T]\times\mathcal{O})$ and solves the following PDE
\be \label{Partial netting model counterparty PDE 1}
\left\{
\begin{array}
[c]{ll}
\frac{\partial v}{\partial t}(t,\ssx)+\mathcal{L}v(t,\ssx)=\overline{g}^{c}\big(t,x_{2},\ssx,v(t,\ssx),\sigma(t,\ssx)\frac{\partial v}{\partial \ssx}\big),\quad (t,\ssx)\in[0,T]\times\mathcal{O},\medskip\\
v(T,\ssx)=H(\ssx), \quad \ssx\in\mathcal{O},
\end{array}
\right.
\ee
or, more explicitly,
\be\label{Partial netting model counterparty PDE}
\left\{
\begin{array}
[c]{ll}
\frac{\partial v}{\partial t}(t,\ssx)+\frac{1}{2}\sigma^{2}(t,\ssx)\frac{\partial^{2}v}{\partial \ssx^{2}}(t,\ssx)=\kappa(t,\ssx)\frac{\partial v}{\partial \ssx}(t,\ssx)+x_{2}\rll_t\Blr_t\I_{\{x_{2}\ge0\}}+x_{2}\rbb_t\Bbr_t\I_{\{x_{2}\leq0\}}+r^{1,b}_{t}\big( \ssx \frac{\partial v}{\partial \ssx}(t,\ssx)\big)^{+}\medskip\\
\quad\mbox{} -\rll_t \Big(-v(t,\ssx)+x_{2}\Blr_t\I_{\{x_{2}\ge0\}}+x_{2}\Bbr_t\I_{\{x_{2}\leq0\}}+\big(- \ssx \frac{\partial v}{\partial \ssx}(t,\ssx) \big)^{-}\Big)^+\medskip\\
\quad\mbox{}+\rbb_t \Big(-v(t,\ssx)+x_{2}\Blr_t\I_{\{x_{2}\ge0\}}+x_{2}\Bbr_t\I_{\{x_{2}\leq0\}}+\big(- \ssx \frac{\partial v}{\partial \ssx}(t,\ssx)\big)^{-}\Big)^-,\quad (t,\ssx)\in[0,T]\times\mathcal{O},\medskip\\
v(T,\ssx)=H(\ssx), \quad \ssx\in\mathcal{O}.
\end{array}
\right.
\ee

Conversely, if a function $v\in C^{1,2}([0,T]\times\mathcal{O})$ solves PDE (\ref{Partial netting model counterparty PDE}), then $(v(u,S_{u}),\sigma(u,S_{u})\frac{\partial v}{\partial \ssx}(u,S_{u}))$ solves BSDE (\ref{Partial netting model Brownian BSDE 2 for counterparty}) on $u\in[t,T]$ where we write $S = \Sst$. Consequently, the pair  $(v(u,S_{u}),\frac{\partial v}{\partial \ssx}(u,S_{u}))$ solves BSDE (\ref{Partial netting model Brownian BSDE for counterparty}). Consequently, the unique replicating strategy for the hedger equals $\phi = \big(\xi, \psi^{l}, \psi^{b},\psi^{1,b}\big)$ where, for every $u\in[t,T]$,
\be\label{Partial netting model replicating strategy for the counterparty}
\begin{array}
[c]{ll}
\xi_{u}=-\frac{\partial v}{\partial \ssx}(u,S_{u}),\quad \psi^{1,b}_{u}=-(B^{1,b}_{u})^{-1}(- S_u \frac{\partial v}{\partial \ssx}(u,S_{u}))^{+},\medskip\\
\psi^{l}_u = (\Blr_u)^{-1} \Big(-v(u,S_{u})+x_{2}\Blr_u\I_{\{x_{2}\ge0\}}+x_{2}\Bbr_u\I_{\{x_{2}\leq0\}}+ \big(-S_u\frac{\partial v}{\partial \ssx}(u,S_{u}) \big)^{-}\Big)^+, \medskip\\
\psi^{b}_u = - (\Bbr_u)^{-1} \Big(-v(u,S_{u})+x_{2}\Blr_u\I_{\{x_{2}\ge0\}}+x_{2}\Bbr_u\I_{\{x_{2}\leq0\}}+\big(-S_u\frac{\partial v}{\partial \ssx}(u,S_{u})\big)^{-}\Big)^-.
\end{array}
\ee

In summary, we are in a position to formulate the following proposition.

\bp \label{Partial netting model pricing using PDE}
Let $v(t,\ssx)\in  C^{1,2}([0,T]\times\mathcal{O})$ be the solution of quasi-linear PDE (\ref{Partial netting model hedger PDE}).
Then the hedger's ex-dividend price of the European contingent claim $H(S_T)$ is given by $v(t,S_{t})$ and the unique replicating strategy $\phi = \big(\xi, \psi^{l}, \psi^{b},\psi^{1,b}\big)$ for the hedger is given by (\ref{Partial netting model replicating strategy for the hedger}). Similarly, if $v(t,\ssx)\in C^{1,2}([0,T]\times\mathcal{O})$ is the solution of quasi-linear PDE (\ref{Partial netting model counterparty PDE}), then the counterparty's ex-dividend price of the European contingent claim $H(S_T)$ is given by $v(t,S_{t})$ and the unique replicating strategy $\phi = \big(\xi, \psi^{l}, \psi^{b},\psi^{1,b}\big)$ for the counterparty is given by (\ref{Partial netting model replicating strategy for the counterparty}).
\ep

If smoothness of model coefficients is not postulated then, from Theorem 4.3 in Peng \cite{PP-1992},  the function $v(t,\ssx):=Y_{t}^{h,x_{1},\ssx}$ (resp., $v(t,\ssx):=Y_{t}^{c,x_{2},\ssx}$) is known to be the unique viscosity solution of PDE (\ref{Partial netting model hedger PDE}) (resp., (\ref{Partial netting model counterparty PDE})).

We notice that PDE (\ref{Partial netting model hedger PDE}) depends on the initial endowment $x_{1}$. In the special case where $r^{l}=r^{b}=r$, equation (\ref{Partial netting model hedger PDE}) reduces to the following PDE independent of $x_{1}$
\be\label{classical Partial netting model hedger PDE}
\left\{
\begin{array}
[c]{ll}
\frac{\partial v}{\partial t}(t,\ssx)+\frac{1}{2}\sigma^{2}(t,\ssx)\frac{\partial^{2}v}{\partial \ssx^{2}}(t,\ssx)=\kappa(t,\ssx)\frac{\partial v}{\partial \ssx}(t,\ssx)-r^{1,b}_{t}\big(\ssx \frac{\partial v}{\partial \ssx}(t,\ssx)\big)^{+}\medskip\\
\qquad\qquad \mbox{} +r_t \Big(v(t,\ssx)+\big(\ssx \frac{\partial v}{\partial \ssx}(t,\ssx)\big)^{-}\Big), \quad (t,\ssx)\in[0,T]\times\mathcal{O},\medskip\\
v(T,\ssx)=H(\ssx), \quad \ssx\in\mathcal{O}.
\end{array}
\right.
\ee
Note that PDE (\ref{classical Partial netting model hedger PDE}) can characterize the price and the strategy for the European contingent claim in the case where the borrowing rate and the lending rates are equal.

If we assume, in addition,
that $r^{i,b}=r$, then PDE (\ref{classical Partial netting model hedger PDE}) becomes
\be\label{classical model hedger PDE}
\left\{
\begin{array}
[c]{ll}
\frac{\partial v}{\partial t}(t,\ssx)+\frac{1}{2}\sigma^{2}(t,\ssx)\frac{\partial^{2}v}{\partial \ssx^{2}}(t,\ssx)=
\kappa(t,\ssx)\frac{\partial v}{\partial \ssx}(t,\ssx)\medskip\\
\qquad\qquad \mbox{}+r_t \Big(v(t,\ssx)- \ssx \frac{\partial v}{\partial \ssx}(t,\ssx)\Big),\quad (t,\ssx)\in[0,T]\times\mathcal{O},\medskip\\
v(T,\ssx)=H(\ssx), \quad \ssx\in\mathcal{O}.
\end{array}
\right.
\ee
We observe that PDE (\ref{classical model hedger PDE}) is nothing else but the classic Black and Scholes PDE. We mentioned in Example \ref{European call option} that the market model partial netting does not cover the standard case of different borrowing and lending rates when $r^{i,b}=r^b > r^l$ and trading is assumed to be unrestricted. However, when the equalities $r^{i,b}=r^b=r^l$ are postulated, then the related PDEs for the European contingent claim are identical so, as expected, the prices and hedging strategies coincide as well.

Without using the BSDEs method, one can still obtain Proposition \ref{Partial netting model pricing using PDE} by applying the classical arguments, as was done, for instance, in \cite{B-1995}. Both methods essentially hinge on the same tool,
the non-linear Feynman-Kac formula. We mention that when the solution of the related PDE is not smooth, then
the BSDE approach gives a probabilistic representation for the viscosity solution of the PDE.

\brem
In the related paper \cite{NR4}, we also revisit the market model studied by Bergman \cite{B-1995} and we extend his analysis by considering a general contract $(A,C)$, rather than path-independent European claims, and investors with non-zero initial endowments. In this model, the funding accounts for risky assets are not introduced and thus the last constraint in (\ref{portfolio choose}) is relaxed. Hence the hedger can use his initial endowment to buy shares for the purpose of hedging.  Consequently, for each particular set-up, the properties of prices will be quite different, but most of them can be deduced from the general results for the auxiliary BSDEs. We also derive the pricing PDE for path-independent European claims in a Markovian framework.
\erem

\vskip 10 pt

%%%%%%%%%%%%%%%%%%%%%%%%%%%%%%%%%%%%%%%%%%%%%%%%%%%%%%%%%%%%%%%%%%%%%%%%%%%%%

\noindent {\bf Acknowledgement.}
The research of Tianyang Nie and Marek Rutkowski was supported under Australian Research Council's
Discovery Projects funding scheme (DP120100895).

%%%%%%%%%%%%%%%%%%%%%%%%%%%%%%%%%%%%%%%%%%%%%%%%%%%%%%%%%%%%%%%%%%%%%%%%%%%%%

%%%%%%%%%%%%%%%%%%%%%%%%%%%%%%%%%%%%%%%%%%%%%%%%%%%%%%%%%%%%%%%%%%%%%%%%%%%%%%%%%%%%%%%%%%%%%%%%%%%%
%%%%%%%%%%%%%%%%%%%%%%%%%%%%%%%%%%%%%%     REFERENCES      %%%%%%%%%%%%%%%%%%%%%%%%%%%%%%%%%%%%%%%%%
%%%%%%%%%%%%%%%%%%%%%%%%%%%%%%%%%%%%%%%%%%%%%%%%%%%%%%%%%%%%%%%%%%%%%%%%%%%%%%%%%%%%%%%%%%%%%%%%%%%%

%%%%%%%%%%%%%%%%%%%%%%%%%%%%%%%%%%%%%%%%%%%%%%%%%%%%%%%%%%%%%%%%%%%%%%%%%%%%%%%%%%
%%%%%%%%%%%%%%%%%%%%%%%%%%%%%%%%%%%%%%%%%%%%%%%%%%%%%%%%%%%%%%%%%%%%%%%%%%%%%%%%%%
\section{Appendix: Proofs of Propositions % \ref{inequality proposition for both positive initial wealth},
 \ref{inequality proposition for both negative initial wealth}
and \ref{inequality proposition for positive negative initial wealth}} \label{sect7}
\subsection{Proof of Proposition  \ref{inequality proposition for both negative initial wealth}}
%%%%%%%%%%%%%%%%%%%%%%%%%%%%%%%%%%%%%%%%%%%%%%%%%%%%%%%%%%%%%%%%%%%%%%%%%%%%%%%%%%%%%%%%%%%%%%%%

\noindent{\it Proof of  Proposition \ref{inequality proposition for both negative initial wealth}.}$\, $
We assume that $x_{1}\leq0$ and $x_{2}\leq0$. From Propositions \ref{hedger ex-dividend price} and \ref{counterparty ex-dividend price}, we know that $P^{h}(x_{1},A,C) =  \Bbr (Y^{h,b,x_{1}} - x_{1})  - C$ where $(Y^{h,b,x_{1}}, Z^{h,b,x_{1}})$ is the unique solution of BSDE (\ref{BSDE with negative x for hedger}) and $P^{c} (x_{2},-A,-C) =-(\Bbr (Y^{c,b,x_{2}} - x_{2})+C)$ where $(Y^{c,b,x_{2}}, Z^{c,b,x_{2}})$ is the unique solution of BSDE (\ref{BSDE with negative x for counterparty}).
As in the proof of Proposition \ref{inequality proposition for both positive initial wealth}, to establish \eqref{eqq2},
it suffices to check that
$\bar{Y}^{c,b,x_{2}}\leq\bar{Y}^{h,b,x_{1}}$ where $(\bar{Y}^{h,b,x_{1}},\bar{Z}^{h,b,x_{1}})$ is
the unique solution to the following BSDE
\bde % \label{transferred BSDE for hedger 1x}
\left\{
\begin{array}
[c]{l}
d\bar{Y}^{h,b,x_{1}}_t = \bar{Z}^{h,b,x_{1},\ast}_t \, d \wt S^{b,{\textrm{cld}}}_t
+\wt{f}_b \big(t, \bar{Y}^{h,b,x_{1}}_t+x_{1}, \bar{Z}^{h,b,x_{1}}_t \big)\, dt +(\Bbr_t)^{-1} \, dA^C_t, \medskip\\
\bar{Y}^{h,b,x_{1}}_T=0,
\end{array}
\right.
\ede
and $(\bar{Y}^{c,b,x_{2}}, \bar{Z}^{c,b,x_{2}})$ is the unique solution to the  BSDE
\bde % \label{transferred BSDE for counterparty 1x}
\left\{
\begin{array}
[c]{l}
d\bar{Y}^{c,b,x_{2}}_t = \bar{Z}^{c,b,x_{2},\ast}_t \, d \wt S^{b,{\textrm{cld}}}_t
-\wt{f}_b \big(t, -\bar{Y}^{c,b,x_{2}}_t+x_{2}, -\bar{Z}^{c,b,x_{2}}_t \big)\, dt +(\Bbr_t)^{-1} \, dA^C_t, \medskip\\
\bar{Y}^{c,b,x_{2}}_T=0.
\end{array}
\right.
\ede
To apply Theorem 3.3 in \cite{NR3}, we need to prove that either
\bde % \label{inequality for the borrowing drivers 1}
-\wt{f}_b \big(t, \bar{Y}^{h,b,x_{1}}_t+x_{1}, \bar{Z}^{h,b,x_{1}}_t \big)
\ge\wt{f}_b \big(t, -\bar{Y}^{h,b,x_{1}}_t+x_{2}, -\bar{Z}^{h,b,x_{1}}_t \big), \quad \PT^b\otimes \Leb-\aaee
\ede
or
\bde % \label{inequality for the borrowing drivers 2}
-\wt{f}_b \big(t, \bar{Y}^{c,b,x_{2}}_t+x_{1}, \bar{Z}^{c,b,x_{2}}_t \big)
\ge\wt{f}_b \big(t, -\bar{Y}^{c,b,x_{2}}_t+x_{2}, -\bar{Z}^{c,b,x_{2}}_t \big), \quad \PT^b\otimes \Leb-\aaee
\ede
To establish these inequalities, it suffices to use Lemma \ref{lemnew2}. We conclude that inequality \eqref{eqq2} is valid.
\endproof

\bl \label{lemnew2}
Assume that $x_{1}\leq 0$ and $x_{2} \leq 0$. Then the mapping $\wt{f}_b : \Omega \times [0,T] \times \mathbb{R}\times\mathbb{R}^{d} \to \mathbb{R}$ given by equation \eqref{drift function borrowing} satisfies
\be\label{inequality for the borrowing drivers 3}
-\wt{f}_b \big(t, y+x_{1}, z \big)
\ge \wt{f}_b \big(t, -y+x_{2}, -z \big), \ \text{ for all } (y,z)\in \mathbb{R}\times\mathbb{R}^{d}, \quad \PT^b \otimes \Leb-\aaee
\ee
\el

\proof  Recall that, for all $(y,z)\in \mathbb{R}\times\mathbb{R}^{d}$,
\bde
\wt{f}_b( t, y ,z ): = (\Bbr_t)^{-1} f_b(t,\Bbr_t y ,z ) -  \rbb_t y
\ede
where
\bde
f_b(t, y,z)  :=   \sumik_{i=1}^d \rbb_t z^i S^i_t
- \sumik_{i=1}^d \ribb_t( z^i S^i_t )^+
+   \rll_t \Big( y  + \sumik_{i=1}^d ( \zzhi )^- \Big)^+
 - \rbb_t \Big( y + \sumik_{i=1}^d ( \zzhi )^- \Big)^-  .
\ede
We now denote $\zzhi = (\Bbr_t)^{-1} z^i S^i_t$. Then
\be
\begin{array}
[c]{ll}
\delta &:=\wt{f}_b \big(t, y+x_{1}, z \big)+\wt{f}_b \big(t, -y+x_{2}, -z \big)\medskip\\
&=- \rbb_t (y+x_{1})+(\Bbr_t)^{-1} f_b(t,\Bbr_t (y+x_{1}), z)- \rbb_t (-y+x_{z})+(\Bbr_t)^{-1} f_b(t,\Bbr_t (-y+x_{2}),-z)\medskip\\
&=- \rbb_t (x_{1}+x_{2})- \sum_{i=1}^d \ribb_t|\zzhi|+\rll_t (\delta_{1}^{+}+\delta_{2}^{+})-\rbb_t (\delta_{1}^{-}+\delta_{2}^{-})  \nonumber
\end{array}
\ee
where
\be
\delta_{1}:= y+ x_{1}+\sumik_{i=1}^d (\zzhi)^{-}, \quad
\delta_{2}:=- y+ x_{2}+\sumik_{i=1}^d (-\zzhi)^{-}.\nonumber
\ee
Since $r^{l}\leq r^{b}$, we have
\bde
\begin{array}
[c]{ll}
\delta &=- \rbb_t (x_{1}+x_{2})- \sum_{i=1}^d \ribb_t|\zzhi|+\rll_t (\delta_{1}^{+}+\delta_{2}^{+})-\rbb_t (\delta_{1}^{-}+\delta_{2}^{-})  \medskip\\
&\leq - \rbb_t (x_{1}+x_{2})- \sum_{i=1}^d \ribb_t|\zzhi|+\rbb_t (\delta_{1}+\delta_{2})\medskip\\
&=- \rbb_t (x_{1}+x_{2})- \sum_{i=1}^d \ribb_t|\zzhi|+\rbb_t (x_{1}+x_{2}+\sumik_{i=1}^d |\zzhi|)\medskip\\
&=\sumik_{i=1}^d (\rbb_t-\ribb_t)|\zzhi|\leq0.
\end{array}
\ede
Thus inequality (\ref{inequality for the borrowing drivers 3}) holds.
\endproof

%%%%%%%%%%%%%%%%%%%%%%%%%%%%%%%%%%%%%%%%%%%%%%%%%%%%%%%%%%%%%%%%%%%%%%%%%%%%%%%%%%%%%%%%%%%%%%%%%%%
\subsection{Proof of Proposition \ref{inequality proposition for positive negative initial wealth}}
%%%%%%%%%%%%%%%%%%%%%%%%%%%%%%%%%%%%%%%%%%%%%%%%%%%%%%%%%%%%%%%%%%%%%%%%%%%%%%%%%%%%%%%%%%%%%%%%%%%

\noindent{\it Proof of part (i) in Proposition \ref{inequality proposition for positive negative initial wealth}.}$\, $
We first prove that if the initial endowments satisfy $x_{1}x_{2}=0$, then the inequality $P^{c}_t (x_{2},-A,-C)\leq P^{h}_t (x_{1},A,C)$ holds for any contract $(A,C)$. From Proposition \ref{general pricing proposition}, if $x_{1} \ge0,\, x_{2}\leq0$, then for any contract $(A,C)$ admissible under $\PT$ we have $P^{h}(x_{1},A,C)=\widetilde{Y}^{h,l,x_{1}}-C$ and $P^{c}(x_{2},-A,-C)=\widetilde{Y}^{c,b,x_{2}}-C$
where  $(\widetilde{Y}^{h,l,x_{1}},\widetilde{Z}^{h,l,x_{1}})$ is the unique solution of the following BSDE
\be \label{artifical BSDE for positive hedger}
\left\{
\begin{array}
[c]{ll}
d\widetilde{Y}^{h,l,x_{1}}_t =\widetilde{Z}^{h,l,x_{1},\ast}_t\, d\wt S_t^{\textrm{cld}}+g^{h,l}(t, x_{1}, \widetilde{Y}^{h,l,x_{1}}_t,\widetilde{Z}^{h,l,x_{1}}_t)\,dt+dA^C_t,\medskip\\
\widetilde{Y}^{h,l,x_{1}}_T=0,
\end{array}
\right.
\ee
and $(\widetilde{Y}^{c,b,x_{2}},\widetilde{Z}^{c,b,x_{2}})$ is the unique solution of the following BSDE
\be \label{artifical BSDE for negative counterparty}
\left\{
\begin{array}
[c]{ll}
d\widetilde{Y}^{c,b,x_{2}}_t=\widetilde{Z}^{c,b,x_{2},\ast}_t\, d\wt S_t^{\textrm{cld}}+g^{c,b}(t, x_{2}, \widetilde{Y}^{c,b,x_{2}}_t,\widetilde{Z}^{c,b,x_{2}}_t)\,dt+dA^C_t,\medskip\\
\widetilde{Y}^{c,b,x_{2}}_T=0,
\end{array}
\right.
\ee
where the generators are given by
\bde
g^{h,l}(t,x,y,z):=\sumik_{i=1}^d \beta^{i}_{t} z^iS_{t}^{i}-x\rll_t\Blr_t+g(t,y+x\Blr_t,z)
\ede
and
\bde
g^{c,b}(t,x,y,z):=\sumik_{i=1}^d \beta^{i}_{t} z^i S_{t}^{i}+x\rbb_t\Bbr_t-g(t, -y+x\Bbr_t,-z)
\ede
where in turn (see \eqref{drift driver for positive and negative initial wealth}))
\be  \label{drift driver for positive and negative initial wealthxx}
g(t,y,z)=-\sumik_{i=1}^d \ribb_t(z^{i}S^i_t )^++\rll_t \Big(y+ \sumik_{i=1}^d ( z^{i}S^i_t )^- \Big)^+
- \rbb_t \Big( y+ \sumik_{i=1}^d ( z^{i}S^i_t )^- \Big)^-.
\ee
Let us denote $\zzb = z^i S^i_t$. To apply Theorem 3.3 in \cite{NR3}, it suffices to show that
\bde
-\sumik_{i=1}^{d}\beta^{i}_{t} \zzb +x_{1}r_{t}^{l}B_{t}^{l}-g(t,y+x_{1}B_{t}^{l},z)
\ge -\sumik_{i=1}^{d}\beta^{i}_{t} \zzb -x_{2}r_{t}^{b}B_{t}^{b}+g(t,-y+x_{2}B_{t}^{b},-z),
\ede
which is equivalent to
\bde
\delta:=g(t,y+x_{1}B_{t}^{l},z)+g(t,-y+x_{2}B_{t}^{b},-z)-x_{1}r_{t}^{l}B_{t}^{l}-x_{2}r_{t}^{b}B_{t}^{b}\leq0.
\ede
In view of Lemma \ref{lemmnew3}, we conclude that the inequality $P^{c}_t (x_{2},-A,-C)\leq P^{h}_t (x_{1},A,C)$ holds
for every $t\in[0,T]$.
\endproof

\vskip 5 pt \noindent{\it Proof of part (ii) in Proposition \ref{inequality proposition for positive negative initial wealth}.}$\, $  We now assume that $r^{l}$ and $r^{b}$ are deterministic and satisfy $r^{l}_{t}<r^{b}_{t}$ for all $t\in[0,T]$. Assume that $x_{1}x_{2}\neq0$, that is, $x_{1}>0$ and $x_{2}<0$. Our goal is to find a contract $(A,C)$ and a
date $\widehat{t}\in[0,T]$  such that the inequality $P^{c}_{\widehat{t}} (x_{2},-A,-C)\leq
P^{h}_{\widehat{t}} (x_{1},A,C)$ fails to hold. To this end, we consider a contract with $C=0$ and
\bde
A_t = p \, \I_{[0,T]}(t)-\alpha\I_{[t_{0},T]}(t)+\alpha e^{\int_{t_{0}}^{T}r_{u}\,du}\I_{[T]}(t),
\ede
where $t_{0}\in(0,T)$ and the function $r$ satisfies $r_{u}\in(r^{l}_{u},r^{b}_{u})$ for all $u\in[0,T]$.
Moreover, a constant $\alpha>0$ is such that
\begin{align*}
& x_{1}\Blr_{t_0}-\alpha e^{\int_{t_{0}}^{T}(r_{u}-r^l_{u})\,du}\ge0, \quad
x_{1}+\alpha  (\Blr_{t_0})^{-1}-\alpha  (\Blr_{T})^{-1} e^{\int_{t_0}^{T}r_{u}\,du}\ge0,\\
&x_{2}\Bbr_{t_0}+\alpha e^{\int_{t_{0}}^{T}(r_{u}-r^b_{u})\,du}\leq0,\quad
x_{2}-\alpha (\Bbr_{t_0})^{-1} +\alpha  (\Bbr_{T})^{-1} e^{\int_{t_0}^{T}r_{u}\,du}\leq0,
\end{align*}
which in turn is equivalent to:  $\alpha>0$ and
\be \label{alphax}
x_{2}\kappa_{2}^{-1}\leq\alpha\leq\min\left\{ x_{1}\Blr_{t_0}e^{-\int_{t_{0}}^{T}(r_{u}-r^l_{u})\,du},\,
-x_{2}\Bbr_{t_0}e^{-\int_{t_{0}}^{T}(r_{u}-r^b_{u})\,du},\, x_{1}\kappa_{1}^{-1}\right\}
\ee
where
\bde
\kappa_{1}:=- (\Blr_{t_0})^{-1} +  (\Blr_{T})^{-1} e^{\int_{t_0}^{T}r_{u}\,du },
\quad \kappa_{2}:= (\Bbr_{t_0})^{-1}+  (\Bbr_{T})^{-1} e^{\int_{t_0}^{T}r_{u}\,du}.
\ede
Note that $\kappa_{1}>0$ and  $\kappa_{2}>0$ since from $r^{l}<r<r^{b}$, we obtain
\bde
-\int_{0}^{t_{0}}r^l_{u}\,du-\Big(\int_{t_0}^{T}r_{u}\,du-\int_{0}^{T}r^{l}_{u}\,du\Big)=-\int_{t_{0}}^{T}(r_{u}-r^l_{u})\,du<0.
\ede
and
\bde
-\int_{0}^{t_{0}}r^b_{u}\,du-\Big(\int_{t_0}^{T}r_{u}\,du-\int_{0}^{T}r^{b}_{u}\,du\Big)=-\int_{t_{0}}^{T}(r_{u}-r^b_{u})\,du>0,
\ede
Therefore, a constant $\alpha >0$ satisfying \eqref{alphax} exists and for
\bde
x := x_{1}+\alpha  (\Blr_{t_0})^{-1} - \alpha (\Blr_{T})^{-1}  e^{\int_{t_0}^{T}r_{u}\,du }\ge0,
\ede
we obtain
\bde
x \Blr_{t_0} - \cac=x_{1}\Blr_{t_0}-\alpha e^{\int_{t_{0}}^{T}(r_{u}-r^l_{u})\,du}\ge0.
\ede
Now we define the strategy $\phi = \big(\xi^1,\dots ,\xi^d, \psi^{l}, \psi^{b},\psi^{1,b},\dots ,\psi^{d,b}, \etab, \etal\big)$
where $\xi^i=\psi^{i,b}=\psi^{b}=\etab=\etal=0$ for $i=1,2,\ldots,d$ and
\bde
\psi^l_t = x \I_{[0,t_0)} + (\Blr_{t_0})^{-1}\big(x \Blr_{t_0} - \cac \big)  \I_{[t_0,T)} +  (\Blr_{T})^{-1} \Big(x \Blr_{T}
- \cac e^{\int_{t_{0}}^{T}r^l_{u}\,du} + \cac  e^{\int_{t_{0}}^{T}r_{u}\,du} \Big)\I_{[T,T]}.
\ede
The wealth process satisfies
\bde
\begin{array}
[c]{ll}
V_T (x, \varphi ,A,C) &= \big( x \Blr_{t_0} - \cac \big)  e^{\int_{t_{0}}^{T}r^l_{u}\,du} + \cac e^{\int_{t_{0}}^{T}r_{u}\,du}=  x \Blr_{T} - \cac  e^{\int_{t_{0}}^{T}r^l_{u}\,du} + \cac   e^{\int_{t_{0}}^{T}r_{u}\,du}\medskip\\
&=\Big(x_{1}+\alpha  (\Blr_{t_0})^{-1} -\alpha  (\Blr_{T})^{-1} e^{\int_{t_0}^{T}r_{u}\,du}\Big) \Blr_{T}
- \cac  e^{\int_{t_{0}}^{T}r^l_{u}\,du} + \cac   e^{\int_{t_{0}}^{T}r_{u}\,du}\medskip\\
&=x_{1} \Blr_{T} =V_{T}^{0}(x_{1}).
\end{array}
\ede
Therefore, the self-financing strategy $(x, \phi , A, \pC )$ replicates the contract $(A,C)$ on $[0,T]$. Moreover from the uniqueness of the related pricing BSDE, we know that this is the unique strategy.  From Definition \ref{definition of ex-dividend price}, it follows that
\bde
P^{h}_{0} (x_{1},A,C)=x-x_{1}=\alpha  (\Blr_{t_0})^{-1} -\alpha  (\Blr_{T})^{-1} e^{\int_{t_0}^{T}r_{u}\,du}=-\alpha\kappa_{1}
<0.
\ede
Let us now focus on the counterparty.  If we set
\bde
\widetilde{x}=x_{2}-\alpha (\Bbr_{t_{0}})^{-1}+\alpha (\Bbr_{T})^{-1} e^{\int_{t_0}^{T}r_{u}\,du}\leq0,
\ede
then we obtain
\bde
\widetilde{x} \Bbr_{t_0}+\cac=x_{2}\Bbr_{t_0}+\alpha e^{\int_{t_{0}}^{T}(r_{u}-r^b_{u})\,du}\leq0.
\ede
We define the strategy $\tilde{\phi} = \big(\tilde{\xi}^1,\dots ,\tilde{\xi}^d, \tilde{\psi}^{l}, \tilde{\psi}^{b},\tilde{\psi}^{1,b},\dots ,\tilde{\psi}^{d,b}, \tilde{\eta}^{b}, \tilde{\eta}^{l}\big)$
where $\tilde{\xi}^i=\tilde{\psi}^{i,b}=\tilde{\psi}^{l}=\tilde{\eta}^b=\tilde{\eta}^l=0$ for $i=1,2,\ldots,d$ and
\bde
\tilde{\psi}^b_t = \tilde{x} \I_{[0,t_0)} + (\Bbr_{t_0})^{-1}\big(\tilde{x} \Bbr_{t_0} +\cac \big)  \I_{[t_0,T)} +  (\Bbr_{T})^{-1} \Big( \tilde{x } \Bbr_{T}+ \cac e^{\int_{t_{0}}^{T}r^b_{u}\,du} - \cac  e^{\int_{t_{0}}^{T}r_{u}\,du} \Big)\I_{[T,T]}.
\ede
Then we have
\bde
\begin{array}
[c]{ll}
V_T (\tilde{x}, \tilde{\varphi} , -A,-C) &= \big( \tilde{x} \Bbr_{t_0}+\cac \big)  e^{\int_{t_{0}}^{T}r^b_{u}\,du} - \cac e^{\int_{t_{0}}^{T}r_{u}\,du} =  \tilde{x} \Bbr_{T} +\cac  e^{\int_{t_{0}}^{T}r^b_{u}\,du} -\cac   e^{\int_{t_{0}}^{T}r_{u}\,du}\medskip\\
&=\Big(x_{2}-\alpha  (\Bbr_{t_0})^{-1} +\alpha  (\Bbr_{T})^{-1} e^{\int_{t_0}^{T}r_{u}\,du }\Big) \Bbr_{T}
+ \cac  e^{\int_{t_{0}}^{T}r^b_{u}\,du}-\cac   e^{\int_{t_{0}}^{T}r_{u}\,du}\medskip\\
&=x_{2} \Bbr_{T}=V_{T}^{0}(x_{2}).
\end{array}
\ede
Hence $(\tilde{x},\tilde{ \phi} , -A, -\pC )$ is the unique self-financing strategy replicating the contract $(-A,-C)$ on $[0,T]$. From Definition \ref{remark for counterparty's ex-dividend price}, it follows that
\bde
P^{c}_{0} (x_{2},-A,-C)=x_{2}-\tilde{x}=\alpha  (\Bbr_{t_0})^{-1} -\alpha  (\Bbr_{T})^{-1} e^{\int_{t_0}^{T}r_{u}\,du }=\alpha\kappa_{2}>0.
\ede
We have thus found a date $\widehat{t}=0$ and a contract $(A,C)$ such that
\bde
P^{c}_{0} (x_{2},-A,-C)> \alpha\kappa_{2} >  0 > -\alpha\kappa_{1} = P^{h}_{0} (x_{1},A,C),
\ede
so that the range of bilaterally profitable prices ${\cal R}^p_0(x_1,x_2)$ is non-empty.
Note that here both parties are willing to pay a strictly positive
amount to the other party for the right to enter the contract. This completes the proof of part (ii).
\endproof

\bl \label{lemmnew3}
Assume that $x_1 x_2 =0$. If the function $g$ is given by (\ref{drift driver for positive and negative initial wealthxx}),
 then following inequality holds
\bde
\delta:=g(t,y+x_{1}B_{t}^{l},z)+g(t,-y+x_{2}B_{t}^{b},-z)-x_{1}r_{t}^{l}B_{t}^{l}-x_{2}r_{t}^{b}B_{t}^{b} \leq 0.
\ede
\el

\proof
In view of (\ref{drift driver for positive and negative initial wealthxx}), we have
\bde
\delta =-x_{1}r_{t}^{l}B_{t}^{l}-x_{2}r_{t}^{b}B_{t}^{b}-\sumik_{i=1}^d \ribb_t|\zzb|+\rll_t (\delta_{1}^{+}+\delta_{2}^{+})-\rbb_t (\delta_{1}^{-}+\delta_{2}^{-}),
\ede
where
\bde
\delta_{1}=y+x_{1}B_{t}^{l}+\sumik_{i=1}^{d}(\zzb )^{-},\quad
\delta_{2}=-y+x_{2}B_{t}^{b}+\sumik_{i=1}^{d}(- \zzb )^{-}.
\ede
We claim that
\bde
\delta\leq\min \left\{(\rll_t- \rbb_t )x_{2}B_{t}^{b}
+\sumik_{i=1}^d (\rll_t-\ribb_t)|\zzb|,\, (\rbb_t- \rll_t)x_{1}B_{t}^{l}+\sumik_{i=1}^d (\rbb_t-\ribb_t)|\zzb|\right\}.
\ede
Indeed, from $r^{l}\leq r^b$, we have
\bde
\begin{array}[c]{rl}
& \rll_t (\delta_{1}^{+}+\delta_{2}^{+})-\rbb_t (\delta_{1}^{-}+\delta_{2}^{-})\medskip\\
\leq & \min\{\rll_t (\delta_{1}+\delta_{2}),\, \rbb_t (\delta_{1}^{-}+\delta_{2}^{-})\}\medskip\\
=& \min\left\{\rll_t (x_{1}B_{t}^{l}+x_{2}B_{t}^{b}+\sumik_{i=1}^{d}|\zzb |),\, \rbb_t (x_{1}B_{t}^{l}+x_{2}B_{t}^{b}+\sumik_{i=1}^{d}|\zzb |)\right\}.
\end{array}
\ede
Thus
\bde
\begin{array}[c]{rl}
\delta &=-x_{1}r_{t}^{l}B_{t}^{l}-x_{2}r_{t}^{b}B_{t}^{b}-\sumik_{i=1}^d \ribb_t|\zzb|+\rll_t (\delta_{1}^{+}+\delta_{2}^{+})-\rbb_t (\delta_{1}^{-}+\delta_{2}^{-})\medskip\\
&\leq -x_{1}r_{t}^{l}B_{t}^{l}-x_{2}r_{t}^{b}B_{t}^{b}-\sumik_{i=1}^d \ribb_t|\zzb|\medskip\\
&\qquad+\min\left\{\rll_t (x_{1}B_{t}^{l}+x_{2}B_{t}^{b}+\sumik_{i=1}^{d}|\zzb |),\, \rbb_t (x_{1}B_{t}^{l}+x_{2}B_{t}^{b}+\sumik_{i=1}^{d}|\zzb |)\right\}\medskip\\
&= \min \left\{(\rll_t- \rbb_t )x_{2}B_{t}^{b}
+\sumik_{i=1}^d (\rll_t-\ribb_t)|\zzb|,\, (\rbb_t- \rll_t)x_{1}B_{t}^{l}+\sumik_{i=1}^d (\rbb_t-\ribb_t)|\zzb|\right\}.
\end{array}
\ede
If $x_{1}x_{2}=0$, then the right-hand side in the above inequality is non-positive.
We conclude that $\delta\leq0$. \endproof

%%%%%%%%%%%%%%%%%%%%%%%%%%%%%%%%%%%%%%%%%%%%%%%%%%%%%%%%%%%%%%%%%%%%%%%%%%%%%%%%%%%%%%%%%%%%%%%%%%%%
\end{document}